\documentclass{article}
\usepackage{etoolbox}
\newtoggle{draft}
\toggletrue{draft}
\usepackage[utf8]{inputenc}
\usepackage{amssymb,amsfonts}
\usepackage{amsthm}
\usepackage{mathtools}
\usepackage{hyperref}
\hypersetup{
  colorlinks=true,
  linkcolor=green!30!black,
  linktoc=all,
  citecolor=blue,
  bookmarksnumbered=true,
  pdfproducer={},
  pdfauthor={},
  pdfkeywords={},
  pdfsubject={},
  pdfcreator={},
  pdflang={}
}
\usepackage{cleveref}

\newtheorem{theorem}{Theorem}[section]
\newtheorem{observation}[theorem]{Observation}
\newtheorem{claim}[theorem]{Claim}

\newtheorem{lemma}[theorem]{Lemma}

\newtheorem{corollary}[theorem]{Corollary}

\newtheorem{problem}{Problem}
\theoremstyle{definition}
\newtheorem{definition}[theorem]{Definition}

\crefname{theorem}{Theorem}{Theorems}
\crefname{observation}{Observation}{Observations}
\crefname{claim}{Claim}{Claims}
\crefname{condition}{Condition}{Conditions}
\crefname{example}{Example}{Examples}
\crefname{fact}{Fact}{Facts}
\crefname{lemma}{Lemma}{Lemmas}
\crefname{corollary}{Corollary}{Corollaries}
\crefname{definition}{Definition}{Definitions}
\crefname{remark}{Remark}{Remarks}
\crefname{proposition}{Proposition}{Propositions}
\crefname{property}{Property}{Properties}
\crefname{problem}{Problem}{Problems}
\crefname{section}{Section}{Sections}

\usepackage{mathrsfs}
\usepackage{amsmath}
\usepackage{tikz}
\usetikzlibrary{shapes,arrows,positioning}
\usetikzlibrary{er,positioning,bayesnet}
\usepackage{fullpage}
\usepackage[tracking]{microtype}
\usepackage{float}

\usepackage[linesnumbered,ruled]{algorithm2e}
\crefname{algocf}{alg.}{algs.}

\usepackage{natbib}

\title{Faster Fixed Parameter Tractable Algorithms for Counting Markov Equivalence Classes with Special Skeletons}
\author{}
 \author{Vidya Sagar Sharma\thanks{Tata Institute of Fundamental Research,
     Mumbai. Email: \texttt{vidyasagartifr@gmail.com}.} 
}
\date{}

\newcommand{\cp}[0]{chordless path}
\newcommand{\cps}[0]{chordless paths}

\newcommand{\cc}[0]{chordless cycle}
\newcommand{\ucc}[0]{undirected connected component}
\newcommand{\uccs}[0]{undirected connected components}

\newcommand{\skel}[1]{\textup{skeleton}({#1})}

\newcommand{\setofpartialMECs}[1]{\text{PMEC}(#1)}

\newcommand{\setofMECs}[1]{\text{MEC}(#1)}
\newcommand{\tfp}[0]{\text{triangle-free path}}
\newcommand{\tfps}[0]{\text{triangle-free paths}}

\newcounter{casenum}

\makeatletter
\newcommand*{\defeq}{\mathrel{\rlap{\raisebox{0.3ex}{$\m@th\cdot$}}\raisebox{-0.3ex}{$\m@th\cdot$}}=}
\newcommand*{\eqdef}{=
  \mathrel{\rlap{\raisebox{0.3ex}{$\m@th\cdot$}}\raisebox{-0.3ex}{$\m@th\cdot$}}}
\makeatother

\renewcommand{\emptyset}[0]{\varnothing}

\iftoggle{draft}{\usepackage{marginnote}
      }{}

\begin{document}

\maketitle
\begin{abstract}
  The structure of Markov equivalence classes (MECs) of causal DAGs has been studied extensively. A natural question in this regard is to algorithmically find the number of MECs with a given skeleton. Until recently, the known results for this problem were in the setting of very special graphs (such as paths, cycles and star graphs). More recently, a fixed-parameter tractable (FPT) algorithm was given for this problem which, given an input graph $G$, counts the number of MECs with the skeleton $G$ in $O(n(2^{O(d^4k^4)} + n^2))$ time, where $n$, $d$, and $k$, respectively, are the numbers of nodes, the treewidth, and the degree of $G$.

We give a faster FPT algorithm that solves the problem in $O(n(2^{O(d^2k^2)} + n^2))$ time when the input graph is chordal. Additionally, we show that the runtime can be further improved to polynomial time when the input graph $G$ is a tree.
\end{abstract} 
\section{Introduction}
 A graphical model is a graphical way to represent the conditional dependence between the random variables. Both the directed and the undirected versions of the graphical models are well-studied in the literature. The directed graphical models are of two kinds: cyclic and acyclic. We study the directed acyclic graphical model, also known as the Bayesian network. 

 In the Bayesian network, a directed acyclic graph (DAG) with vertex set $V$ represents a probability distribution on the set of random variables $V$ if for each node $w$ in the DAG, the probability distribution satisfies the conditional independence relation that $w$ is independent of its non-descendants given its parent. A DAG entails a conditional independence relation if all the probability distribution that the DAG represents satisfies the conditional independence relation. There can be more than one DAG that entails the same set of conditional independence relations. Such DAGs are said to be Markov equivalent. \cite{verma1990equivalence} show the graphical resemblance between two Markov equivalent DAGs. They show that two DAGs are Markov equivalent if, and only if, both have the same skeleton (underlying undirected graph), and both have the same set of v-structures (induced subgraphs of the form $a\rightarrow b \leftarrow c$). Markov equivalence classes (MECs) are used to partition the DAGs. Two Markov equivalent DAGs belong to the same MEC. In other words, an MEC is a maximal size set of Markov equivalent DAGs.

 An MEC entails the same set of conditional independence relations between the random variables that are entailed by the DAGs it contains. An MEC is graphically represented by the union of DAGs it contains. We treat an MEC and its graphical representation as the same. \cite{andersson1997characterization} give necessary and sufficient conditions for a graph to be an MEC (\Cref{thm:nes-and-suf-cond-for-chordal-graph-to-be-an-MEC}).

\cite{meek1995causal} gives a method to construct an MEC given the set of conditional independence relations entailed by the MEC. He shows that two random variables $A$ and $B$ are not adjacent in the MEC if, and only if, the MEC entails a conditional independence relation of the form $A \perp B \mid S$ for some $S \subseteq V \setminus{\{A, B \}}$, where $V$ is the set of nodes of the MEC (\cite[p.~3]{meek1995causal}).
Suppose $M_1$ and $M_2$ are two MECs with the same skeleton and entail the set of conditional independence relations, respectively, $\mathcal{M}_1$ and $\mathcal{M}_2$, on the set of random variables $V$.
\cite{meek1995causal}'s result shows that for two random variables $A$ and $B$, $\mathcal{M}_1$ contains $A \perp B \mid S_1$ for some $S_1 \subseteq V\setminus{\{A, B\}}$ if, and only if, $\mathcal{M}_2$ contains $A \perp B \mid S_2$ for some $S_2 \subseteq V\setminus{\{A, B\}}$ ($S_1$ and $S_2$ may not be the same). 

This shows that the two MECs with the same skeleton have not only a graphical relation but also a statistical one. This connection between MECs with the same skeleton gives us the motivation to solve the problem of counting MECs having the same skeleton.  In particular, a natural question to ask is: given an undirected graph $G$, how many MECs have the same skeleton as $G$?

\paragraph{Related Work} Counting MECs is a well-studied problem. \cite{perlman2001graphical} written a computer program that for input $n$, outputs the number of MECs with $n$ nodes. \cite{gillispie2002size} written a computer program that for input $n$, enumerates the MECs. They also studied the distribution of the size of MECs. They observed that the ratio of DAGs to the number of MECs asymptotically converges to 3.7. \cite{steinsky2003enumeration} gives a recursive formula that counts the number of MECs of size 1. \cite{gillispie2006formulas} come up with a recursive formula that counts MECs of any size. \cite{schmid2022number} show that as $n$ tends to infinity the ratio of the number of DAGs with $n$ nodes and the number of MECs with $n$ nodes approaches a positive constant.

The above results were on the classes of MECs with a given number of nodes. We now focus on the results of the classes of MECs with a fixed skeleton. \cite{verma1990equivalence}'s result shows that each MEC has a unique skeleton and a set of v-structures. \cite{radhakrishnan2016counting} work on counting MECs with the same skeleton. They partitioned the set of MECs based on the number of v-structures in them and constructed a generating function for counting MECs. They show experimentally that even for the graphs with the same number of nodes, the generating function varies. Their further work on the problem of counting MECs with the same skeleton explored generating functions for specific graph structures such as path graphs, cycle graphs, star graphs, and bi-star graphs (\cite{radhakrishnan2018counting}). Their result also provides a tight lower and upper bound for the number of MECs with a tree skeleton. Recently, \cite{sharma2023fixedparameter} give a fixed parameter tractable algorithm that for an input graph $G$, counts the number of MECs with skeleton $G$ in time $O(n(2^{O(k^4d^4)}+n^2))$, where $n$, $k$ and $d$ are, respectively, the number nodes of $G$, the treewidth of $G$, and the degree of $G$.

\paragraph{Our Contribution and Organization of the Paper}
We provide a faster-fixed parameter tractable algorithm (\Cref{alg:counting-MEC-chordal}) that for a chordal graph $G$, counts MECs with skeleton $G$ in time $O(n(2^{O(k^2d^2)}+n^2))$, where $n$, $k$ and $d$ are, respectively, the number nodes of $G$, the treewidth of $G$, and the degree of $G$. We also provide a polynomial algorithm (\Cref{alg:counting-MEC-of-tree}) that for a tree graph $G$, counts MECs with skeleton $G$ in time $O(n^2d)$, where $n$ and $d$ are, respectively, the number of nodes of $G$, and the degree of $G$.  Our result is a significant improvement over the present best known solution given by \cite{sharma2023fixedparameter} for the chordal and tree graphs. The current best result for counting MECs with a tree skeleton is by \cite{sharma2023fixedparameter} with the time complexity exponential in the degree of the tree graph, whereas our algorithm for counting MECs with a tree skeleton is polynomial in the size of the input graph.

We define the terminologies used in the paper in \cref{sec:preliminary}. In \cref{sec:algorithm-for-counting-MECs}, we formally define the problem of counting MECs with the same skeleton (\cref{prob:counting-MEC}). In \cref{subsection:tree}, we solve \cref{prob:counting-MEC} when the skeleton is a tree graph. In \cref{subsection:chordal-graph}, we solve \cref{prob:counting-MEC} when the skeleton is a chordal graph. \cref{sec:time-complexity} deals with the time complexity of the algorithms studied in the paper. All the missing proofs of the main section are provided in the Supplementary Material. 
\section{Preliminary}
\label{sec:preliminary}
\paragraph{Graph Terminologies} We follow the graph theory terminologies used by \cite{andersson1997characterization} and \cite{sharma2023fixedparameter}. A graph $G = (V, E)$ is a tuple, where $V$ is the set of nodes, and $E \subseteq V \times V$ is the set of edges of $G$. For a graph $G$, we denote the set of nodes of $G$ by $V_G$ and the set of edges of $G$ by $E_G$. An edge $(u,v)$ is said to be an \emph{undirected edge} of $G$, denoted as $u-v \in E_G$, if both $(u,v)$ and $(v,u)$ are in $E_G$. An edge $(u,v)$ is said to be a \emph{directed edge} of $G$, denoted as $u\rightarrow v \in E_G$, if $(u,v) \in E_G$ and $(v,u) \notin E_G$.  A graph is said to be an \emph{undirected graph} if each edge of the graph is undirected. A graph is said to be a \emph{directed graph} if each edge of the graph is directed. A sequence of distinct nodes $(u_1, u_2, \ldots, u_l)$ is said to be a \emph{path} of a graph $G$ if for each $1\leq i < l$, $(u_i, u_{i+1}) \in E_G$.
$P$ is said to be a path from $(u,v)$ to $(x,y)$ if the first two nodes of $P$ are $u$ and $v$ and the last two nodes are $x$ and $y$. Similarly, $P$ is said to be a path from $(u,v)$ to $w$ if the first two nodes of the path are $u$ and $v$ and the last node of the path is $w$.
A sequence $(u_1, u_2, \ldots, u_l, u_{l+1} = u_1)$ is said to be a \emph{cycle} if $u_1, u_2, \ldots, u_l$ are distinct and for each $1\leq i \leq l$, $(u_i, u_{i+1}) \in E_G$. A cycle of a graph is said to be a \emph{directed cycle} of the graph if one edge of the cycle is a directed edge in the graph. A cycle $(u_1, u_2, \ldots, u_l, u_{l+1} = u_1)$ of a graph is said to be a \emph{chordless cycle} of the graph if there does not exist any edge between any non-adjacent nodes of the cycle. A graph that does not have any directed cycle is said to be a \emph{chain graph}. A graph is said to be a \emph{chordal graph} if each cycle of the graph is chordless. A \emph{skeleton} of a graph is the underlying undirected graph of it, i.e., we get the skeleton of a graph by replacing each directed edge $u\rightarrow v$ of the graph with an undirected edge $u-v$.
$\skel{G}$ denotes the the skeleton of $G$.
A \emph{v-structure} of a graph is an induced subgraph of the form $a\rightarrow b \leftarrow c$. For a graph $G$, $\mathcal{V}(G)$ is denoted as the set of v-structures of $G$. $N(X, G)$ denotes the set of nodes of $G$ that are a neighbor of a node in $X$. A clique of an undirected graph $G$ is a set of nodes $C$ of $G$ such that for each pair of nodes $u,v \in C$, $(u,v) \in E_G$.

\begin{definition}[\textbf{Union of graphs}]
\label{def:union-of-graphs}
Let $G_1$ and $G_2$ be two graphs. $G$ is said to be the \emph{union} of $G_1$ and $G_2$ if $V_G = V_{G_1}\cup V_{G_2}$ and $E_G = E_{G_1} \cup E_{G_2}$.
\end{definition}

\paragraph{Markov Equivalence Class (MEC)} We have defined MEC in the introduction. \cite{andersson1997characterization} give the following necessary and sufficient conditions for a graph to be an MEC (\Cref{thm:nes-and-suf-cond-for-chordal-graph-to-be-an-MEC}).
  \begin{figure}[ht]
    \begin{center}
    \begin{tabular}{ c c c c }
 (a): &
 \begin{tikzpicture}
     \label{fig-a:strongly-protected-edge}
    \node[](u){$u$};
    \node[](v)[right= 1.5 of u]{$v$};
    \node[](w)[below right = 0.8 and 0.6 of u]{$w$};
    \draw[->](u)--(v);
    \draw[->](w)--(u);
    \end{tikzpicture}
    & (b): &
    \begin{tikzpicture}
    \node[](u){$u$};
    \node[](v)[right= 1.5 of u]{$v$};
    \node[](w)[below right = 0.8 and 0.6 of u]{$w$};
    \draw[->](u)--(v);
    \draw[->](w)--(v);
     \label{fig-b:strongly-protected-edge}
    \end{tikzpicture}
    \\ 
 (c): &
 \begin{tikzpicture}
    \node[](u){$u$};
    \node[](v)[right= 1.5 of u]{$v$};
    \node[](w)[below right = 0.8 and 0.6 of u]{$w$};
    \draw[->](u)--(v);
    \draw[->](u)--(w);
    \draw[->](w)--(v);
     \label{fig-c:strongly-protected-edge}
    \end{tikzpicture}
    & (d): &
    \begin{tikzpicture}
    \node[](u){$u$};
    \node[](v)[right= 1.5 of u]{$v$};
\node[](b)[below right=0.5 and 0.5 of u]{$w'$};
    \node[](a)[above right=0.5 and 0.5 of u]{$w$};
    \draw[->](u)--(v);
    \draw[-](u)--(b);
    \draw[->](b)--(v);
    \draw[-](a)--(u);
    \draw[->](a)--(v);
    \label{fig-d:strongly-protected-edge}
    \end{tikzpicture}
    \\     
\end{tabular}
\end{center}
\caption{Strongly protected $u \rightarrow v$.}
\label{fig:strongly-protected-edge}
\end{figure}
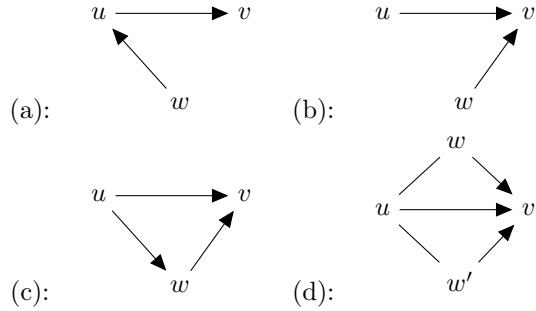  \begin{theorem}[\cite{andersson1997characterization}]
\label{thm:nes-and-suf-cond-for-chordal-graph-to-be-an-MEC}
A graph $G$ is an MEC if, and only if, 
\begin{enumerate}
    \item
    \label{item-1-theorem-nec-suf-cond-for-MEC}
    $G$ is a chain graph.
    \item
    \label{item-2-theorem-nec-suf-cond-for-MEC}
    For every chain component $\tau$ of $G$, $G_{\tau}$ is chordal, i.e., every \ucc{} of $G$ is chordal.
    \item
    \label{item-3-theorem-nec-suf-cond-for-MEC}
    The configuration $a\rightarrow b - c$ does not occur as an induced subgraph of $G$.
    \item
    \label{item-4-theorem-nec-suf-cond-for-MEC}
    Every directed edge $u\rightarrow v \in G$ is strongly protected in $G$, i.e., $u\rightarrow v$ is a part of at least one of the induced subgraphs of $G$ of the form as shown in \Cref{fig:strongly-protected-edge}.
\end{enumerate}
\end{theorem}
For an MEC $M$ with skeleton $G$, we say $M$ is an MEC of $G$. We denote the set of MECs of $G$ with $\setofMECs{G}$.

The following definitions and results come from \cite{sharma2023fixedparameter}.

\begin{definition}[\emph{Partial MEC}, \cite{sharma2023fixedparameter}]
\label{def:partial-MEC}
    A graph $G$ is said to be a partial MEC if it obeys \cref{item-1-theorem-nec-suf-cond-for-MEC,item-2-theorem-nec-suf-cond-for-MEC,item-3-theorem-nec-suf-cond-for-MEC} of \cref{thm:nes-and-suf-cond-for-chordal-graph-to-be-an-MEC}.
    A partial MEC $O$ is said to be a partial MEC of $G$, if $\skel{O} = G$.
    $\setofpartialMECs{G}$ denotes the set of partial MECs with skeleton $G$.
\end{definition}

\begin{observation}[\cite{sharma2023fixedparameter}]
\label{obs:induced-subgraph-of-MEC-is-PMEC}
    For an MEC $G$, its induced subgraphs are partial MECs.
\end{observation}

\begin{definition}[\emph{Projection}, \cite{sharma2023fixedparameter}]
\label{def:projection}
    Suppose $G$ is an undirected graph, $G'$ is an induced subgraph of $G$, and $M$ and $M'$ are MECs of, respectively, $G$ and $G'$. We say $M'$ is a projection of $M$ on $Y = V_{G'}$ if the set of v-structures of $M'$ and the set of v-structures of $M[Y]$ are the same, i.e.,  $\mathcal{V}(M') =  \mathcal{V}(M[Y])$.  We denote the projection of $M$ on $Y$ with $\mathcal{P}(M, Y)$. With a slight abuse of notation, we define $\mathcal{P}(M, Y_1, Y_2) =  (\mathcal{P}(M, Y_1), \mathcal{P}(M, Y_2))$.
\end{definition}

\begin{observation}[\cite{sharma2023fixedparameter}]
    \label{obs:there-exists-a-unique-projection}
    Let $M$ be an MEC of $G$. For all $Y \subseteq V_G$, there exists a unique projection of $M$ on $Y$, i.e., $\mathcal{P}(M, Y)$ is unique.
\end{observation}

\begin{lemma}[\cite{sharma2023fixedparameter}]
\label{lem:directed-edge-in-projection-implies-directed-edge-in-MEC}
\label{lem:directed-edge-is-same-in-projected-MEC}
    Let $M$ be a MEC, and $M'$ is a projection of $M$. If $u\rightarrow v  \in M'$ then $u\rightarrow v \in M$.
\end{lemma}

\begin{definition}[\emph{Clique Tree Representation}, \cite{blair1993introduction}]
\label{def:clique-tree}
For a chordal graph, a clique tree representation of a chordal graph $G$ is a tree graph $T$ such that (a) $V_T = \Pi(G)$, a set of maximal cliques of $G$, 
and (b) $T$ obeys the \emph{clique intersection property}: for $x, y \in V_T$, if $z$ is a node along the path between $x$ and $y$ in $T$ then $x\cap y \subseteq z$. 
\end{definition}
\begin{lemma}[\cite{blair1993introduction}]
    \label{lem:every-chordal-graph-has-clique-tree}
    For every chordal graph $G$, a clique tree representation of $G$ exists.
\end{lemma}

\begin{definition}[\emph{LBFS odering}, \cite{rose1976algorithmic}]
    \label{def:LBFS}
    An LBFS ordering $\tau$ of a chordal graph $G$ is an ordering of the nodes of $G$ with the property that for any node $u$ of the graph, the neighbors of $u$ that comes before $u$ in $\tau$ form a clique, i.e., if $x-u, y-u \in E_G$, $\tau(x) < \tau(u)$, and $\tau(y) < \tau(u)$ then $x-y \in E_G$.
\end{definition}
\begin{lemma}[\cite{rose1976algorithmic}]
For a chordal $G$, for any clique $C$, there exists an LBFS ordering that starts with $C$.
\end{lemma}

\begin{definition}[\textbf{Synchronous Graphs}, \cite{sharma2023fixedparameter}]
\label{def:synchronous-graphs} 
Two graphs $G$ and $H$ are said to be \emph{synchronous graphs} if there do not exist vertices $x, y \in V_{G}\cap V_{H}$ such that $x\rightarrow y \in E_G$ and $y\rightarrow x \in E_H$.
\end{definition}

\begin{definition}[\textbf{Markov Union of Synchronous Graphs}, \cite{sharma2023fixedparameter}]
\label{def:Markov-union-of-graphs} 
    Let $G_1, G_2, \ldots, G_l$ be pairwise synchronous graphs. The \emph{Markov union} of $G_1, G_2, \ldots, G_l$ is a graph $G$, denoted by $G = U_M(G_1, G_2, \dots, G_l)$, such that $\skel{G}$ is the skeleton of the union of $G_1, G_2, \ldots, G_l$ (see \cref{def:union-of-graphs} for the definition of the graph union), and for any edge $u-v \in \skel{G}$, $u\rightarrow v \in G$ if $u\rightarrow v \in G_i$ for any $G_i$. More formally, $V_G = \bigcup_{i=1}^{l} V_{G_i}$, the set of directed edges of $G$ is $D_G = \{u\rightarrow v: u\rightarrow v \in E_{G_i}$ for some $1 \leq i \leq l\}$, and the set of undirected edges of $G$ is $U_G = \{u-v: u-v \in E_{G_i}$ for some $1 \leq i \leq l$, and there does not exist a $j$ such that $u\rightarrow v \in E_{G_j}$ or $v\rightarrow u \in E_{G_j}\}$. In other words, if $u\rightarrow v$ is a directed edge in any graph $G_i$, then $u\rightarrow v$ is also a directed edge in the Markov union graph, and an undirected edge $u-v$ of $G_i$ becomes an undirected edge of $G$ only if neither $u\rightarrow v$ nor $v\rightarrow u$ is part of any graph in the union. Since the graphs are pairwise synchronous, if one graph $G_i$ contains a directed edge $u\rightarrow v$, then no other graph $G_j$ can contain the directed edge $v\rightarrow u$.
\end{definition}

All proofs omitted from the main text are provided in the Supplementary Material.
 
\section{Counting MECs with a tree or a chordal graph skeleton}
We start with a formal definition of the problem.

\begin{problem}[Counting MECs of an undirected graph]
\label{prob:counting-MEC}
\textbf{Input:} An undirected graph $G$. 
\textbf{Output:} The number of MECs of $G$, i.e., $|\setofMECs{G}|$.
\end{problem}

We solve the problem when $G$ is a tree graph or a chordal graph. 

\subsection{Counting MECs of a tree graph}
\label{subsection:tree}
In this subsection, we solve \cref{prob:counting-MEC} when $G$ is a tree graph.  We provide an algorithm that for a tree graph $G$, counts the MECs of $G$.
We start with the tree version of \cref{thm:nes-and-suf-cond-for-chordal-graph-to-be-an-MEC}.
\begin{theorem}
    \label{thm:nes-and-suf-cond-for-tree-graph-to-be-an-MEC}
    A graph $G$ with a tree skeleton is an MEC if, and only if,
    \begin{enumerate}
        \item
        \label{item-1-of-thm:nes-and-suf-cond-for-tree-graph-to-be-an-MEC}
        $G$ does not have an induced subgraph of the form $a\rightarrow b -c$.
        \item 
        \label{item-2-of-thm:nes-and-suf-cond-for-tree-graph-to-be-an-MEC}
        For each directed edge $u\rightarrow v$, there exists a node $w$ such that $u\rightarrow v$ is part of an induced subgraph of $G$ of the form either (a) $w\rightarrow u \rightarrow v$, or (b) $u\rightarrow v \leftarrow w$.
    \end{enumerate}
\end{theorem}
\begin{proof}
    A graph with a tree skeleton cannot have a cycle. Therefore, it neither has any directed cycle nor any chordless cycle. Consequently, the graph must be a chain graph. Additionally, each \ucc{} of the graph must be chordal. This implies that for a graph with a tree skeleton, \cref{item-1-theorem-nec-suf-cond-for-MEC,item-2-theorem-nec-suf-cond-for-MEC} of \cref{thm:nes-and-suf-cond-for-chordal-graph-to-be-an-MEC} is always true.
    
    Also, in a graph with a tree skeleton, there cannot be any induced subgraph of the form as shown in \cref{fig:strongly-protected-edge}.c and \cref{fig:strongly-protected-edge}.d. Therefore, if an edge of the graph is strongly protected then it must be due to part of an induced subgraph of the form  either \cref{fig:strongly-protected-edge}.a, or \cref{fig:strongly-protected-edge}.b.
\end{proof}

To count MECs of a tree graph $G$, we arbitrarily pick a node $r_1$ of $G$ as its root node. We initially partition the set of MECs of $G$ into two parts: (a) the set of MECs that have an incoming edge adjacent to $r_1$, denoted as $\setofMECs{G, r_1, 1}$, and (b) the set of MECs that have no incoming edge adjacent to $r_1$, denoted as $\setofMECs{G, r_1, 0}$. For simplicity, for any tree graph $H$ with root node $r$, for $i\in \{0,1\}$, we denote $|\setofMECs{H, r, i}|$ by $n_i(H, r)$.
This implies
\begin{equation}
\label{eq:counting-MEC-of-tree-1}
    |\setofMECs{G}| = n_0(G, r_1) + n_1(G, r_1) 
\end{equation}

Let the degree of $r_1$ in $G$ is $\delta$. We partition $\setofMECs{G, r_1, 0}$ into $\setofMECs{G, r_1, 0, 0}$, $\setofMECs{G, r_1, 0, 1}$, \ldots, $\setofMECs{G, r_1, 0, \delta}$, where $\setofMECs{G, r_1, 0, i}$ is the set of MECs of $G$ for which out of $\delta$ edges adjacent to $r_1$, $i$ edges are undirected and the remaining edges are outgoing edges. \Cref{fig:MEcs-of-tree} is given for better understanding. For simplicity, for any tree graph $H$ with root node $r$, we denote $|\setofMECs{H, r, 0, i}|$ by $n_0^i(H, r)$. This implies 
\begin{equation}
\label{eq:counting-MEC-of-tree-2}
    n_0(G, r_1) = \sum_{i = 0}^{\delta}{n_0^i(G, r_1)}
\end{equation}

 \begin{figure}[ht]
    \begin{center}
     \begin{tikzpicture}
    \node[](r){$r$};
    \node[](a)[below left= 0.5 and 0.3 of r]{$a$};
    \node[](b)[below right = 0.5 and 0.3 of r]{$b$};
    \node[](c)[below left= 0.5 and 0.3 of a]{$c$};
    \node[](d)[below right = 0.5 and 0.3 of a]{$d$};
    \node[](G)[below left = 0.2 and 0.0 of d]{$G$};
    
    \draw[-](r)--(a);
    \draw[-](r)--(b);
    \draw[-](a)--(c);
    \draw[-](a)--(d);

    \node[](r1)[right = 2.2 of r]{$r$};
    \node[](a1)[below left= 0.5 and 0.3 of r1]{$a$};
    \node[](b1)[below right = 0.5 and 0.3 of r1]{$b$};
    \node[](c1)[below left= 0.5 and 0.3 of a1]{$c$};
    \node[](d1)[below right = 0.5 and 0.3 of a1]{$d$};
    \node[](G1)[below left = 0.2 and -0.8 of d1]{MEC $M_1$};
    
    \draw[<-](r1)--(a1);
    \draw[<-](r1)--(b1);
    \draw[<-](a1)--(c1);
    \draw[<-](a1)--(d1);

    \node[](r2)[right = 2.2 of r1]{$r$};
    \node[](a2)[below left= 0.5 and 0.3 of r2]{$a$};
    \node[](b2)[below right = 0.5 and 0.3 of r2]{$b$};
    \node[](c2)[below left= 0.5 and 0.3 of a2]{$c$};
    \node[](d2)[below right = 0.5 and 0.3 of a2]{$d$};
    \node[](G2)[below left = 0.2 and -0.8 of d2]{MEC $M_2$};
    
    \draw[->](r2)--(a2);
    \draw[-](r2)--(b2);
    \draw[<-](a2)--(c2);
    \draw[<-](a2)--(d2);
    \end{tikzpicture}
\end{center}
\caption{ $G$ is a tree graph with root node $r$. $M_1$ and $M_2$ are MECs (both obey \cref{item-1-of-thm:nes-and-suf-cond-for-tree-graph-to-be-an-MEC,item-2-of-thm:nes-and-suf-cond-for-tree-graph-to-be-an-MEC} of \cref{thm:nes-and-suf-cond-for-tree-graph-to-be-an-MEC}) with skeleton $G$. $M_1$ belongs to $\setofMECs{G, r, 1}$ as there is an edge $a\rightarrow r$ incoming towards $r$. $M_2$ belongs to $\setofMECs{G, r, 0}$ as none of the edges adjacent to $r$ in $M_2$ is incoming towards $r$. $M_2$ also belongs to $\setofMECs{G, r, 0, 1}$ as it belongs to $\setofMECs{G, r, 0}$ and one edge adjacent to $r$ in $M_2$ is undirected.}
\label{fig:MEcs-of-tree}
\end{figure}
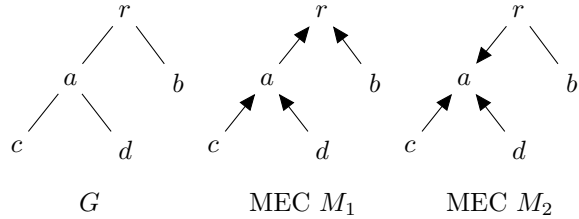

We give a recursive algorithm to compute $n_0^i(G, r_1)$ and $n_1(G, r_1)$. Pick a neighbor $r_2$ of $r_1$ in $G$. Cut the edge $r_1-r_2$ in $G$. This will give us two induced subgraphs of $G$: $G_1$ (containing node $r_1$) and $G_2$ (containing node $r_2$). The degree of $r_1$ in $G_1$ must be $\delta_1  = \delta -1$, where $\delta$ is the degree of $r_1$ in $G$. Let the degree of $r_2$ in $G_2$ be $\delta_2$. \Cref{lem:recursive-method-to-count-MECs-of-tree} shows that if we have the knowledge of $n_1(G_1, r_1)$, $n_0^0(G_1, r_1)$, $n_0^1(G_1, r_1)$, \ldots, $n_0^{\delta_1}(G_1, r_1)$, and
$n_1(G_2, r_2)$, $n_0^0(G_2, r_2)$, $n_0^1(G_2, r_2)$, \ldots, $n_0^{\delta_2}(G_2, r_2)$ then we can compute $n_1(G, r_1)$, $n_0^0(G, r_1)$, $n_0^1(G, r_1)$, \ldots, $n_0^{\delta}(G, r_1)$.

\begin{lemma}
\label{lem:recursive-method-to-count-MECs-of-tree}
Let $G$ be a tree graph, and $(r_1,r_2)\in E_G$ be an edge of $G$. By cutting the edge $(r_1,r_2)$, we get two induced subgraphs of $G$: $G_1$ (containing node $r_1$) and $G_2$ (containing node $r_2$). Then,
\begin{enumerate}
    \item
    \label{item-1-of-recursion}
\begin{equation*}
        n_1(G, r_1) = n_1(G_1,r_1) \times [2\cdot n_1(G_2,r_2)  +  \sum_{j}{(j+2)\cdot n_0^j(G_2,r_2)}] + \sum_i{n_0^i(G_1,r_1)\times[(i+1)\cdot n_1(G_2,r_2)+i\cdot {n_0(G_2,r_2)}]}
    \end{equation*}

    \item 
    \label{item-2a-of-recursion}
\begin{equation*}
     n_0^0(G,r_1)= n_0^0(G_1,r_1)\times [n_1(G_2,r_2)+\sum_j{j.n_0^j(G_2,r_2)}]
 \end{equation*}

 \item 
  \label{item-2-of-recursion}
 For $1\leq i \leq \delta$, 
\begin{equation*}
     n_0^i(G,r_1)= n_0^i(G_1,r_1)\times [n_1(G_2,r_2)+\sum_j{j.n_0^j(G_2,r_2)}]+ n_0^{i-1}(G_1,r_1)\times {n_0(G_2,r_2)}
 \end{equation*}
\end{enumerate}

\end{lemma}

\begin{algorithm}[ht]
\caption{Counting-MEC-Tree($G,r_1$)}
\label{alg:counting-MEC-of-tree}
\SetAlgoLined
\SetKwInOut{KwIn}{Input}
\SetKwInOut{KwOut}{Output}
\SetKwFunction{MEC-Construction}{MEC-Construction}
\KwIn{A tree graph $G$, a node $r_1\in V_G$}
    \KwOut{ $(n_1(G,r_1), n_0^0(G,r_1), n_0^1(G,r_1), \ldots, n_0^{\delta}(G,r_1))$, where $\delta$ is the degree of $r_1$ in $G$.}

    $\delta \leftarrow $ degree of $r_1$ in $G$ \label{alg:counting-MEC-of-tree:delta_intro}
    
    \If{$\delta$ = 0 \label{alg:countingMEC-tree-if-start}}{return (0,1) \label{alg:countingMEC-tree-base-case-return}}\label{alg:countingMEC-tree-if-end}
    
    Pick an edge $(r_1,r_2)\in E_G$. \label{alg:counting-MEC-of-tree:pick-edge-r1-r2}

    Cut the edge $(r_1, r_2)$. \label{alg:counting-MEC-of-tree:cut-edge-r1-r2}
    
    $G_1 \leftarrow$ subtree of $G$ containing $r_1$ after cutting the edge $(r_1,r_2)$ of $G$. \label{alg:counting-MEC-of-tree:G1-intro}
    
    $G_2 \leftarrow$ subtree of $G$ containing $r_2$ after cutting the edge $(r_1,r_2)$ of $G$. \label{alg:counting-MEC-of-tree:G2-intro}

    $\delta_1 \leftarrow$ degree of $r_1$ in $G_1$ \label{alg:counting-MEC-of-tree:delta-1-intro}
    
    $\delta_2 \leftarrow$ degree of $r_2$ in $G_2$ \label{alg:counting-MEC-of-tree:delta-2-intro}
    
    $(b_1, c_1^0, c_1^1, \ldots, c_1^{\delta_1}) \leftarrow$ Counting-MEC-Tree($G_1,r_1)$ \label{alg:counting-MEC-of-tree:count-MEC-call-for-G1}
    
     $(b_2, c_2^0, c_2^1, \ldots, c_2^{\delta_2}) \leftarrow$ Counting-MEC-Tree($G_2,r_2)$ \label{alg:counting-MEC-of-tree:count-MEC-call-for-G2}

     $N \leftarrow c_2^0+c_2^1+\ldots +c_2^{\delta_2}$ \label{alg:counting-MEC-of-tree:N-intro}
    
    $b \leftarrow b_1 \times [2\cdot b_2 + \sum_{j}{(j+2)\cdot c_2^j}]+ \sum_{i}{c_1^i\times[(i+1)\cdot b_2+i\cdot N]}$ \label{alg:counting-MEC-of-tree:b-intro}
    
    $c_0 \leftarrow c_1^0\times [b_2 + \sum_{j}{j\cdot c_2^j}]$ \label{alg:counting-MEC-of-tree:c0-intro}

    $i \leftarrow 1$ \label{alg:counting-MEC-of-tree:i-intro}

    \While{$i\leq \delta$ \label{alg:counting-MEC-of-tree:while-start}}
    {
        $c_i\leftarrow c_1^i \times [b_2 + \sum_{j}{j\cdot c_2^j}] + c_1^{i-1} \times N$ \label{alg:counting-MEC-of-tree:ci-intro}

        $i = i + 1$ \label{alg:counting-MEC-of-tree:i-update}
    }\label{alg:counting-MEC-of-tree:while-end}

    \KwRet $(b,  c_0, c_1, \ldots, c_{\delta})$ \label{alg:counting-MEC-of-tree:return}
\end{algorithm} 
\begin{proof}[Proof sketch of \cref{item-1-of-recursion} of \cref{lem:recursive-method-to-count-MECs-of-tree}]
From \cref{obs:there-exists-a-unique-projection}, for each MEC $M$ of $G$, $(M_1, M_2) = \mathcal{P}(M, V_{G_1}, V_{G_2})$ is unique. This implies that for each MEC $M \in \setofMECs{G, r_1, 1}$, there exists a unique pair $(M_1, M_2)$ such that $(M_1, M_2) = \mathcal{P}(M, V_{G_1}, V_{G_2})$. The pair $(M_1, M_2)$ can fall into four categories: either (a) $M_1 \in \setofMECs{G_1, r_1, 1}$ and $M_2 \in \setofMECs{G_2, r_2, 1}$, or (b) $M_1 \in \setofMECs{G_1, r_1, 1}$ and $M_2 \in \setofMECs{G_2, r_2, 0, j}$ for some $0 \leq j \leq \delta_2$, or (c) $M_1 \in \setofMECs{G_1, r_1, 0, i}$ for some $0 \leq i \leq \delta_1$ and $M_2 \in \setofMECs{G_2, r_2, 1}$, or (d) $M_1 \in \setofMECs{G_1, r_1, 0, i}$ for some $0 \leq i \leq \delta_1$ and $M_2 \in \setofMECs{G_2, r_2, 0, j}$ for some $0 \leq j \leq \delta_2$. For each possible pair $(M_1, M_2)$, we determine the number of MECs $M$ belonging to $\setofMECs{G, r_1, 1}$ such that $\mathcal{P}(M, V_{G_1}, V_{G_2}) = (M_1, M_2)$. By counting all such $M$, we obtain the required $n_1(G, r_1)$.

In the case of each pair of MECs $(M_1, M_2)$, where $M_1 \in \setofMECs{G_1, r_1, 1}$ and $M_2 \in \setofMECs{G_2, r_2, 1}$, two MECs $M$ exist, both of which belong to $\setofMECs{G, r_1, 1}$, with their projections on $V_{G_1}$ and $V_{G_2}$ being $M_1$ and $M_2$, respectively. One possible $M$ is the union of $M_1$, $M_2$, and $r_1\rightarrow r_2$, and the other possible $M$ is the union of $M_1$, $M_2$, and $r_2 \rightarrow r_1$. Verification of these $M$s regarding their status as an MEC is based on satisfying both \cref{item-1-of-thm:nes-and-suf-cond-for-tree-graph-to-be-an-MEC} and \cref{item-2-of-thm:nes-and-suf-cond-for-tree-graph-to-be-an-MEC} of \cref{thm:nes-and-suf-cond-for-tree-graph-to-be-an-MEC}.

Similarly, in the second scenario, for each pair of MECs $(M_1, M_2)$, where $M_1 \in \setofMECs{G_1, r_1, 1}$ and $M_2 \in \setofMECs{G_2, r_2, 0, j}$ for some $0 \leq j \leq \delta_2$, $j+2$ MECs $M$ exist belonging to $\setofMECs{G, r_1, 1}$ with projections on $V_{G_1}$ and $V_{G_2}$ as $M_1$ and $M_2$, respectively. 

For the third case, for each pair of MECs $(M_1, M_2)$ such that $M_1 \in \setofMECs{G_1, r_1, 0, i}$ for some $0 \leq i \leq \delta_1$ and $M_2 \in \setofMECs{G_2, r_2, 1}$, $i+1$ MECs $M$ belong to $\setofMECs{G, r_1, 1}$ with projections of $M$ on $V_{G_1}$ and $V_{G_2}$ as $M_1$ and $M_2$, respectively.

In the fourth scenario, for each pair of MECs $(M_1, M_2)$ such that $M_1 \in \setofMECs{G_1, r_1, 0, i}$ for some $0 \leq i \leq \delta_1$ and $M_2 \in \setofMECs{G_2, r_2, 0, j}$ for some $0 \leq j \leq \delta_2$, $i$ MECs $M$ belong to $\setofMECs{G, r_1, 1}$ with projections on $V_{G_1}$ and $V_{G_2}$ as $M_1$ and $M_2$, respectively.

Counting all the aforementioned possible $M$ gives us the required number of MECs belonging to $\setofMECs{G, r_1, 1}$, i.e., $n_1(G, r_1)$. This proves \cref{item-1-of-recursion} of \cref{lem:recursive-method-to-count-MECs-of-tree}.
\end{proof}

Proof of \cref{item-2a-of-recursion,item-2-of-recursion} of \cref{lem:recursive-method-to-count-MECs-of-tree} follows the same approach as the proof of \cref{item-1-of-recursion} of \cref{lem:recursive-method-to-count-MECs-of-tree}. The complete proof of \cref{lem:recursive-method-to-count-MECs-of-tree} is provided in the Supplementary Material.

We construct \cref{alg:counting-MEC-of-tree} to count the MECs of a tree graph. Lines \ref{alg:countingMEC-tree-if-start}--\ref{alg:countingMEC-tree-if-end} handle the base case when the degree of the root node in the tree graph $G$ is 0, i.e., when $G$ has only one node. In this scenario, only one MEC of $G$ is possible, which is $G$ itself. This MEC belongs to $\setofMECs{G, r, 0, 0}$ because in the MEC, no edge adjacent to $r$ is an incoming edge towards $r$, and no edge adjacent to $r$ is undirected. This implies that when $G$ has only one node, $n_1(G, r) = 0$, and $n_0^0(G, r) = 1$. Line \ref{alg:countingMEC-tree-base-case-return} reflects the same. The remaining lines correctly implement \cref{lem:recursive-method-to-count-MECs-of-tree}.

We now proceed to the second part of the paper, where we construct an algorithm that counts the MECs of a chordal graph.

\subsection{Counting MECs of a chordal graph}
\label{subsection:chordal-graph}
In this subsection, we address \cref{prob:counting-MEC} when $G$ is a chordal graph. We give an algorithm that for a chordal graph $G$, counts the MECs of $G$.
The idea of this algorithm comes from the above-discussed algorithm for counting MECs of a tree graph.
Let us first define an image of an MEC.

\begin{definition}
    \label{def:image-of-an-MEC}
    Let $M$ be an MEC of an undirected graph $G$. For $X \subseteq V_G$, $O$ is said to be an image of $M$ on $X$, if $O = M[X]$.
\end{definition}

The following corollaries come from \cref{def:image-of-an-MEC,obs:induced-subgraph-of-MEC-is-PMEC}.

\begin{corollary}
    \label{corr:MEC-has-a-unique-MEC-on-X}
    Let $M$ be an MEC of an undirected graph $G$. For any $X \subseteq V_G$, there exists a unique image of $M$ on $X$.
\end{corollary}
\begin{proof}
    For $X\subseteq V_G$, the induced subgraph of $M$ on $X$ is unique.
\end{proof}

\begin{corollary}
    \label{corr:image-is-a-PMEC}
   Let $M$ be an MEC of an undirected graph $G$, and $O$ be an image of $M$ on $X \subseteq V_G$. Then, $O$ is a partial MEC of $G[X]$.
\end{corollary}
\begin{proof}
    From \cref{obs:induced-subgraph-of-MEC-is-PMEC}, an induced subgraph of an MEC is a partial MEC. Since an image of an MEC is an induced subgraph of it, it is a partial MEC.
\end{proof}

We partition the set of MECs $M$ of $G$ based on its image. Let $\setofMECs{G, O}$ denotes the set of MECs $M$ of $G$ for which $O$ is an image, i.e., $M[V_O] = O$.

From \cref{corr:MEC-has-a-unique-MEC-on-X}, the image of an MEC $M$ of $G$ on $X \subseteq V_G$ is unique. Additionally, from \cref{corr:image-is-a-PMEC}, the image $O$ of an MEC $M$ of an undirected graph $G$ on $X$ is a partial MEC of $G[X]$. This implies that for each MEC $M$ of $G$, there exists a unique partial MEC $O$ of $G[X]$ such that $M \in \setofMECs{G, O}$. Therefore, the sum of $|\setofMECs{G, O}|$ over all partial MECs $O$ of $G[X]$ is $|\setofMECs{G}|$, i.e., for any $X \subseteq V_G$,

\begin{equation}
    \label{eq:summation-of-PMECs}
    |\setofMECs{G}| = \sum_{O\in\setofpartialMECs{G[X]}}{|\setofMECs{G, O}|}
\end{equation}

\Cref{eq:summation-of-PMECs} reduces the problem of counting MECs of $G$ into \cref{problem-2}
\begin{problem}
\label{problem-2}
    \textbf{Input:} An undirected connected chordal graph $G$, and $X\subseteq V_G$.
    \textbf{Output:} Compute $|\setofMECs{G, O}|$ for each partial MEC $O$ of $G[X]$.    
\end{problem}

For any chordal graph $G$, we first fix an $X$, and then give a recursive method to solve \cref{problem-2}. In this way, using \cref{eq:summation-of-PMECs}, we able to compute $|\setofMECs{G}|$.
 For this, we first construct a clique tree representation (\cref{def:clique-tree}) $T$ of the chordal graph $G$. 
 We arbitrarily pick a node $r_1$ as a root node of $T$. 
We give a recursive method to solve \cref{problem-2} for $X = r_1\cup N(r_1, G)$.
 
 We arbitrarily pick a neighbor $r_2$ of $r_1$ in $T$. Cut the edge $r_1-r_2$ in $T$. This gives us two induced subgraphs $T_1$ and $T_2$ of $T$ such that $r_1$ and $r_2$ are the root nodes of $T_1$ and $T_2$, respectively.  For $a\in \{1,2\}$, let $G_a$ be the induced subgraph of $G$ that is represented by $T_a$. 
 We show that if for each $a\in \{1,2\}$, for each $O_a \in \setofpartialMECs{G_a[r_a\cup N(r_a, G_a)]}$, we have the knowledge of $|\setofMECs{G_a, O_a}|$ then for each $O \in \setofpartialMECs{G[r_1\cup N(r_1, G)]}$, we can compute $|\setofMECs{G, O}|$. 

 We start with finding a relation between the MECs of $G$, $G_1$, and $G_2$.
   The induced subgraphs $G_1$ and $G_2$ of $G$ has the following properties: (a) $G_1 \cup G_2 = G$, and (b) $I = V_{G_1} \cap V_{G_2} = r_1 \cap r_2$ is a vertex separator of  $G$ that separates $V_{G_1}\setminus{I}$ and $V_{G_2}\setminus I$ (\cite{blair1993introduction}). 
 Let $M, M_1$ and $M_2$ be the MECs of, respectively, $G, G_1$ and $G_2$ such that $\mathcal{P}(M, V_{G_1}, V_{G_2}) = (M_1, M_2)$. 
The following observation provides necessary conditions obeyed by images of $M$, $M_1$, and $M_2$.
\begin{observation}
\label{obs:nes-cond-for-chordal-graph}
Let $G$ be a chordal graph, and $G_1$ and $G_2$ be two induced subgraphs of $G$ such that $G = G_1\cup G_2$, and $I = V_{G_1}\cap V_{G_2}$ is a vertex separator of $G$ that separates $V_{G_1}\setminus I$ and $V_{G_2}\setminus I$. Let $r_1$ and $r_2$ be cliques of, respectively, $G_1$ and $G_2$ such that $r_1\cap r_2 = I$. Let $X_1 = r_1 \cup N(r_1, G_1)$, $X_2 = r_2 \cup N(r_2, G_2)$, and $X = r_1 \cup r_2 \cup N(r_1\cup r_2, G)$. Let $M$, $M_1$ and $M_2$ be MECs of $G$, $G_1$, and $G_2$, respectively. Let $O = M[X]$, $O_1 = M_1[X_1]$, and $O_2 = M_2[X_2]$. If $(M_1, M_2) = \mathcal{P}(M, V_{G_1}, V_{G_2})$ then
    \begin{enumerate}
        \item
        \label{item-1-of-obs:nes-cond-for-chordal-graph}
        For $a\in \{1,2\}$, if $u\rightarrow v \in O_a$ then $u\rightarrow v \in O$.
        \item
        \label{item-2-of-obs:nes-cond-for-chordal-graph} For $a \in \{1,2\}$, $\mathcal{V}(O_a) = \mathcal{V}(O[V_{O_a}])$.
        \item
        \label{item-3-of-obs:nes-cond-for-chordal-graph}
        For $a \in \{1,2\}$, if $u -v \in O_a$ and $u\rightarrow v \in O$ then $u\rightarrow v$ is strongly protected in $O$.
    \end{enumerate}
\end{observation}
\begin{proof}[Proof of \cref{obs:nes-cond-for-chordal-graph}]
\Cref{obs:nes-cond-for-chordal-graph} follows the following result of \cite{sharma2023fixedparameter}.
\begin{lemma}[Lemma 3.2 of \cite{sharma2023fixedparameter}]
\label{lem:sharma2023results}
    Let $G$ be an undirected graph, and let $G_1$ and $G_2$ be two induced subgraphs of $G$ such that $G = G_1 \cup G_2$, and $I = V_{G_1} \cap V_{G_2}$ is a vertex separator of $G$ that separates $V_{G_1}\setminus I$ and $V_{G_2} \setminus I$. Let $r_1$ and $r_2$ be subsets of, respectively,  $V_{G_1}$ and $V_{G_2}$ such that $r_1\cap r_2 = I$. Let $M$ be an MEC of $G$, $\mathcal{P}(M, V_{G_1}, V_{G_2}) = (M_1, M_2)$, $O = M[r_1\cup r_2 \cup N(r_1\cup r_2, G)]$, $O_1 = M_1[r_1 \cup N(r_1, G_1)]$, and $O_2 = M_2[r_2 \cup N(r_2, G_2)]$. Then, for each $a\in \{1,2\}$,
    \begin{enumerate}
        \item 
        \label{item-1-of-lem:sharma2023results}
        If $u\rightarrow v \in O_a$ then $u\rightarrow v \in O$.
         \item
        \label{item-2-of-lem:sharma2023results}
        $\mathcal{V}(O_a) = \mathcal{V}(O[V_{O_a}])$.
         \item
        \label{item-3-of-lem:sharma2023results}
        If $u -v \in O_a$ then $u\rightarrow v \in O$ if, and only if, either of the following occurs:
         \begin{enumerate}
             \item
             \label{item-3a-of-lem:sharma2023results}
             $u\rightarrow v$ is strongly protected in $O$.
             \item
             \label{item-3b-of-lem:sharma2023results}
             there exists $x-y \in O_a$ such that $x\rightarrow y\in O$, and there exists a \tfp{} from $(x,y)$ to $(u,v)$ in $M_a$ with $(x,y) \neq (u,v)$.
             \item
             \label{item-3c-of-lem:sharma2023results}
             there exists $x-y \in O_a$ such that $x\rightarrow y \in O$, and there exist \tfps{} from $(x, y)$ to $v$, and from $(v,u)$ to $x$ with $y\neq v$ and $u\neq x$.
         \end{enumerate}
    \end{enumerate}
\end{lemma}

From the input constructions provided in \cref{obs:nes-cond-for-chordal-graph}, $G_1$ and $G_2$ are induced subgraphs of $G$ such that $G = G_1\cup G_2$, $I = V_{G_1} \cap V_{G_2}$ is a vertex separator of $G$ that separates $V_{G_1} \setminus I$ and $V_{G_2} \setminus I$, $r_1$ and $r_2$ are the subsets of, respectively, $V_{G_1}$ and $V_{G_2}$ (as they are cliques of, respectively, $G_1$ and $G_2$) such that $r_1\cap r_2 = I$. 
This shows that the input constraints of \cref{obs:nes-cond-for-chordal-graph} obey the input constraints of \cref{lem:sharma2023results}. Therefore, we can use the result of \cref{lem:sharma2023results} to prove our claims in \cref{obs:nes-cond-for-chordal-graph}.
\Cref{item-1-of-lem:sharma2023results,item-2-of-lem:sharma2023results} of 
     \cref{lem:sharma2023results} verifies, respectively, \cref{item-1-of-obs:nes-cond-for-chordal-graph,item-2-of-obs:nes-cond-for-chordal-graph} of \cref{obs:nes-cond-for-chordal-graph}. Correctness of \cref{item-3-of-obs:nes-cond-for-chordal-graph} of \cref{obs:nes-cond-for-chordal-graph} comes from the fact that when $G, G_1$, and $G_2$ are chordal, and $r_1$ and $r_2$ are cliques, \cref{item-3a-of-lem:sharma2023results,item-3b-of-lem:sharma2023results} of \cref{lem:sharma2023results} implies \cref{item-3a-of-lem:sharma2023results} of \cref{lem:sharma2023results}. This further implies that \cref{item-3-of-lem:sharma2023results} of \cref{lem:sharma2023results} implies \cref{item-3-of-obs:nes-cond-for-chordal-graph} of \cref{obs:nes-cond-for-chordal-graph}. The complete proof of \cref{item-3-of-obs:nes-cond-for-chordal-graph} of \cref{obs:nes-cond-for-chordal-graph} is provided in the Supplementary Material.
\end{proof}

We define the following to summarize the necessary conditions provided by \cref{obs:nes-cond-for-chordal-graph}.

\begin{definition}[\textbf{Extension}]
    \label{def:extension}
    Let $G$ be a chordal graph, and $G_1$ and $G_2$ be two induced subgraphs of $G$ such that $G = G_1\cup G_2$, and $I = V_{G_1}\cap V_{G_2}$ is a vertex separator of $G$ that separates $V_{G_1}\setminus I$ and $V_{G_2}\setminus I$. Let $r_1$ and $r_2$ be cliques of, respectively, $G_1$ and $G_2$ such that $r_1\cap r_2 = I$. Let $X_1 = r_1 \cup N(r_1, G_1)$, $X_2 = r_2 \cup N(r_2, G_2)$, and $X = r_1 \cup r_2 \cup N(r_1\cup r_2, G)$. Let $O \in \setofpartialMECs{G[X]}$, $O_1 \in \setofpartialMECs{G_1[X_1]}$, and $O_2 \in \setofpartialMECs{G_2[X_2]}$. We call $O$ as an extension of $O_1$ and $O_2$, denoted as $O \in \mathcal{E}(O_1, O_2)$, if
    \begin{enumerate}
        \item 
        \label{item-1-of-def:extensions}
        For $a \in \{1,2\}$, if $u\rightarrow v \in O_a$ then $u\rightarrow v \in O$.
         \item
        \label{item-2-of-def:extension}
        For $a \in \{1,2\}$, $\mathcal{V}(O_a) = \mathcal{V}(O[V_{O_a}])$.
         \item
        \label{item-3-of-def:extension}
        For $a \in \{1,2\}$, for  $u -v \in O_a$, if $u\rightarrow v \in O$ then $u\rightarrow v$ is strongly protected in $O$.
    \end{enumerate}
\end{definition}

Since  $O$, $O_1$ and $O_2$ used in \cref{obs:nes-cond-for-chordal-graph}  are partial MECs of, respectively, $G[X]$, $G_1[X_1]$, and $G_2[X_2]$ (from \cref{obs:induced-subgraph-of-MEC-is-PMEC}). And, they obey \cref{item-1-of-def:extensions,item-2-of-def:extension,item-3-of-def:extension} of \cref{def:extension}. Therefore, we can summarize \cref{obs:nes-cond-for-chordal-graph} as follows:

\begin{lemma}
    \label{lem:nes-cond-for-chordal-graph}
    Let $G$ be a chordal graph, and $G_1$ and $G_2$ be two induced subgraphs of $G$ such that $G = G_1\cup G_2$, and $I = V_{G_1}\cap V_{G_2}$ is a vertex separator of $G$ that separates $V_{G_1}\setminus I$ and $V_{G_2}\setminus I$. Let $r_1$ and $r_2$ be cliques of, respectively, $G_1$ and $G_2$ such that $r_1\cap r_2 = I$. Let $X_1 = r_1 \cup N(r_1, G_1)$, $X_2 = r_2 \cup N(r_2, G_2)$, and $X = r_1 \cup r_2 \cup N(r_1\cup r_2, G)$. Let $M$, $M_1$ and $M_2$ be MECs of $G$, $G_1$, and $G_2$, respectively. Let $O = M[X]$, $O_1 = M_1[X_1]$, and $O_2 = M_2[X_2]$. If $(M_1, M_2) = \mathcal{P}(M, V_{G_1}, V_{G_2})$ then $O \in \mathcal{E}(O_1, O_2)$.
\end{lemma}

We now show that the necessary conditions that are obeyed by the partial MECs $O, O_1$, and $O_2$ in \cref{lem:nes-cond-for-chordal-graph} are sufficient to answer whether there exists an MEC $M \in \setofMECs{G, O}$ for any pair of MECs $(M_1, M_2) \in \setofMECs{G_1, O_1}\times \setofMECs{G_2, O_2}$ such that $\mathcal{P}(M, V_{G_1}, V_{G_2}) = (M_1, M_2)$.

\begin{lemma}
    \label{lem:suff-cond-for-chordal-graph}
    Let $G$ be a chordal graph, and $G_1$ and $G_2$ be two induced subgraphs of $G$ such that $G = G_1\cup G_2$, and $I = V_{G_1}\cap V_{G_2}$ is a vertex separator of $G$ that separates $V_{G_1}\setminus I$ and $V_{G_2}\setminus I$. Let $r_1$ and $r_2$ be cliques of, respectively, $G_1$ and $G_2$ such that $r_1\cap r_2 = I$. Let $X_1 = r_1 \cup N(r_1, G_1)$, $X_2 = r_2 \cup N(r_2, G_2)$, and $X = r_1 \cup r_2 \cup N(r_1\cup r_2, G)$. Let $O, O_1$, and $O_2$ be  partial MECs of, respectively, $G[X]$, $G_1[X_1]$, and $G_2[X_2]$ such that $O$ is an extension of $O_1$ and $O_2$, i.e., $O \in \mathcal{E}(O_1, O_2)$.
    Then, for each pair of MECs $(M_1, M_2) \in \setofMECs{G_1, O_1} \times \setofMECs{G_2, O_2}$, there exists a unique MEC $M \in \setofMECs{G, O}$ such that $\mathcal{P}(M, V_{G_1}, V_{G_2}) = (M_1, M_2)$. 
\end{lemma}

\begin{proof}[Proof of \cref{lem:suff-cond-for-chordal-graph}]
    Let $M_1 \in \setofMECs{G_1, O_1}$, and $M_2 \in \setofMECs{G_2, O_2}$. We construct an MEC $M$ such that $M \in \setofMECs{G, O}$ and $\mathcal{P}(M, V_{G_1}, V_{G_2}) = (M_1, M_2)$. Later we show the uniqueness of $M$.
    Our construction of $M$ starts with the Markov union (\cref{def:Markov-union-of-graphs}) of $M_1$ and $M_2$ and $O$. Since the precondition for the Markov union of graphs is that the graphs need to be pairwise synchronous, we demonstrate that $M_1$, $M_2$, and $O$ satisfy the pairwise synchronous condition.

    \begin{observation}
        \label{obs:M1-M2-and-O-are-synchronous}
        $M_1$, $M_2$, and $O$ are synchronous graphs.
    \end{observation}
    \begin{proof}
        From the construction, $V_{M_1}\cap V_O = r_1\cup N(r_1, G_1)$, and $V_{M_2}\cap V_O = r_2\cup N(r_2, G_2)$. Since $O$ is an extension of $O_1 = M_1[r_1\cup N(r_1, G_1)]$ and $O_2 = M_2[r_2\cup N(r_2, G_2)]$, from \cref{def:extension,def:synchronous-graphs}, $O$ and $M_1$, and $O$ and $M_2$ are pairwise synchronous. 

        Suppose $M_1$ and $M_2$ are not synchronous. Then, there must exist nodes $u,v \in V_{M_1}\cap V_{M_2} = r_1\cap r_2 = I$ such that $u\rightarrow v \in M_1$ and $v\rightarrow u \in M_2$. Since $O_1 = M_1[r_1\cup N(r_1, G_1)]$, and $O_2 = M_2[r_2\cup N(r_2, G_2)]$, $u\rightarrow v \in O_1$ and $v\rightarrow u \in O_2$. But, then $O$ cannot be an extension of $O$ as it must violate \cref{item-1-of-def:extensions} of \cref{def:extension}. This leads to a contradiction, as $O$ is an extension of $O_1$ and $O_2$. This implies $M_1$ and $M_2$ are synchronous. This completes the proof of \cref{obs:M1-M2-and-O-are-synchronous}.
    \end{proof}

    We construct the required MEC $M$ using the following steps:
    \begin{enumerate}
        \item
            \label{step-1-of-construction-of-M}
            Initialize $M = U_M(M_1, M_2, O)$, Markov union (\cref{def:Markov-union-of-graphs}) of $M_1$, $M_2$, and $O$.
        \item
            \label{step-2-of-construction-of-M}
            For each $i \in \{1,2\}$, for each \ucc{} $\mathcal{C}$ of $M_i$,
            if $V_{\mathcal{C}} \cap r_i \neq \emptyset$ then 
            we do the following:
        \begin{enumerate}
            \item
            \label{step-2-a-of-construction-of-M}
            Construct an LBFS ordering $\tau$ of $\mathcal{C}$ that starts with $V_{\mathcal{C}} \cap r_i$.
\item
            \label{step-2-b-of-construction-of-M}
            For $u-v \in \mathcal{C}$ such that $\tau(u) < \tau(v)$, and $v \in V_{\mathcal{C}}\setminus{r_i\cup N(r_i, G_i)}$, if $u-v \in M$ then replace $u-v$ with $u\rightarrow v$ in $M$ if one of the following occurs:
            \begin{enumerate}
                \item
            \label{step-2-b-a-of-construction-of-M}
            There exists an induced subgraph of the form $w\rightarrow u -v \in M$ such that $\tau(w) < \tau(v)$.
                \item
            \label{step-2-b-b-of-construction-of-M}
            There exists an induced subgraph of the form $u\rightarrow w\rightarrow v-u \in M$ such that $\tau(w) < \tau(v)$.
            \end{enumerate}
\end{enumerate}
    \end{enumerate}

    We now prove that $M$ is the required MEC. For this, we prove the following: (a) the skeleton of $M$ is $G$ (\cref{obs:skeleton-of-M-is-G}), (b) $M[V_O] = O$ (\cref{obs:O-is-an-induced-subgraph-of-M}),  (c) $M$ is an MEC (\cref{obs:M-is-an-MEC}), (d) $\mathcal{P}(M, V_{G_1}, V_{G_2}) = (M_1, M_2)$ (\cref{obs:M1-and-M2-are-projections-of-M}), (e) $M$ is unique (\cref{obs:M-is-unique}).
    We start with the first one.

    \begin{observation}
        \label{obs:skeleton-of-M-is-G}
        The skeleton of $M$ is $G$.
    \end{observation}
    \begin{proof}
        From the initialization step (step \ref{step-1-of-construction-of-M} of the construction of $M$), $\skel{M} = \skel{M_1} \cup \skel{M_2} \cup \skel{O} = G_1 \cup G_2 \cup G[r_1\cup r_2\cup N(r_1\cup r_2, G] = G$. At step \ref{step-2-of-construction-of-M} of the construction of $M$, we only update $M$ by adding an orientation to an undirected edge.  This does not change the skeleton. Therefore, the skeleton of $M$ is $G$.
    \end{proof}

    We now show the second one.

    \begin{observation}
        \label{obs:O-is-an-induced-subgraph-of-M}
        $M[V_O] = O$.
    \end{observation}
    \begin{proof}
        From \cref{obs:skeleton-of-M-is-G}, the skeleton of $M$ is $G$. This implies $\skel{M[V_O]} = G[r_1\cup r_2\cup N(r_1\cup r_2, G)] = \skel{O}$. Therefore, to prove $M[V_O] = O$, we only need to prove that if $u\rightarrow v \in O$ then $u\rightarrow v \in M$, and if $u-v \in O$ then $u-v \in M$.

        Suppose $u\rightarrow v \in O$. Then, from \cref{def:Markov-union-of-graphs}, $u\rightarrow v \in U_M(M_1, M_2, O)$. This implies at step \ref{step-1-of-construction-of-M} of the construction of $M$, we have $u\rightarrow v \in M$. Since, at the second step of the construction of $M$, we do not change any directed edge of $M$, $u\rightarrow v \in M$.

        Suppose $u-v \in O$. This implies $u,v \in V_O = r_1\cup r_2 \cup N(r_1\cup r_2, G)$. We first show that at step \ref{step-1-of-construction-of-M} of the construction of $M$, we have $u-v \in M$. 
        Since $r_1 \cap r_2$ is a vertex separator of $G$, either $u, v \in r_1\cap N(r_1, G_1)$ or $u, v  \in r_2 \cap N(r_1, G_1)$. If $u,v \in r_1\cup N(r_1, G_1)$ then from \cref{item-1-of-def:extensions} of \cref{def:extension}, $u-v \in O_1$. Since $O_1$ is an induced subgraph of $M_1$, $u-v \in M_1$. And, if $u,v \notin r_1\cup N(r_1, G_1)$ then $u-v \notin \skel{M_1}$.
        Similarly, if $u,v \in r_2\cup N(r_2, G_2)$ then from \cref{item-1-of-def:extensions} of \cref{def:extension}, $u-v \in O_2$. Since $O_2$ is an induced subgraph of $M_2$, $u-v \in M_2$. And, if $u,v \notin r_2\cup N(r_2, G_2)$ then $u-v \notin \skel{M_2}$.
        From \cref{def:Markov-union-of-graphs}, this implies $u-v \in U_M(M_1, M_2, O)$.
        This further implies that at  \cref{step-1-of-construction-of-M} of the construction of $M$, we have $u-v \in M$. Since \cref{step-2-of-construction-of-M} of the construction of $M$ does not change an undirected edge with endpoints in $r_1\cup r_2\cup N(r_1\cup r_2, G)$, this implies that $u-v \in M$.

        The above discussion implies that $O$ is an induced subgraph of $M$. This completes the proof of \cref{obs:O-is-an-induced-subgraph-of-M}.
    \end{proof}

    We now prove the third one.
    \begin{observation}
        \label{obs:M-is-an-MEC}
        $M$ is an MEC.
    \end{observation}
    \begin{proof}
        \Cref{thm:nes-and-suf-cond-for-chordal-graph-to-be-an-MEC} provides necessary and sufficient conditions for a graph to be an MEC. \Cref{obs:M-is-a-chain-graph,obs:ucc-of-M-is-chordal,obs:M-does-not-have-a->b-c,obs:directed-edges-of-M-are-sp} show that $M$ obeys all the conditions of
        \cref{thm:nes-and-suf-cond-for-chordal-graph-to-be-an-MEC}. 
               \begin{observation}
            \label{obs:M-is-a-chain-graph}
            $M$ is a chain graph.
        \end{observation}
        \begin{observation}
            \label{obs:ucc-of-M-is-chordal}
            The \uccs{} of $M$ are chordal.
        \end{observation}
        \begin{observation}
            \label{obs:M-does-not-have-a->b-c}
            $M$ does not have an induced subgraph of the form $a\rightarrow b-c$.
        \end{observation}
        \begin{observation}
            \label{obs:directed-edges-of-M-are-sp}
            The directed edges of $M$ are strongly protected.
        \end{observation}
        From \cref{thm:nes-and-suf-cond-for-chordal-graph-to-be-an-MEC}, \cref{obs:M-is-a-chain-graph,obs:ucc-of-M-is-chordal,obs:M-does-not-have-a->b-c,obs:directed-edges-of-M-are-sp} implies that $M$ is an MEC. Proofs of \cref{obs:M-is-a-chain-graph,obs:ucc-of-M-is-chordal,obs:M-does-not-have-a->b-c,obs:directed-edges-of-M-are-sp} are provided in the Supplementary material. This completes the proof of \cref{obs:M-is-an-MEC}.
    \end{proof}    

    The following observations show that $M_1$ and $M_2$ are the projections of $M$ on, respectively, $V_{G_1}$ and $V_{G_2}$, and $M$ is unique. Proofs are available in the Supplementary Material.

    \begin{observation}
        \label{obs:M1-and-M2-are-projections-of-M}
        $\mathcal{P}(M, V_{G_1}, V_{G_2}) = (M_1, M_2)$.
    \end{observation}

    \begin{observation}
        \label{obs:M-is-unique}
        There exists a unique $M \in \setofMECs{G, O}$ such that $\mathcal{P}(M, V_{G_1}, V_{G_2}) = (M_1, M_2)$.
    \end{observation}
    
      This completes the proof of \cref{lem:suff-cond-for-chordal-graph}. 
\end{proof}

\Cref{lem:formula-to-compute-MEC-G-O} uses \cref{lem:nes-cond-for-chordal-graph,lem:suff-cond-for-chordal-graph} to construct a formula that computes $|\setofMECs{G, O}|$ for any partial MEC $O \in \setofpartialMECs{G[r_1\cup r_2\cup N(r_1\cup r_2, G)]}$.

\begin{lemma}
    \label{lem:formula-to-compute-MEC-G-O}
    Let $G$ be a chordal graph, and $G_1$ and $G_2$ be two induced subgraphs of $G$ such that $G = G_1\cup G_2$, and $I = V_{G_1}\cap V_{G_2}$ is a vertex separator of $G$ that separates $V_{G_1}\setminus I$ and $V_{G_2}\setminus I$. Let $r_1$ and $r_2$ be cliques of, respectively, $G_1$ and $G_2$ such that $r_1\cap r_2 = I$. Let $X_1 = r_1 \cup N(r_1, G_1)$, $X_2 = r_2 \cup N(r_2, G_2)$, and $X = r_1 \cup r_2 \cup N(r_1\cup r_2, G)$. Then, for any $O \in \setofpartialMECs{G[X]}$,
    \begin{equation}
        \label{eq:lem:formula-to-compute-MEC-G-O}
        |\setofMECs{G, O}| = \sum_{\substack{O_1 \in \setofpartialMECs{G_1[X_1]}\\ O_2 \in \setofpartialMECs{G_2[X_2]}\\O \in \mathcal{E}(O_1, O_2)}}{|\setofMECs{G_1, O_1}| \times |\setofMECs{G_2, O_2}|}.
    \end{equation}
\end{lemma}

Using \cref{lem:formula-to-compute-MEC-G-O}, knowing $|\setofMECs{G_1, O_1}|$ for each $O_1 \in \setofpartialMECs{G[r_1\cup N(r_1, G_1)]}$, and $|\setofMECs{G_2, O_2}|$ for each $O_2 \in \setofpartialMECs{G[r_2\cup N(r_2, G_2)]}$, we can compute $|\setofMECs{G, O}|$ for each $O \in \setofpartialMECs{G[r_1\cup r_2 \cup N(r_1\cup r_2, G]}$. 
But, for the recursion, we require to compute $|\setofMECs{G, O'}|$ for each $O' \in \setofpartialMECs{G[r_1\cup N(r_1, G)]}$.

We use the following lemma to evaluate $|\setofMECs{G, O'}|$ for any $O' \in \setofpartialMECs{G[r_1\cup N(r_1, G)]}$, when we have the knowledge of $|\setofMECs{G, O}|$ for all $O \in \setofpartialMECs{G[r_1\cup r_2\cup N(r_1 \cup r_2, G)]}$.
\begin{lemma}
    \label{lem:computation-of-MEC-G-O'}
    Let $G$ be a chordal graph, and $X'\subseteq X \subseteq V_G$. Then, for any $O' \in \setofpartialMECs{G[X]}$,
    \begin{equation}
        \label{eq:lem:formula-to-compute-MEC-G-O'}
        |\setofMECs{G, O'}| = \sum_{\substack{O \in \setofpartialMECs{G[X]}\\ O' = O[X']}}{|\setofMECs{G, O}| }.
    \end{equation}
\end{lemma}
\begin{proof}[Proof sketch of \cref{lem:computation-of-MEC-G-O'}]
    For any MEC $M \in \setofMECs{G, O'}$, there exists $O = M[X]$. Since $X' \subseteq X$, $O' = O[X']$. And, for any MEC $M$ with induced subgraph $O$ such that $O'$ is an induced subgraph of $O$, $O'$ is also an induced subgraph of $M$, i.e., $M \in \setofMECs{G, O'}$.
\end{proof}

\Cref{lem:formula-to-compute-MEC-G-O,lem:computation-of-MEC-G-O'} imply the following:
\begin{lemma}
    \label{lem:computation-of-MEC-G-O'-2}
    Let $G$ be a chordal graph, and $G_1$ and $G_2$ be two induced subgraphs of $G$ such that $G = G_1\cup G_2$, and $I = V_{G_1}\cap V_{G_2}$ is a vertex separator of $G$ that separates $V_{G_1}\setminus I$ and $V_{G_2}\setminus I$. Let $r_1$ and $r_2$ be cliques of, respectively, $G_1$ and $G_2$ such that $r_1\cap r_2 = I$. Let $X_1 = r_1 \cup N(r_1, G_1)$, $X_2 = r_2 \cup N(r_2, G_2)$, $X = r_1 \cup r_2 \cup N(r_1\cup r_2, G)$, and $X' = r_1 \cup N(r_1, G)$. Then, for any $O' \in \setofpartialMECs{G[X']}$,
    \begin{equation}
        \label{eq:lem:computation-of-MEC-G-O'-2}
        |\setofMECs{G, O'}| = \sum_{\substack{O_1 \in \setofpartialMECs{G_1[X_1]}\\ O_2 \in \setofpartialMECs{G_2[X_2]}\\O \in \setofpartialMECs{G[X]}\\O \in \mathcal{E}(O_1, O_2)\\ O' = O[X]}}{|\setofMECs{G_1, O_1}| \times |\setofMECs{G_2, O_2}|}.
    \end{equation}
\end{lemma}

\subsubsection{Algorithm for counting MECs of a chordal graph}
\label{sec:algorithm-for-counting-MECs}
We use \cref{lem:computation-of-MEC-G-O'-2} to construct a recursive algorithm (\cref{alg:counting-MEC-chordal-tree-decommposition}) that for inputs $G$, a clique tree representation $T$ of $G$, and a node $r_1$ of $T$, outputs a function $f$ such that for $O \in \setofpartialMECs{G[r_1\cup N(r_1, G)]}$, $f(O) = |\setofMECs{G, O}|$. \Cref{lem:proof-of-correctness-of-count-MEC-algorithm} provides the correctness of \cref{alg:counting-MEC-chordal-tree-decommposition}.

\begin{algorithm}[H]
\caption{Count($G, T, r_1$)}
\label{alg:counting-MEC-chordal-tree-decommposition}
\SetAlgoLined
\SetKwInOut{KwIn}{Input}
\SetKwInOut{KwOut}{Output}
\SetKwFunction{MEC-Construction}{MEC-Construction}
\KwIn{A connected chordal graph $G$, a clique tree representation $T$ of $G$, a node $r_1\in V_T$.}
    \KwOut{A function $f: \setofpartialMECs{G[r_1\cup N(r_1, G_1)]} \rightarrow  \mathbb{Z}$ such that for $O \in \setofpartialMECs{G[r_1\cup N(r_1, G)]}$, $f(O) = |\setofMECs{G,O}|$.}

    Create a function $f: \setofpartialMECs{G[r_1\cup N(r_1, G_1)]} \rightarrow  \mathbb{Z}$ \label{alg:counting-MEC-chordal-tree-decommposition-creating-function-f}

    \ForEach{$O \in \setofpartialMECs{G[r_1\cup N(r_1, G_1)]}$ \label{alg:counting-MEC-chordal-tree-decommposition-foreach-1-start}}
    {
        $f(O) \leftarrow 0$ \label{alg:counting-MEC-chordal-tree-decommposition-f-O-init}
    }\label{alg:counting-MEC-chordal-tree-decommposition-for-each-1-end}
    
   \If{degree of $r_1$ in $T$ = 0 \label{alg:counting-MEC-chordal-tree-decommposition-if-degree-is-1}}
   {
        $f(G) =  1$ \label{alg:counting-MEC-chordal-tree-decommposition:inside-if-case}
        
       \KwRet $f$ \label{alg:counting-MEC-chordal-tree-decommposition-return-inside-if-case}
   }\label{alg:counting-MEC-chordal-tree-decommposition-if-end}
   
   Pick an arbitrary edge $r_1-r_2 \in E_T$ \label{alg:counting-MEC-chordal-tree-decommposition-pick-edge-r1-r2}
   
   Cut the edge $r_1-r_2$ of $T$ \label{alg:counting-MEC-chordal-tree-decommposition-cut-edge}
   
   $T_1 \leftarrow$ induced subtree of $T$ containing $r_1$ \label{alg:counting-MEC-chordal-tree-decommposition-T1-definition}
   
   $T_2 \leftarrow$ induced subtree of $T$ containing $r_2$ \label{alg:counting-MEC-chordal-tree-decommposition-T2-definition}
   
   $G_1 \leftarrow$ induced subgraph of $G$ represented by $T_1$ \label{alg:counting-MEC-chordal-tree-decommposition-G1-definition}
   
   $G_2 \leftarrow$ induced subgraph of $G$ represented by $T_2$ \label{alg:counting-MEC-chordal-tree-decommposition-G2-definition}
   
   $f_1 \leftarrow$ Count$(G_1,T_1,r_1)$ \label{alg:counting-MEC-chordal-tree-decommposition-f1-introduction}
   
   $f_2 \leftarrow$ Count$(G_2,T_2,r_2)$ \label{alg:counting-MEC-chordal-tree-decommposition-f2-introduction}

   \ForEach{$O_1 \in \setofpartialMECs{G_1[r_1\cup N(r_1, G_1)]}$ \label{alg:counting-MEC-chordal-tree-decommposition-for-each-2-start}}
   {
        \ForEach{$O_2 \in \setofpartialMECs{G_2[r_2\cup N(r_2, G_2)]}$ \label{alg:counting-MEC-chordal-tree-decommposition-for-each-3-start}}
           {
                \ForEach{$O \in \setofpartialMECs{G[r_1\cup r_2 \cup N(r_1\cup r_2, G)]}$ \label{alg:counting-MEC-chordal-tree-decommposition-for-each-4-start}}
                   {
                        \If{$O\in \mathcal{E}(O_1, O_2)$ \label{alg:counting-MEC-chordal-tree-decommposition-if-2-start}}
                        {
                            $O' \leftarrow O[r_1\cup N(r_1, G)]$ \label{alg:counting-MEC-chordal-tree-decommposition-O'-introduction}

                            $f(O') = f(O') + f_1(O_1)\times f_2(O_2)$ \label{alg:counting-MEC-chordal-tree-decommposition-f-O'-update}
                        } \label{alg:counting-MEC-chordal-tree-decommposition-if-2-end}
                   }\label{alg:counting-MEC-chordal-tree-decommposition-for-each-4-end}
           }\label{alg:counting-MEC-chordal-tree-decommposition-for-each-3-end}
   } \label{alg:counting-MEC-chordal-tree-decommposition-for-each-2-end}

   \KwRet $f$ \label{alg:counting-MEC-chordal-tree-decommposition-final-return-statement}
\end{algorithm} 
\begin{lemma}
    \label{lem:proof-of-correctness-of-count-MEC-algorithm}
    Let $G$ be a chordal graph, $T$ be a clique tree representation of $G$, and $r_1$ be a node of $T$.  For input $G, T$ and $r_1$, \cref{alg:counting-MEC-chordal-tree-decommposition} outputs a function $f: \setofpartialMECs{G[r_1\cup N(r_1, G)]} \rightarrow \mathbb{Z}$ such that for $O \in \setofpartialMECs{G[r_1\cup N(r_1, G)]}$, $f(O) = |\setofMECs{G, O}|$.
\end{lemma}

\begin{algorithm}[H]
\caption{Count-MEC($G$)}
\label{alg:counting-MEC-chordal}
\SetAlgoLined
\SetKwInOut{KwIn}{Input}
\SetKwInOut{KwOut}{Output}
\SetKwFunction{MEC-Construction}{MEC-Construction}
\KwIn{A connected chordal graph $G$}
    \KwOut{ $|MEC(G)|$}
    
   Construct a clique tree representation $T$ of $G$. \label{alg:counting-MEC-chordal-construct-clique-tree}
   
   Pick an arbitrary node $r_1$ of $T$ as root node. \label{alg:counting-MEC-chordal-pick-r_1}
   
   $f \leftarrow$ Count$(G,T,r_1)$ \label{alg:counting-MEC-chordal-f-returns}
   
   $sum\leftarrow 0$ \label{alg:counting-MEC-chordal-sum-init}
   
   \ForEach{$O\in \setofpartialMECs{G[r_1\cup N(r_1,G)]}$ \label{alg:counting-MEC-chordal-for-each-start}}
   {
        sum=sum+$f[O]$ \label{alg:counting-MEC-chordal-sum-update}
   }\label{alg:counting-MEC-chordal-for-each-end}
   
   \KwRet sum; \label{alg:counting-MEC-chordal-return-statement}
\end{algorithm} 
We use \cref{alg:counting-MEC-chordal-tree-decommposition,eq:summation-of-PMECs} to construct \cref{alg:counting-MEC-chordal} that for an input chordal graph $G$, outputs $|\setofMECs{G}|$. \Cref{lem:proof-of-correctness-of-count-MEC-algorithm-2} provides the correctness of \cref{alg:counting-MEC-chordal}.

\begin{lemma}
    \label{lem:proof-of-correctness-of-count-MEC-algorithm-2}
    Let $G$ be a chordal graph.  For input $G$, \cref{alg:counting-MEC-chordal} outputs  $|\setofMECs{G}|$.
\end{lemma}

\section{Time Complexity}
\label{sec:time-complexity}
\begin{lemma}
    \label{lem:time-complexity-of-tree}
    Let $G$ be a tree graph. For input graph $G$, and a node $r\in V_{G}$, the time complexity of \cref{alg:counting-MEC-of-tree} is $O(d^2n)$, where $d$ is the degree of $G$, and $n$ is the number of nodes of $G$.
\end{lemma}

\begin{lemma}
    \label{lem:time-complexity-of-chordal-graph-tree-decomposition}
    Let $G$ be a chordal graph, $T$ be a tree decomposition of $G$, and $r$ be a node of $T$. For input $G$, $T$, and $r$, the time complexity of \cref{alg:counting-MEC-chordal-tree-decommposition} is $O(n(2^{O(d^2k^2)} + n^2))$, where $d$ is the degree of graph, $k$ is the treewidth of $T$, and $n$ is the number of nodes in $G$.
\end{lemma}

\begin{lemma}
    \label{lem:time-complexity-of-chordal-graph}
    For input chordal graph $G$, the time complexity to run \cref{alg:counting-MEC-chordal} is $O(n(2^{O(d^2k^2)} + n^2))$, where $d$ is the degree of graph, $k$ is the treewidth of $T$, and $n$ is the number of nodes in $G$.
\end{lemma}

\begin{theorem}
    \label{thm:counting-MECs-for-tree-graph}
    For a tree graph $G$, we can count the MECs of $G$ in $O(d^2n)$ time, where $n$ is the number of nodes in $G$, and $d$ is the degree of $G$.
\end{theorem}

\begin{theorem}
    \label{thm:counting-MECs-for-chordal-graph}
    For a chordal graph $G$, there exists a fixed parameter tractable algorithm that can count the MECs of $G$ in $O(n(2^{O(d^2k^2)}+n^2))$, where $n$ is the number of nodes of $G$, and the parameters $d$ and $k$ are, respectively, the degree and the treewidth of $G$. 
\end{theorem}

 \medskip

\bibliography{biblio.bib}
\bibliographystyle{abbrvnat}

\newpage
\section{Supplementary Material}
\subsection{Omitted proofs of \cref{subsection:tree}}
\label{subsection:omitted-proof-of-tree}
\begin{proof}[Proof of \cref{lem:recursive-method-to-count-MECs-of-tree}]
From \cref{obs:there-exists-a-unique-projection}, for any MEC $M$ of $G$, $\mathcal{P}(M, V_{G_1}, V_{G_2})$ is unique. 
For each pair of MECs $(M_1, M_2) \in \setofMECs{G_1}\times \setofMECs{G_2}$, we find the MECs $M$ of $G$ such that $\mathcal{P}(M, V_{G_1}, V_{G_2}) = (M_1, M_2)$.
We then classify these MECs $M$ into $\setofMECs{G, r_1, 1}$, $\setofMECs{G, r_1, 0, 0}, \setofMECs{G, r_1, 0, 1}, \ldots , \setofMECs{G, r_1, 0, \delta}$ (recall $\delta$ is the degree of $r_1$ in $G$).
We further count the MECs in each class. This gives us the required numbers $n_1(G, r_1)$, $n_0^0(G, r_1)$, $n_0^1(G, r_1)$, $\ldots$, $n_0^{\delta}(G, r_1)$.

For a pair of MECs $(M_1, M_2) \in \setofMECs{G_1}\times \setofMECs{G_2}$, \cref{obs:v-structures-of-G-neither-in-M1-nor-in-M2} provides us with necessary and sufficient conditions for an MEC $M$ of $G$ such that $\mathcal{P}(M, V_{G_1}, V_{G_2}) = (M_1, M_2)$.

\begin{observation}
    \label{obs:v-structures-of-G-neither-in-M1-nor-in-M2}
    Let $M, M_1$ and $M_2$ be the MECs of $G, G_1$ and $G_2$, respectively. Then,  $\mathcal{P}(M, V_{G_1}, V_{G_2}) = (M_1, M_2)$ if, and only if,
    \begin{enumerate}
        \item
        \label{item-1-of-obs:v-structures-of-G-neither-in-M1-nor-in-M2}
        $M$ contains all the v-structures of $M_1$ and $M_2$, i.e., $\mathcal{V}(M_1) \cup \mathcal{V}(M_2) \subseteq \mathcal{V}(M)$.
        \item
        \label{item-2-of-obs:v-structures-of-G-neither-in-M1-nor-in-M2}
        Either all the v-structures of $M$ that are neither v-structures in $M_1$ nor in $M_2$ contain an edge $r_1\rightarrow r_2$, or all the v-structures of $M$ that are neither v-structures in $M_1$ nor in $M_2$ contain an edge $r_2\rightarrow r_1$.
    \end{enumerate}
\end{observation}
\begin{proof}
    We first prove $\leftarrow$ of \cref{obs:v-structures-of-G-neither-in-M1-nor-in-M2}.  Suppose $\mathcal{P}(M, V_{G_1}, V_{G_2}) = (M_1, M_2)$.
    Since $M_1$ is a projection of $M$ on $V_{G_1}$, from \cref{def:projection}, $\mathcal{V}(M_1) = \mathcal{V}(M[V_{G_1}])$. This implies $\mathcal{V}(M_1) \subseteq \mathcal{V}(M)$. Similarly, $\mathcal{V}(M_2) \subseteq \mathcal{V}(M)$. This proves \cref{item-1-of-obs:v-structures-of-G-neither-in-M1-nor-in-M2}.

    From the structure of $M$, each edge of $M$ has either both endpoints in $V_{G_1}$, or both endpoints in $V_{G_2}$, or the endpoints are $r_1$ and $r_2$. 
If both the edges of the v-structure have endpoints in $V_{G_1}$ then the v-structure must be in $M_1$ since $\mathcal{V}(M[V_{G_1}]) = \mathcal{V}(M_1)$. Similarly, if both the edges of the v-structure have endpoints in $V_{G_2}$ then the v-structure must be in $M_2$.
From the construction, there cannot be a v-structure in $M$ such that one edge with endpoints in $V_{G_1}$ and another edge with endpoints in $V_{G_2}$, as there is no vertex common in between $V_{G_1}$ and $V_{G_2}$. Thus, for a v-structure of $M$ that is neither in $M_1$ nor in $M_2$, the only option that remains is one edge of the v-structure has endpoints $r_1$ and $r_2$. Since the edges of a v-structure are directed edges, a v-structure of $M$ that is neither in $M_1$ nor in $M_2$ contains either $r_1 \rightarrow r_2$ or contains $r_2 \rightarrow r_1$.
Also, since $M$ cannot have both $r_1 \rightarrow r_2$ and $r_2\rightarrow r_1$. Thus, either all the v-structures of $M$ that are neither in $M_1$ nor in $M_2$ contain $r_1\rightarrow r_2$, or all the v-structures of $M$ that are neither in $M_1$ nor in $M_2$ contain $r_2\rightarrow r_1$. This proves \cref{item-2-of-obs:v-structures-of-G-neither-in-M1-nor-in-M2}. 
This completes the proof of $\leftarrow$ of \cref{obs:v-structures-of-G-neither-in-M1-nor-in-M2}.

We now prove $\rightarrow$ of \cref{obs:v-structures-of-G-neither-in-M1-nor-in-M2}. Suppose $M$ obeys \cref{item-1-of-obs:v-structures-of-G-neither-in-M1-nor-in-M2,item-2-of-obs:v-structures-of-G-neither-in-M1-nor-in-M2}. Since $V_{G_1} \cap V_{G_2} = \emptyset$, and both $r_1$ and $r_2$ are neither  in $V_{G_1}$ nor in $V_{G_2}$, from \cref{item-1-of-obs:v-structures-of-G-neither-in-M1-nor-in-M2,item-2-of-obs:v-structures-of-G-neither-in-M1-nor-in-M2} of \cref{obs:v-structures-of-G-neither-in-M1-nor-in-M2},  $\mathcal{V}(M[V_{G_1}]) = \mathcal{V}(M_1)$, and $\mathcal{V}(M[V_{G_2}]) = \mathcal{V}(M_2)$. Then, from \cref{def:projection}, $\mathcal{P}(M, V_{G_1}, V_{G_2}) = (M_1, M_2)$. 
\end{proof}

For each pair of MECs $(M_1, M_2) \in \setofMECs{G_1}\times \setofMECs{G_2}$, we now find MECs $M$ of $G$ such that $\mathcal{P}(M, V_{G_1}, V_{G_2}) = (M_1, M_2)$, i.e., $M$ obeys \cref{item-1-of-obs:v-structures-of-G-neither-in-M1-nor-in-M2,item-2-of-obs:v-structures-of-G-neither-in-M1-nor-in-M2} of \cref{obs:v-structures-of-G-neither-in-M1-nor-in-M2}.
$(M_1, M_2)$ can be of four kinds: either (a) $M_1 \in \setofMECs{G_1, r_1, 1}$ and  $M_2 \in \setofMECs{G_2, r_2, 1}$, or (b) $M_1 \in \setofMECs{G_1, r_1, 1}$ and for some $0\leq j \leq \delta_2$, $M_2 \in \setofMECs{G_2, r_2, 0, j}$ (recall that $\delta_2$ is the degree of $r_2$ in $G_2$), or (c) for some $0\leq i \leq \delta_1$, $M_1 \in \setofMECs{G_1, r_1, 0, i}$ and  $M_2 \in \setofMECs{G_2, r_2, 1}$ (recall that $\delta_1$ is the degree of $r_1$ in $G_1$), or (d) for some $0\leq i \leq \delta_1$, $M_1 \in \setofMECs{G_1, r_1, 0, i}$ and for some $0\leq j \leq \delta_2$, $M_2 \in \setofMECs{G_2, r_2, 0, j}$. \cref{lem:relation-between-v-structures-M1-M1} deals with the first possibility.
    
\begin{lemma}
\label{lem:relation-between-v-structures-M1-M1}
Let $M_1\in \setofMECs{G_1, r_1, 1}$, and $M_2\in \setofMECs{G_2, r_2, 1}$. Then, 
\begin{enumerate}
    \item 
        \label{item-1-of-lem:relation-between-v-structures-M1-M1}
        For any MEC $M$ of $G$, if $\mathcal{P}(M, V_{G_1}, V_{G_2}) = (M_1, M_2)$ then one of the following occurs:
        \begin{enumerate}
            \item
                \begin{equation}
                \label{eq:v-structure-1}
                    \mathcal{V}(M)=\mathcal{V}(M_1) \cup \mathcal{V}(M_2) \cup \bigcup_{x\rightarrow r_1 \in E_{M_1}}{x\rightarrow r_1\leftarrow r_2}
                \end{equation}
            \item 
                \begin{equation}
                \label{eq:v-structure-2}
                    \mathcal{V}(M)=\mathcal{V}(M_1) \cup \mathcal{V}(M_2) \cup \bigcup_{y\rightarrow r_2 \in E_{M_2}}{y\rightarrow r_2\leftarrow r_1}
                \end{equation}
        \end{enumerate}
    \item 
        \label{item-2-of-lem:relation-between-v-structures-M1-M1}
        There exists a unique MEC $M$ of $G$ that obeys \cref{eq:v-structure-1}. Also, $M$ belongs to $\setofMECs{G, r_1, 1}$, and $\mathcal{P}(M, V_{G_1}, V_{G_2}) = (M_1, M_2)$.
\item 
        \label{item-3-of-lem:relation-between-v-structures-M1-M1}
        There exists a unique MEC $M$ of $G$ that obeys \cref{eq:v-structure-2}. Also, $M$ belongs to $\setofMECs{G, r_1, 1}$, and $\mathcal{P}(M, V_{G_1}, V_{G_2}) = (M_1, M_2)$.
\end{enumerate}
\end{lemma}

\begin{proof}
We first prove \cref{item-1-of-lem:relation-between-v-structures-M1-M1} of \cref{lem:relation-between-v-structures-M1-M1}.
\begin{proof}[Proof of \cref{item-1-of-lem:relation-between-v-structures-M1-M1} of \cref{lem:relation-between-v-structures-M1-M1}]
    Let $M$ be an MEC of $G$ such that $\mathcal{P}(M, V_{G_1}, V_{G_2}) = (M_1, M_2)$. 
From \cref{obs:v-structures-of-G-neither-in-M1-nor-in-M2}, $M$ has the v-structures of $M_1$ and $M_2$, and additionally, it can only have the v-structures that involve the edge with endpoints $r_1$ and $r_2$.    
Since $M_1 \in \setofMECs{G_1, r_1, 1}$, there exists a node $x\in X$ such that $x\rightarrow r_1 \in M_1$. As $M_1$ is an MEC, from \cref{item-1-of-thm:nes-and-suf-cond-for-tree-graph-to-be-an-MEC} of \cref{thm:nes-and-suf-cond-for-tree-graph-to-be-an-MEC}, all the edges adjacent to $r_1$ in $M_1$ are directed. This implies that we can partition the set of neighbors of $r_1$ in $G_1$ into $X$ and $X'$ such that $X = \{x: x\rightarrow r_1 \in M_1\}$ and $X' = \{x': x'\leftarrow r_1 \in M_1\}$. Similarly, as $M_2\in \setofMECs{G_2, r_2, 1}$, we can partition the set of neighbors of $r_2$ in $G_2$ into $Y$ and $Y'$ such that $Y = \{y: y\rightarrow r_2 \in M_2\}$ and $Y' = \{y': y'\leftarrow r_2 \in M_2\}$.
According to \cref{lem:directed-edge-is-same-in-projected-MEC}, $M$ has edges $x\rightarrow r_1$ for $x\in X$, $r_1\rightarrow x'$ for $x' \in X'$, $y\rightarrow r_2$ for $y\in Y$, and $r_2\rightarrow y'$ for $y' \in Y'$.
Based on \cref{item-1-of-thm:nes-and-suf-cond-for-tree-graph-to-be-an-MEC} of \cref{thm:nes-and-suf-cond-for-tree-graph-to-be-an-MEC}, either $r_2\rightarrow r_1 \in E_M$ or $r_1\rightarrow r_2\in E_M$, otherwise, $M$ has an induced subgraph of the form $x\rightarrow r_1-r_2$ for $x\in X$, contradicting \cref{item-1-of-thm:nes-and-suf-cond-for-tree-graph-to-be-an-MEC} of \cref{thm:nes-and-suf-cond-for-tree-graph-to-be-an-MEC}. If $r_2\rightarrow r_1 \in E_M$, then, in addition to the v-structures of $M_1$ and $M_2$, $M$ has the v-structures of the form $r_2\rightarrow r_1 \leftarrow x$ for $x \in X$, i.e., $M$ obeys $\mathcal{V}(M)$ as stated in \cref{eq:v-structure-1}. Similarly, if $r_1\rightarrow r_2 \in E_M$, then $M$ obeys \cref{eq:v-structure-2}. This completes the proof of \cref{item-1-of-lem:relation-between-v-structures-M1-M1}.
\end{proof}

 We now prove \cref{item-2-of-lem:relation-between-v-structures-M1-M1,item-3-of-lem:relation-between-v-structures-M1-M1} of \cref{lem:relation-between-v-structures-M1-M1}.
As discussed in the introduction, there cannot be two distinct MECs with the same skeleton and the same set of v-structures. 
This means there can be at most one MEC of $G$ containing the v-structures in \cref{eq:v-structure-1}, and at most one MEC of $G$ containing the v-structures in \cref{eq:v-structure-2}. 
The uniqueness of MECs present in \cref{item-2-of-lem:relation-between-v-structures-M1-M1,item-3-of-lem:relation-between-v-structures-M1-M1} of \cref{lem:relation-between-v-structures-M1-M1} comes from this fact. 
 We will show that $M_x=M_1\cup M_2 \cup \{r_2\rightarrow r_1\}$ is the MEC of $G$ that obeys \cref{eq:v-structure-1}, and $M_y=M_1\cup M_2 \cup \{r_1\rightarrow r_2\}$ is the MEC of $G$ that obeys \cref{eq:v-structure-2}.
  Let's begin with $M_x$.

To prove \cref{item-2-of-lem:relation-between-v-structures-M1-M1}, we need to demonstrate that (a) $M_x$ obeys \cref{eq:v-structure-1}, (b) $M_x$ is an MEC of $G$, (c) $\mathcal{P}(M_x, V_{G_1}, V_{G_2}) = (M_1, M_2)$, and (d) $M_x \in \setofMECs{G, r_1, 1}$. Let's first establish that $M_x$ obeys \cref{eq:v-structure-1}.

From the construction of $M_x$, the induced subgraph of $M_x$ on $V_{G_1}$ is $M_1$. Therefore, $\mathcal{V}(M_x[V_{G_1}]) = \mathcal{V}(M_1)$. Similarly, $\mathcal{V}(M_x[V_{G_2}]) = \mathcal{V}(M_2)$. This implies that $M_x$ contains the v-structure of $M_1$ and $M_2$, and in addition, the v-structures it contains must have one edge $r_2\rightarrow r_1$. Considering the construction of $M_x$, if $x\rightarrow r_1 \in M$, then either $x = r_2$ or $x\rightarrow r_1 \in M_1$. Therefore, the v-structures of $M_x$ with one edge $r_2\rightarrow r_1$ must be of the form $x\rightarrow r_1 \leftarrow r_2$ for $x\rightarrow r_1 \in M_1$. This demonstrates that the v-structures contained in $M_x$ obey \cref{eq:v-structure-1}.

Now, let's verify that $M_x$ is an MEC of $G$. Given the construction of $M_x$, its skeleton is $G$, so we only need to establish that $M_x$ is an MEC. As the skeleton of $M_x$ is a tree graph, we demonstrate that it satisfies \cref{item-1-of-thm:nes-and-suf-cond-for-tree-graph-to-be-an-MEC,item-2-of-thm:nes-and-suf-cond-for-tree-graph-to-be-an-MEC} of \cref{thm:nes-and-suf-cond-for-tree-graph-to-be-an-MEC}.

 Suppose $M_x$ contains an induced subgraph of the form $a\rightarrow b -c$. From the construction of $M_x$, there are 3 possibilities: either (a) $a,b,c \in V_{G_1}$, or (b) $a,b,c \in V_{G_2}$, or (c) $a = r_2$, $b =r_1$, and $c \in V_{G_1}$. If $a, b, c \in V_{G_1}$ then from the construction of $M_x$, $a\rightarrow b -c \in M_1$. Since $M_1$ is an MEC of a tree graph $G_1$, this contradicts \cref{item-1-of-thm:nes-and-suf-cond-for-tree-graph-to-be-an-MEC} of \cref{thm:nes-and-suf-cond-for-tree-graph-to-be-an-MEC}. This implies that this case cannot occur. Similarly, there cannot be $a, b, c  \in V_{G_2}$. If $a = r_2$, $b =r_1$, and $c \in V_{G_1}$ then from the construction of $M_x$, $r_1-c \in M_1$. Since $M_1 \in \setofMECs{G_1, r_1, 1}$, there must exist a node $x$ such that $x\rightarrow r_1 \in M_1$. This implies $M_1$ contains an induced subgraph of the form $x\rightarrow r_1 - c \in M_1$. As seen earlier, this cannot occur as $M_1$ obeys \cref{item-1-of-thm:nes-and-suf-cond-for-tree-graph-to-be-an-MEC} of \cref{thm:nes-and-suf-cond-for-tree-graph-to-be-an-MEC}.  This implies $M_x$ cannot have any induced subgraph of the form $a\rightarrow b - c$.  This shows that $M_x$ obeys \cref{item-1-of-thm:nes-and-suf-cond-for-tree-graph-to-be-an-MEC} of \cref{thm:nes-and-suf-cond-for-tree-graph-to-be-an-MEC}.

We now show that each directed edge $u\rightarrow v$ in $M_x$ obeys \cref{item-2-of-thm:nes-and-suf-cond-for-tree-graph-to-be-an-MEC} of \cref{thm:nes-and-suf-cond-for-tree-graph-to-be-an-MEC}. Let $u\rightarrow v$ be a directed edge in $M_x$. From the construction of $M_x$, either $u\rightarrow v \in M_1$, or $u\rightarrow v \in M_2$, or $u =r_2$ and $v =r_1$. Since $M_1$ is an MEC of a tree graph $G_1$, from \cref{item-2-of-thm:nes-and-suf-cond-for-tree-graph-to-be-an-MEC} of \cref{thm:nes-and-suf-cond-for-tree-graph-to-be-an-MEC}, if $a\rightarrow b \in M_1$ then it must be part of an induced subgraph of the form either $w\rightarrow u \rightarrow v$ (similar to \cref{fig:strongly-protected-edge}.a), or $u\rightarrow v \leftarrow w$. From the construction of $M$, the induced subgraph of $M_1$ is an induced subgraph of $M$. Thus, in $M_x$, $u\rightarrow v$ obeys \cref{item-2-of-thm:nes-and-suf-cond-for-tree-graph-to-be-an-MEC} of \cref{thm:nes-and-suf-cond-for-tree-graph-to-be-an-MEC}. Similarly, if $u\rightarrow v \in M_2$ then in $M_x$, $u\rightarrow v$ obeys \cref{item-2-of-thm:nes-and-suf-cond-for-tree-graph-to-be-an-MEC} of \cref{thm:nes-and-suf-cond-for-tree-graph-to-be-an-MEC}. We now go through the final possibility when $u =r_2$ and $v =r_1$. Since $M_1 \in \setofMECs{G_1, r_1, 1}$, there exist a node $x\in V_{G_1}$ such that $x\rightarrow r_1 \in M_1$. From the construction of $M_x$, $x\rightarrow r_1\leftarrow r_2$ is an induced subgraph of $M_x$. This implies $r_2\rightarrow r_1$  also obeys \cref{item-2-of-thm:nes-and-suf-cond-for-tree-graph-to-be-an-MEC} of \cref{thm:nes-and-suf-cond-for-tree-graph-to-be-an-MEC}. This shows all the directed edges of $M_x$ obeys \cref{item-2-of-thm:nes-and-suf-cond-for-tree-graph-to-be-an-MEC} of \cref{thm:nes-and-suf-cond-for-tree-graph-to-be-an-MEC}.
The above discussion implies that $M_x$ is an MEC of $G$ that obeys \cref{eq:v-structure-1}. We now prove that $\mathcal{P}(M_x, V_{G_1}, V_{G_2}) = (M_1, M_2)$.

Since $M_x$ obeys \cref{item-1-of-obs:v-structures-of-G-neither-in-M1-nor-in-M2,item-2-of-obs:v-structures-of-G-neither-in-M1-nor-in-M2} of \cref{obs:v-structures-of-G-neither-in-M1-nor-in-M2}, therefore, from \cref{obs:v-structures-of-G-neither-in-M1-nor-in-M2}, $\mathcal{P}(M_x, V_{G_1}, V_{G_2}) = (M_1, M_2)$. We now prove the remaining part of \cref{item-2-of-lem:relation-between-v-structures-M1-M1}, that is $M_x \in \setofMECs{G, r_1, 1}$. The following observation proves this.

\begin{observation}
    \label{lem:MEC-with-projected-MEC-in-M1-is-in-M1}
    Let $M$ and $M_1$ be MECs of $G$ and $G_1$, respectively, such that $\mathcal{P}(M, V_{G_1}) = M_1$, and $M_1 \in \setofMECs{G_1, r_1, 1}$. Then, $M \in \setofMECs{G, r_1, 1}$.
\end{observation}
\begin{proof}
    Since $M_1 \in \setofMECs{G_1, r_1, 1}$, there exists an $x$ such that $x\rightarrow r_1 \in M_1$. Since $M_1$ is a projection of $M$, from \cref{lem:directed-edge-in-projection-implies-directed-edge-in-MEC}, $x\rightarrow r_1 \in M$. This implies $M \in \setofMECs{G, r_1, 1}$.
\end{proof}

This completes the proof of \cref{item-2-of-lem:relation-between-v-structures-M1-M1} of \cref{lem:relation-between-v-structures-M1-M1}. Similarly, we can prove \cref{item-3-of-lem:relation-between-v-structures-M1-M1} of \cref{lem:relation-between-v-structures-M1-M1}. This concludes the proof of \cref{lem:relation-between-v-structures-M1-M1}.
\end{proof}

\Cref{lem:relation-between-v-structures-M1-M1} implies the following corollary:

\begin{corollary}
\label{lem:counting-MECs-of-M1-M1}
For $M_1 \in \setofMECs{G_1, r_1, 1}$ and $M_2 \in \setofMECs{G_2, r_2, 1}$, the number of MECs $M$ of $G$ such that $\mathcal{P}(M, V_{G_1}, V_{G_2}) = (M_1, M_2)$ is two. Specifically, both MECs belong to $\setofMECs{G, r_1, 1}$.
\end{corollary}
\begin{proof}
From \cref{item-1-of-lem:relation-between-v-structures-M1-M1} of \cref{lem:relation-between-v-structures-M1-M1}, for any MEC $M$ of $G$, if $\mathcal{P}(M, V_{G_1}, V_{G_2}) = (M_1, M_2)$ then $M$ has to obey either \cref{eq:v-structure-1} or \cref{eq:v-structure-2}. 

From \cref{item-2-of-lem:relation-between-v-structures-M1-M1} of \cref{lem:relation-between-v-structures-M1-M1}, there exists a unique MEC that satisfies \cref{eq:v-structure-1}, belonging to $\setofMECs{G, r_1, 1}$. Additionally, from \cref{item-3-of-lem:relation-between-v-structures-M1-M1} of \cref{lem:relation-between-v-structures-M1-M1}, another unique MEC, following \cref{eq:v-structure-2}, also belongs to $\setofMECs{G, r_1, 1}$.

Hence, the total count of MECs is two, and both of them belong to $\setofMECs{G, r_1, 1}$.
\end{proof}

We now consider the second possibility where $M_1 \in \setofMECs{G_1, r_1, 1}$ and $M_2 \in \setofMECs{G_2, r_2, 0, j}$ for some $0 \leq j \leq \delta_2$.

\begin{lemma}
\label{lem:relation-between-v-structures-M1-M0}
Let $M_1 \in \setofMECs{G_1, r_1, 1}$ and for some $0 \leq j \leq \delta_2$, $M_2 \in \setofMECs{G_2, r_2, 0, j}$. Then:
\begin{enumerate}
    \item
    \label{item-1-of-lem:relation-between-v-structures-M1-M0}
    For any MEC $M$ of $G$, if $\mathcal{P}(M, V_{G_1}, V_{G_2}) = (M_1, M_2)$ then one of the following occurs:
\begin{enumerate}
    \item 
    \begin{equation}
    \label{eq:v-structure-M1-M0-2}
        \mathcal{V}(M)=\mathcal{V}(M_1) \cup \mathcal{V}(M_2) \cup \bigcup_{x\rightarrow r_1 \in E_{M_1}}{x\rightarrow r_1\leftarrow r_2}
    \end{equation}
    \item
    \begin{equation}
        \label{eq:v-structure-M1-M0-1}
        \mathcal{V}(M)=\mathcal{V}(M_1) \cup \mathcal{V}(M_2)
    \end{equation}
    \item For some undirected edge $y-r_2\in E_{M_2}$,
    \begin{equation}
    \label{eq:v-structure-M1-M0-3}
        \mathcal{V}(M)=\mathcal{V}(M_1) \cup \mathcal{V}(M_2) \cup \{r_1\rightarrow r_2\leftarrow y\}
    \end{equation}
\end{enumerate}
\item
\label{item-2-of-lem:relation-between-v-structures-M1-M0}
There exists a unique MEC $M$ of $G$ that obeys \cref{eq:v-structure-M1-M0-2}. Additionally,  $M$ belongs to $\setofMECs{G, r_1, 1}$, and $\mathcal{P}(M, V_{G_1}, V_{G_2}) = (M_1, M_2)$.

\item
\label{item-3-of-lem:relation-between-v-structures-M1-M0}
There exists a unique MEC $M$ of $G$ that obeys \cref{eq:v-structure-M1-M0-1}. Furthermore, $M$ belongs to $\setofMECs{G, r_1, 1}$, and $\mathcal{P}(M, V_{G_1}, V_{G_2}) = (M_1, M_2)$.

\item
\label{item-4-of-lem:relation-between-v-structures-M1-M0}
For each $y$ such that $y-r_2 \in E_{M_2}$, there exists a unique MEC $M$ of $G$ that obeys \cref{eq:v-structure-M1-M0-3}. Also, $M$ belongs to $\setofMECs{G, r_1, 1}$, and $\mathcal{P}(M, V_{G_1}, V_{G_2}) = (M_1, M_2)$.
\end{enumerate}
\end{lemma}

\begin{proof}
We first prove \cref{item-1-of-lem:relation-between-v-structures-M1-M0} of \cref{lem:relation-between-v-structures-M1-M0}.

\begin{proof}[Proof of \cref{item-1-of-lem:relation-between-v-structures-M1-M0} of \cref{lem:relation-between-v-structures-M1-M0}]
Let $M$ be an MEC of $G$ such that $\mathcal{P}(M, V_{G_1}, V_{G_2}) = (M_1, M_2)$.
    From \cref{obs:v-structures-of-G-neither-in-M1-nor-in-M2}, $M$ has the v-structures of $M_1$ and $M_2$. In addition to those, it can only have v-structures that involve an edge with endpoints $r_1$ and $r_2$.

Since $M_1 \in \setofMECs{G_1, r_1, 1}$, there exists an $x\in V_{G_1}$ such that $x\rightarrow r_1$ in $M_1$. Following \cref{item-1-of-thm:nes-and-suf-cond-for-tree-graph-to-be-an-MEC} of \cref{thm:nes-and-suf-cond-for-tree-graph-to-be-an-MEC}, there cannot be an induced subgraph $x\rightarrow r_1 - y$ in $M_1$. Therefore, each edge adjacent to $r_1$ in $M_1$ is directed. Consequently, we can partition the neighbors of $r_1$ in $G_1$ into $X$ and $X'$ such that $X = \{x: x\rightarrow r_1 \in M_1\}$ and $X' = \{x': r_1\rightarrow x' \in M_1\}$ with $X$ being non-empty.

From \cref{lem:directed-edge-in-projection-implies-directed-edge-in-MEC}, for $x\in X$, $x\rightarrow r_1 \in M$, and for $x' \in X'$, $r_1\rightarrow x' \in M$. Further, from \cref{item-1-of-thm:nes-and-suf-cond-for-tree-graph-to-be-an-MEC} of \cref{thm:nes-and-suf-cond-for-tree-graph-to-be-an-MEC}, either $r_2\rightarrow r_1 \in M$ or $r_1\rightarrow r_2 \in M$, otherwise, for some $x \in X$, $M$ contains an induced subgraph $x\rightarrow r_1-r_2$, contradicting \cref{item-1-of-thm:nes-and-suf-cond-for-tree-graph-to-be-an-MEC} of \cref{thm:nes-and-suf-cond-for-tree-graph-to-be-an-MEC}..

If $r_2 \rightarrow r_1 \in M$, then the v-structures with endpoints $r_1$ and $r_2$ are $x\rightarrow r_1\leftarrow r_2$ for each $x \in X$. In this case, the v-structures of $M$ obey \cref{eq:v-structure-M1-M0-2}.

Suppose $r_1\rightarrow r_2 \in M$. Since $M_2 \in \setofMECs{G_2, r_2, 0, j}$, there cannot be an edge $y\rightarrow r_2$ in $M_2$. This implies that we can partition the neighbors of $r_2$ in $G_2$ into $Y$ and $Y'$ such that $Y = \{y: y-r_2 \in M_2\}$ is a set of $j$ elements (from the definition of $\setofMECs{G_2, r_2, 0, j}$), and $Y' = \{y': r_2 \rightarrow y' \in M_2\}$ is a set of $\delta_2 - j$ elements.
Since $r_1\rightarrow r_2 \in M$, according to \cref{item-1-of-thm:nes-and-suf-cond-for-tree-graph-to-be-an-MEC} of \cref{thm:nes-and-suf-cond-for-tree-graph-to-be-an-MEC}, for $y \in Y$, either $y\rightarrow r_2 \in M$ or $r_2\rightarrow y \in M$, otherwise, $r_1\rightarrow r_2-y$ is an induced subgraph of $M$, contradicting \cref{item-1-of-thm:nes-and-suf-cond-for-tree-graph-to-be-an-MEC} of \cref{thm:nes-and-suf-cond-for-tree-graph-to-be-an-MEC}. 
If for two distinct $y_1, y_2 \in Y$,  $y_1\rightarrow r_2, y_2 \rightarrow r_2 \in M$, then we get a v-structure $y_1\rightarrow r_2 \leftarrow y_2$ in $M$ such the nodes of the v-structure are in $V_{G_2}$, but the v-structure is not in $M_2$, implying $\mathcal{V}(M[V_{G_2}]) \neq \mathcal{V}(M_2)$, which further implies $\mathcal{P}(M, V_{G_2}) \neq M_2$, a contradiction. This implies that there exists at most one $y \in Y$ such that $y\rightarrow r_2 \in M$. 

If there does not exist any $y \in Y$ such that $y\rightarrow r_2 \in M$, then there won't be any v-structure in $M$ with endpoints $r_1$ and $r_2$. In this case, $M$ only contains the v-structures of $M_1$ and $M_2$, and obeys \cref{eq:v-structure-M1-M0-1}. 

However, if there exists a $y\in Y$ such that $y\rightarrow r_2 \in M$, then $r_1\rightarrow r_2 \leftarrow y$ is the only v-structure with endpoints $r_1$ and $r_2$ in $M$. In this case, $M$ obeys \cref{eq:v-structure-M1-M0-3}. 

We show that in each possibility, $M$ obeys \cref{item-1-of-lem:relation-between-v-structures-M1-M0} of \cref{lem:relation-between-v-structures-M1-M0}. 
This completes the proof of \cref{item-1-of-lem:relation-between-v-structures-M1-M0} of \cref{lem:relation-between-v-structures-M1-M0}.
\end{proof}

As discussed in the introduction, there cannot be two distinct MECs with the same skeleton and the same set of v-structures. The uniqueness of the MECs present in \cref{item-2-of-lem:relation-between-v-structures-M1-M0,item-3-of-lem:relation-between-v-structures-M1-M0,item-4-of-lem:relation-between-v-structures-M1-M0} of \cref{lem:relation-between-v-structures-M1-M0} arises from this fact.

We now prove \cref{item-2-of-lem:relation-between-v-structures-M1-M0,item-3-of-lem:relation-between-v-structures-M1-M0,item-4-of-lem:relation-between-v-structures-M1-M0} of \cref{lem:relation-between-v-structures-M1-M0}.

\begin{proof}[Proof of \cref{item-2-of-lem:relation-between-v-structures-M1-M0} of \cref{lem:relation-between-v-structures-M1-M0}]
    We construct an MEC of $G$ satisfying \cref{item-2-of-lem:relation-between-v-structures-M1-M0} of \cref{lem:relation-between-v-structures-M1-M0}. Let $M = M_1\cup M_2 \cup \{r_2\rightarrow r_1\}$. To verify that $M$ adheres to \cref{item-2-of-lem:relation-between-v-structures-M1-M0}, we need to demonstrate the following: (a) $M$ is an MEC of $G$, (b) $M$ satisfies \cref{eq:v-structure-M1-M0-2}, (c) $\mathcal{P}(M, V_{G_1}, V_{G_2}) = (M_1, M_2)$, and (d) $M \in \setofMECs{G, r_1, 1}$.

The proof follows a similar approach to that of verifying \cref{item-2-of-lem:relation-between-v-structures-M1-M1} in \cref{lem:relation-between-v-structures-M1-M1}. For completeness, we provide the proof.

Firstly, we establish that $M$ is an MEC of $G$. The construction of $M$ implies that its skeleton is identical to $G$. Hence, we only need to confirm that $M$ is an MEC, i.e., it satisfies \cref{item-1-of-thm:nes-and-suf-cond-for-tree-graph-to-be-an-MEC,item-2-of-thm:nes-and-suf-cond-for-tree-graph-to-be-an-MEC} of \cref{thm:nes-and-suf-cond-for-tree-graph-to-be-an-MEC}.

Suppose $M$ contains an induced subgraph of the form $a\rightarrow b -c$. There are three possible scenarios from the construction of $M$: (a) $a,b,c \in V_{G_1}$, (b) $a,b,c \in V_{G_2}$, or (c) $a = r_2$, $b =r_1$, and $c \in V_{G_1}$. If $a, b, c \in V_{G_1}$, then by construction, $a\rightarrow b -c \in M_1$. As $M_1$ is an MEC, this contradicts \cref{item-1-of-thm:nes-and-suf-cond-for-tree-graph-to-be-an-MEC} of \cref{thm:nes-and-suf-cond-for-tree-graph-to-be-an-MEC}. This scenario is impossible. Similarly, $a, b, c \in V_{G_2}$ is not viable.

Considering $a = r_2$, $b =r_1$, and $c \in V_{G_1}$, from the construction of $M$, $r_1-c \in M_1$. As $M_1 \in \setofMECs{G_1, r_1, 1}$, there must exist a node $x$ such that $x\rightarrow r_1 \in M_1$. This implies $M_1$ contains an induced subgraph of the form $x\rightarrow r_1 - c \in M_1$. However, this contradicts \cref{item-1-of-thm:nes-and-suf-cond-for-tree-graph-to-be-an-MEC} of \cref{thm:nes-and-suf-cond-for-tree-graph-to-be-an-MEC}, as $M_1$ is an MEC of a tree graph $G_1$. Consequently, $M$ cannot contain any induced subgraph of the form $a\rightarrow b - c$. Thus, $M$ satisfies \cref{item-1-of-thm:nes-and-suf-cond-for-tree-graph-to-be-an-MEC} of \cref{thm:nes-and-suf-cond-for-tree-graph-to-be-an-MEC}.

 We now demonstrate that each directed edge $u\rightarrow v$ in $M$ adheres to \cref{item-2-of-thm:nes-and-suf-cond-for-tree-graph-to-be-an-MEC} of \cref{thm:nes-and-suf-cond-for-tree-graph-to-be-an-MEC}. Let $u\rightarrow v$ be a directed edge in $M$. Based on the construction of $M$, it follows that either $u\rightarrow v \in M_1$, or $u\rightarrow v \in M_2$, or $u = r_2$ and $v = r_1$.

Considering $M_1$ as an MEC of tree graph $G_1$, if $u\rightarrow v \in M_1$, then according to \cref{item-2-of-thm:nes-and-suf-cond-for-tree-graph-to-be-an-MEC} of \cref{thm:nes-and-suf-cond-for-tree-graph-to-be-an-MEC}, $u\rightarrow v$ is part of an induced subgraph in $M_1$, such as either $w\rightarrow u \rightarrow v$ or $w\rightarrow v\leftarrow u$. As the induced subgraph of $M_1$ is also a subset of $M$, when $u\rightarrow v \in M_1$, $u\rightarrow v$ satisfies \cref{item-2-of-thm:nes-and-suf-cond-for-tree-graph-to-be-an-MEC} of \cref{thm:nes-and-suf-cond-for-tree-graph-to-be-an-MEC} for $M$.

Similarly, if $u\rightarrow v \in M_2$, then $u\rightarrow v$ satisfies \cref{item-2-of-thm:nes-and-suf-cond-for-tree-graph-to-be-an-MEC} of \cref{thm:nes-and-suf-cond-for-tree-graph-to-be-an-MEC} for $M$. Now, let's consider the scenario when $u = r_2$ and $v = r_1$. As $M_1 \in \setofMECs{G_1, r_1, 1}$, there exists a node $x\in V_{G_1}$ such that $x\rightarrow r_1 \in M_1$. With the construction of $M$, $x\rightarrow r_1\leftarrow r_2$ is an induced subgraph of $M$. Consequently, in this case, $u\rightarrow v$ also fulfills \cref{item-2-of-thm:nes-and-suf-cond-for-tree-graph-to-be-an-MEC} of \cref{thm:nes-and-suf-cond-for-tree-graph-to-be-an-MEC}.

Thus, the analysis shows that $M$ satisfies \cref{item-2-of-thm:nes-and-suf-cond-for-tree-graph-to-be-an-MEC} of \cref{thm:nes-and-suf-cond-for-tree-graph-to-be-an-MEC}. The aforementioned discussions imply that $M$ complies with \cref{item-1-of-thm:nes-and-suf-cond-for-tree-graph-to-be-an-MEC,item-2-of-thm:nes-and-suf-cond-for-tree-graph-to-be-an-MEC} of \cref{thm:nes-and-suf-cond-for-tree-graph-to-be-an-MEC}. Consequently, this verifies that $M$ is an MEC.

    Next, we demonstrate that the v-structures of $M$ adhere to \cref{eq:v-structure-M1-M0-2}. Considering the construction of $M$, $M_1$ and $M_2$ represent the induced subgraphs of $M$ on $V_{G_1}$ and $V_{G_2}$, respectively. Thus, $\mathcal{V}(M[V_{G_1}]) = \mathcal{V}(M_1)$ and $\mathcal{V}(M[V_{G_2}]) = \mathcal{V}(M_2)$. Since $r_2 \rightarrow r_1 \in M$, apart from the v-structures of $M_1$ and $M_2$, $M$ can possess v-structures involving the edge $r_2\rightarrow r_1$. Given $M_1 \in \setofMECs{G_1, r_1, 1}$, there exists a non-empty set $X = \{x: x\rightarrow r_1 \in M\}$. For each $x \in X$, $x\rightarrow r_1 \leftarrow r_2$ constitutes a v-structure in $M$. Thus, it is confirmed that the v-structures in $M$ satisfy \cref{eq:v-structure-M1-M0-2}.

Next, we verify $\mathcal{P}(M, V_{G_1}, V_{G_2}) = (M_1, M_2)$. As per the construction of $M$, $M[V_{G_1}] = M_1$, and $M[V_{G_2}] = M_2$. Consequently, $\mathcal{V}(M[V_{G_1}]) = \mathcal{V}(M_1)$ and $\mathcal{V}(M[V_{G_2}]) = \mathcal{V}(M_2)$. According to \cref{def:projection}, this further indicates $\mathcal{P}(M, V_{G_1}, V_{G_2}) = (M_1, M_2)$.

Referring to \cref{lem:MEC-with-projected-MEC-in-M1-is-in-M1}, we conclude that $M \in \setofMECs{G, r_1, 1}$. This completes the proof for \cref{item-2-of-lem:relation-between-v-structures-M1-M0} of \cref{lem:relation-between-v-structures-M1-M0}.
\end{proof}

\begin{proof}[Proof of \cref{item-3-of-lem:relation-between-v-structures-M1-M0} of \cref{lem:relation-between-v-structures-M1-M0}]
    We create an MEC $M$ of $G$ using the following steps:
\begin{enumerate}
    \item
    \label{item-2-of-M1-M0-1}
    Initialize $M = M_1\cup M_2\cup \{r_1\rightarrow r_2\}$.
    \item
    \label{item-2-of-M1-M0-2}
    Update $M$ by replacing $u-v$ in $M$ with $u\rightarrow v$ if $u-v \in M_2$ and there exists an undirected path from $r_2$ to $(u, v)$ in $M_2$.
\end{enumerate}

We demonstrate that $M$, constructed using the aforementioned steps, adheres to \cref{item-3-of-lem:relation-between-v-structures-M1-M0} of \cref{lem:relation-between-v-structures-M1-M0}. To prove this, we need to establish the following: (a) $M$ is an MEC of $G$, (b) $M$ adheres to \cref{eq:v-structure-M1-M0-1}, (c) $\mathcal{P}(M, V_{G_1}, V_{G_2}) = (M_1, M_2)$, and (d) $M \in \setofMECs{G, r_1, 1}$.

Firstly, we demonstrate that $M$ is an MEC, i.e., $M$ obeys \cref{item-1-of-thm:nes-and-suf-cond-for-tree-graph-to-be-an-MEC,item-2-of-thm:nes-and-suf-cond-for-tree-graph-to-be-an-MEC} of \cref{thm:nes-and-suf-cond-for-tree-graph-to-be-an-MEC}.

    Suppose $M$ contains an induced subgraph of the form $a\rightarrow b - c$. From the construction of $M$, there are three possibilities: either (a) $a,b,c \in V_{G_1}$, or (b) $a,b,c \in V_{G_2}$, or (c) $a = r_1$, $b =r_2$, and $c \in V_{G_2}$. 
    One by one, we show that none of the possibilities occurs. 
    
    From the construction of $M$, an induced subgraph of $M$ with vertices in $V_{G_1}$ is also an induced subgraph of $M_1$ (note that \cref{item-2-of-M1-M0-2} does not make any change in an edge with the endpoints in $V_{G_1}$). Since $M_1$ is an MEC of a tree graph $G_1$, from \cref{item-1-of-thm:nes-and-suf-cond-for-tree-graph-to-be-an-MEC} of 
\cref{thm:nes-and-suf-cond-for-tree-graph-to-be-an-MEC}, there cannot be an induced subgraph $u\rightarrow v-w \in M_1$. This implies that the first possibility cannot occur.

Suppose, $M$ contains an induced subgraph of the form $a\rightarrow b-c$ such that $a,b, c \in V_{G_2}$. From the construction of $M$, $b-c \in M_2$ (note that \cref{item-2-of-M1-M0-2} does not change a directed edge in $M_1\cup M_2 \cup \{r_1\rightarrow r_2\}$ into an undirected edge), and either $a\rightarrow b \in M_2$ or $a-b \in M_2$. If $a\rightarrow b \in M_2$ then $a\rightarrow b - c$ is an induced subgraph in the MEC $M_2$ of the tree graph $G_2$, contradicting \cref{item-1-of-thm:nes-and-suf-cond-for-tree-graph-to-be-an-MEC} of 
\cref{thm:nes-and-suf-cond-for-tree-graph-to-be-an-MEC}. And, if $a-b \in M_2$ then it must have been directed at \cref{item-2-of-M1-M0-2}. This implies there exists a path from $r_2$ to $(a,b)$ in $M_2$. This further implies that there exists a path from $r_2$ to $(b,c)$ in $M_2$. But, then, we have $b\rightarrow c \in M$ at \cref{item-2-of-M1-M0-2}. This shows that the second possibility also cannot occur. 

Suppose $r_1 \rightarrow r_2-c$ in $M$ such that $c\in V_{G_2}$. From the construction of $M$, this implies $r_2-c \in M_2$. But, edge $r_2-c$ is a path from $r_2$ to $(r_2,c)$. Then, it must have been directed at \cref{item-2-of-M1-M0-2}. This implies that the third possibility also cannot occur. This shows that $M$ cannot have an induced subgraph of the form $u\rightarrow v- w$. This implies that $M$ obeys \cref{item-1-of-thm:nes-and-suf-cond-for-tree-graph-to-be-an-MEC} of \cref{thm:nes-and-suf-cond-for-tree-graph-to-be-an-MEC}.

We now show that $M$ obeys \cref{item-2-of-thm:nes-and-suf-cond-for-tree-graph-to-be-an-MEC} of \cref{thm:nes-and-suf-cond-for-tree-graph-to-be-an-MEC}. Suppose $u\rightarrow v$ is a directed edge in $M$ then either $u,v \in V_{G_1}$, or $u,v \in V_{G_2}$, or $u=r_1$ and $v = r_2$. We show that in each of the possibilities $u\rightarrow v$ in $M$ obeys \cref{item-2-of-thm:nes-and-suf-cond-for-tree-graph-to-be-an-MEC} of \cref{thm:nes-and-suf-cond-for-tree-graph-to-be-an-MEC}.

Suppose $u,v \in V_{G_1}$. From the construction of $M$, $u\rightarrow v \in M_1$. Since $M_1$ is an MEC of a tree graph $G_1$, from \cref{item-2-of-thm:nes-and-suf-cond-for-tree-graph-to-be-an-MEC} of \cref{thm:nes-and-suf-cond-for-tree-graph-to-be-an-MEC}, $u\rightarrow v$ is part of an induced subgraph of $M_1$ of the form either $u\rightarrow v \leftarrow w$ or  $w\rightarrow u\rightarrow v$. Since from the construction of $M$, an induced subgraph of $M_1$ is an induced subgraph of $M$. Therefore, $u\rightarrow v$ in $M$ obeys \cref{item-2-of-thm:nes-and-suf-cond-for-tree-graph-to-be-an-MEC} of \cref{thm:nes-and-suf-cond-for-tree-graph-to-be-an-MEC}. 

Suppose $u,v \in V_{G_2}$. From the construction of $M$, there are two cases: either $u\rightarrow v \in M_2$ or $u-v \in M_2$. 

Suppose $u\rightarrow v \in M_2$.
Since $M_2$ is an MEC of a tree graph $G_2$, from \cref{item-2-of-thm:nes-and-suf-cond-for-tree-graph-to-be-an-MEC} of \cref{thm:nes-and-suf-cond-for-tree-graph-to-be-an-MEC}, $u\rightarrow v$ is part of an induced subgraph of $M_2$ of the form either $u\rightarrow v \leftarrow w$ or  $w\rightarrow u\rightarrow v$. Since from the construction of $M$, a directed induced subgraph of $M_2$ is an induced subgraph of $M$. Therefore, in the case when $u\rightarrow v \in M_2$, $u\rightarrow v$ in $M$ obeys \cref{item-2-of-thm:nes-and-suf-cond-for-tree-graph-to-be-an-MEC} of \cref{thm:nes-and-suf-cond-for-tree-graph-to-be-an-MEC}. 

Suppose  $u-v \in M_2$ then it must have been directed in $M$ at \cref{item-2-of-M1-M0-2}. This implies there exists an undirected path from $r_2$ to $(u,v)$ in $M_2$. 
Let the path be $P = (x_0 = r_2, x_1, \ldots, x_{l-1} = u, x_{l} = v)$ for some $l \geq 1$. From the construction of $M$, for $0\leq i < l$, $x_i \rightarrow x_{i+1} \in M$ because $P_i = (x_0 = r_2, \ldots, x_i, x_{i+1})$ is an undirected path in $M_2$ from $r_2$ to $(x_i, x_{i+1})$. If $l > 1$ then $u\rightarrow v$ is part of an induced subgraph  $x_{l-2} \rightarrow x_{l-1} \rightarrow x_l$ in $M$.
And, if $l=1$ then $r_2 = u$, and $u\rightarrow v$ is part of an induced subgraph $r_1\rightarrow r_2\rightarrow v$ in $M$. 
Thus, we show that if $u,v \in V_{G_2}$ then $u\rightarrow v$ obeys \cref{item-2-of-thm:nes-and-suf-cond-for-tree-graph-to-be-an-MEC} of \cref{thm:nes-and-suf-cond-for-tree-graph-to-be-an-MEC}.

We now show that $r_1\rightarrow r_2$ in $M$ also obeys \cref{item-2-of-thm:nes-and-suf-cond-for-tree-graph-to-be-an-MEC} of \cref{thm:nes-and-suf-cond-for-tree-graph-to-be-an-MEC}.
Since $M_1 \in \setofMECs{G_1, r_1, 1}$, there must exist $x\rightarrow r_1 \in M_1$. From the construction of $M$, $x\rightarrow r_1 \in M$. Then, $r_1\rightarrow r_2$ is part of an induced subgraph of the form $x\rightarrow r_1\rightarrow r_2$. Thus, we show that each edge of $M$ obeys \cref{item-2-of-thm:nes-and-suf-cond-for-tree-graph-to-be-an-MEC} of \cref{thm:nes-and-suf-cond-for-tree-graph-to-be-an-MEC}.

The above discussion shows that $M$ is an MEC as it obeys \cref{item-1-of-thm:nes-and-suf-cond-for-tree-graph-to-be-an-MEC,item-2-of-thm:nes-and-suf-cond-for-tree-graph-to-be-an-MEC} of \cref{thm:nes-and-suf-cond-for-tree-graph-to-be-an-MEC}. From the construction of $M$, the skelton of $M$ is $G$. This implies $M$ is an MEC of $G$.

We now show that $M$ obeys \cref{eq:v-structure-M1-M0-1}. From the construction of $M$, directed edges of $M_1$ and $M_2$ are directed edges of $M$. This implies the v-structures of $M_1$ and $M_2$ are the v-structures of $M$, i.e., $\mathcal{V}(M_1) \cup \mathcal{V}(M_2) \subseteq \mathcal{V}(M)$. For the completeness of the proof, we show that if $u\rightarrow v \leftarrow w$ is a v-structure in $M$ then either it is a v-structure of $M_1$ or it is a v-structure of $M_2$, i.e., $\mathcal{V}(M) \subseteq \mathcal{V}(M_1) \cup \mathcal{V}(M_2)$. This implies $\mathcal{V}(M) = \mathcal{V}(M_1) \cup \mathcal{V}(M_2)$, i.e., $M$ obeys \cref{eq:v-structure-M1-M0-1}.

Suppose there exists a v-structure $u\rightarrow v \leftarrow w$ in $M$. We show that either $u\rightarrow v \leftarrow w$ is a v-structure in $M_1$, or it is a v-structure in $M_2$. From the construction of $M$, there are following possibilities: (a) $u,v,w \in V_{G_1}$, or (b) $u = r_1$, $v =r_2$, and $w$ is a neighbor of $r_2$ in $V_{G_2}$, or (c) $u,v,w \in V_{G_2}$. 

    From the construction of $M$, $M_1$ is an induced subgraph of $M$ (note that \cref{item-2-of-M1-M0-2} does not change the orientation of any edge of $M_1$). This implies that if $u,v,w \in V_{G_1}$ then $u\rightarrow v \leftarrow w$ is a v-structure in $M_1$. 
    
    Suppose $u = r_1$, $v =r_2$, and $w$ is a neighbor of $r_2$ in  $V_{G_2}$. We show that this possibility cannot occur. Since $M_2 \in \setofMECs{G_2, r_2, 0}$, either $r_2\rightarrow w \in M_2$ or $r_2-w \in M_2$. We show that none of them occurs.
 Suppose $r_2\rightarrow w \in M_2$.
    From the construction of $M$ (\cref{item-2-of-M1-M0-1,item-2-of-M1-M0-2}), if $r_2\rightarrow w \in M_2$ then $r_2\rightarrow w \in M$. But, we have $r_2\leftarrow w \in M$. This implies $r_2\rightarrow w \notin M_2$. 
Suppose $r_2-w \in M_2$.
    Then, from \cref{item-2-of-M1-M0-2}, $r_2 \rightarrow w \in M$. Again, we have $r_2\leftarrow w \in M$. This implies that this case also cannot occur.
    We now move to the final possibility.
    
    Suppose $u,v, w \in V_{G_2}$. From the construction of $M$, if $a\rightarrow b \in M_2$ then $a\rightarrow b \in M$. This implies either $u-v \in M_2$ or $u\rightarrow v \in M_2$, and either $v-w \in M_2$ or $v\leftarrow w \in M_2$. From \cref{item-1-of-thm:nes-and-suf-cond-for-tree-graph-to-be-an-MEC} of \cref{thm:nes-and-suf-cond-for-tree-graph-to-be-an-MEC}, either $u\rightarrow v\leftarrow w \in M_2$ or $u-v-w \in M_2$ (recall that $M_2$ is an MEC of a tree graph $G_2$). 
    If $u\rightarrow v\leftarrow w \in M_2$ then the v-structure is a v-structure of $M_2$. Suppose $u-v-w \in M_2$. Then, if $u\rightarrow v \in M$ then it must have been directed at \cref{item-2-of-M1-M0-2}, i.e., there exists an undirected path from $r_2$ to $(u,v)$.
    Since $u-v-w \in M_2$, the existence of a path from $r_2$ to $(u,v)$ in $M_2$ implies the existence of a path from $r_2$ to $(v,w)$ in $M_2$. Also, since $M_2$ is an MEC of a tree graph $G_2$, there cannot be two different paths from $r_2$ to $w$. 
    This implies at \cref{item-2-of-M1-M0-2}, if we have $u\rightarrow v \in M$ then we also have $v\rightarrow w\in M$. But, this is a contradiction, as we have assumed that $u\rightarrow v \leftarrow w \in M$.
    This implies there cannot be $u-v-w\in M_2$ and $u\rightarrow v\leftarrow w \in M$. Thus, we show that in all the possible scenarios, if $u\rightarrow v\leftarrow w$ is a v-structure in $M$ then it is a v-structure either in $M_1$ or in $M_2$. 
    This completes the proof that $M$ obeys \cref{eq:v-structure-M1-M0-1}.

    The above discussion implies that $\mathcal{V}(M[V_{G_1}] = \mathcal{V}(M_1)$ and $\mathcal{V}(M[V_{G_2}] = \mathcal{V}(M_2)$. Therefore, from \cref{def:projection}, $\mathcal{P}(M, V_{G_1}, V_{G_2}) = (M_1, M_2)$.

    From \cref{lem:MEC-with-projected-MEC-in-M1-is-in-M1}, $M \in \setofMECs{G, r_1, 1}$. This completes the proof of \cref{item-3-of-lem:relation-between-v-structures-M1-M0} of \cref{lem:relation-between-v-structures-M1-M0}.
\end{proof}

\begin{proof}[Proof of \cref{item-4-of-lem:relation-between-v-structures-M1-M0} of \cref{lem:relation-between-v-structures-M1-M0}]
     We construct an MEC of $G$ using the following steps:
    \begin{enumerate}
        \item
        \label{item-1-of-item-3-of-M1-M0}
        Initialize $M = M_1\cup M_2\cup \{r_1\rightarrow r_2\}$.
        \item
        \label{item-2-of-item-3-of-M1-M0}
        Pick an edge $y-r_2 \in M_2$. Replace the edge $y-r_2$ in $M$  with $y\rightarrow r_2$.
        \item
        \label{item-3-of-item-3-of-M1-M0}
        Update $M$ by replacing  $u-v$ of $M$ with $u\rightarrow v$, if $u-v \in M_2$ and there exists an undirected path from $r_2$ to $(u, v)$ in $M_2$ such that the path does not contain $y$.
    \end{enumerate}

    We show that $M$ constructed using the above steps will obey \cref{item-4-of-lem:relation-between-v-structures-M1-M0} of \cref{lem:relation-between-v-structures-M1-M0}. To prove this, we have to show the following: (a) $M$ is an MEC of $G$, (b) $M$ obeys \cref{eq:v-structure-M1-M0-3}, (c) $\mathcal{P}(M, V_{G_1}, V_{G_2}) = (M_1, M_2)$, and (d) $M \in \setofMECs{G, r_1, 1}$.

    We first show that $M$ is an MEC of $G$. From the construction of $M$, the skeleton of $M$ is $G$. Therefore, it would be sufficient to show that $M$ obeys \cref{item-1-of-thm:nes-and-suf-cond-for-tree-graph-to-be-an-MEC,item-2-of-thm:nes-and-suf-cond-for-tree-graph-to-be-an-MEC} of \cref{thm:nes-and-suf-cond-for-tree-graph-to-be-an-MEC}. 

    We first show that $M$ obeys \cref{item-1-of-thm:nes-and-suf-cond-for-tree-graph-to-be-an-MEC} of \cref{thm:nes-and-suf-cond-for-tree-graph-to-be-an-MEC}.  Suppose $M$ contains an induced subgraph of the form $a\rightarrow b - c$. From the construction of $M$, there are three possibilities: either (a) $a,b,c \in V_{G_1}$, or (b) $a,b,c \in V_{G_2}$, or (c) $a = r_1$, $b =r_2$, and $c \in V_{G_2}$. 
    One by one, we show that none of the possibilities occurs. 
    
    Suppose $a,b, c \in V_{G_1}$. From the construction of $M$, an induced subgraph of $M$ with vertices in $V_{G_1}$ is also an induced subgraph of $M_1$ (note that \cref{item-2-of-item-3-of-M1-M0,item-3-of-item-3-of-M1-M0} does not make any change in an edge with the endpoints in $V_{G_1}$). Since $M_1$ is an MEC of a tree graph $G_1$, from \cref{item-1-of-thm:nes-and-suf-cond-for-tree-graph-to-be-an-MEC} of \cref{thm:nes-and-suf-cond-for-tree-graph-to-be-an-MEC}, there cannot be an induced subgraph $u\rightarrow v-w \in M_1$. This implies that the first possibility cannot occur. 
    
Suppose $a,b, c \in V_{G_2}$. From the construction of $M$, $b-c \in M_2$ (note that \cref{item-2-of-item-3-of-M1-M0,item-3-of-item-3-of-M1-M0} do not change a directed edge in $M_1\cup M_2 \cup \{r_1\rightarrow r_2\}$ into an undirected edge), and either $a\rightarrow b \in M_2$ or $a-b \in M_2$. If $a\rightarrow b \in M_2$ then $a\rightarrow b - c$ is an induced subgraph in the MEC $M_2$, contradicting \cref{item-1-of-thm:nes-and-suf-cond-for-tree-graph-to-be-an-MEC} of \cref{thm:nes-and-suf-cond-for-tree-graph-to-be-an-MEC}. And, if $a-b \in M_2$ then it must have been directed at either at \cref{item-2-of-item-3-of-M1-M0} (first case) or at \cref{item-3-of-item-3-of-M1-M0} (second case). If it has directed at \cref{item-2-of-item-3-of-M1-M0} then $a = y, b = r_2$ and $c \neq y$ is a neighbor of $r_2$. This implies $r_2 -c$ is a path from $r_2$ to $(r_2, c)$. But, then at \cref{item-3-of-item-3-of-M1-M0}, $r_2-c$ has been converted into $r_2\rightarrow c$. This implies that the first case cannot occur. We move to the second case when $a-b$ has been directed at \cref{item-3-of-item-3-of-M1-M0}. This implies there exists a path from $r_2$ to $(a,b)$ in $M_2$. This further implies that there exists a path from $r_2$ to $(b,c)$ in $M_2$. But, then, we have $b\rightarrow c \in M$ at \cref{item-3-of-item-3-of-M1-M0}. This shows that the second possibility cannot occur.

Suppose  $a = r_1$, $b =r_2$, and $c \in V_{G_2}$. From the construction of $M$, this implies $r_2-c \in M_2$. There are two cases, either $c = y$, or $c \neq y$. If $c = y$, then at \cref{item-2-of-item-3-of-M1-M0}, we have $c\rightarrow b \in M$. Therefore, this case cannot occur. Suppose $c\neq y$.  But, edge $r_2-c$ is a path from $r_2$ to $(r_2,c)$. Then, it must have been directed at \cref{item-3-of-item-3-of-M1-M0}. This implies that the third possibility also cannot occur. Thus, we show that $M$ cannot have an induced subgraph of the form $u\rightarrow v- w$. This implies that $M$ obeys \cref{item-1-of-thm:nes-and-suf-cond-for-tree-graph-to-be-an-MEC} of \cref{thm:nes-and-suf-cond-for-tree-graph-to-be-an-MEC}.

 We now show that $M$ obeys \cref{item-2-of-thm:nes-and-suf-cond-for-tree-graph-to-be-an-MEC} of \cref{thm:nes-and-suf-cond-for-tree-graph-to-be-an-MEC}. Suppose $u\rightarrow v$ is a directed edge in $M$ then either $u,v \in V_{G_1}$, or $u,v \in V_{G_2}$, or $u=r_1$ and $v = r_2$. We show that in each of the possibilities $u\rightarrow v$ in $M$ obeys \cref{item-2-of-thm:nes-and-suf-cond-for-tree-graph-to-be-an-MEC} of \cref{thm:nes-and-suf-cond-for-tree-graph-to-be-an-MEC}.

Suppose $u,v \in V_{G_1}$. From the construction of $M$, $M_1$ is an induced subgraph of $M$ (note that \cref{item-2-of-item-3-of-M1-M0,item-3-of-item-3-of-M1-M0} do not change the orientation of any edge of $M_1$). This implies $u\rightarrow v \in M_1$. Since $M_1$ is an MEC of a tree graph $G_1$, from \cref{item-2-of-thm:nes-and-suf-cond-for-tree-graph-to-be-an-MEC} of \cref{thm:nes-and-suf-cond-for-tree-graph-to-be-an-MEC}, $u\rightarrow v$ is part of an induced subgraph of $M_1$ of the form either $u\rightarrow v \leftarrow w$ or  $w\rightarrow u\rightarrow v$. Since $M_1$ is an induced subgraph of $M$, an induced subgraph of $M_1$ is an induced subgraph of $M$. Therefore, $u\rightarrow v$ in $M$ obeys \cref{item-2-of-thm:nes-and-suf-cond-for-tree-graph-to-be-an-MEC} of \cref{thm:nes-and-suf-cond-for-tree-graph-to-be-an-MEC}. 

Suppose $u,v \in V_{G_2}$. From the construction of $M$, there are two cases: either $u\rightarrow v \in M_2$ or $u-v \in M_2$. Suppose $u\rightarrow v \in M_2$. Since $M_2$ is an MEC of a tree graph $G_2$, from \cref{item-2-of-thm:nes-and-suf-cond-for-tree-graph-to-be-an-MEC} of \cref{thm:nes-and-suf-cond-for-tree-graph-to-be-an-MEC}, $u\rightarrow v$ is part of an induced subgraph of $M_2$ of the form either $u\rightarrow v \leftarrow w$ or  $w\rightarrow u\rightarrow v$.  From the construction of $M$, a directed induced subgraph of $M_2$ is an induced subgraph of $M$ (as a directed edge of $M_2$ is a directed edge of $M$, and the skeletons of $M_2$ and $M[V_{G_2}$ are the same).  Therefore, in this case also $u\rightarrow v$ in $M$ obeys \cref{item-2-of-thm:nes-and-suf-cond-for-tree-graph-to-be-an-MEC} of \cref{thm:nes-and-suf-cond-for-tree-graph-to-be-an-MEC}.

Suppose $u-v \in M_2$. Then, it must have been directed in $M$ either at \cref{item-2-of-item-3-of-M1-M0} or at \cref{item-3-of-item-3-of-M1-M0}. Suppose $u-v \in M_2$ and it gets directed in $M$ at \cref{item-2-of-item-3-of-M1-M0}. Then, $u = y$ and $v = r_2$. In that case, $u\rightarrow v$ is part of an induced subgraph $r_1\rightarrow r_2 \leftarrow y$, obeying \cref{item-2-of-thm:nes-and-suf-cond-for-tree-graph-to-be-an-MEC} of \cref{thm:nes-and-suf-cond-for-tree-graph-to-be-an-MEC}. Suppose $u-v$ gets directed in $M$ at \cref{item-3-of-item-3-of-M1-M0}. This implies there exists an undirected path from $r_2$ to $(u,v)$ in $M_2$, and the path does not contain $y$. 
Let the path be $P = (x_0 = r_2, x_1, \ldots, x_{l-1} = u, x_{l} = v)$ for some $l \geq 1$. From the construction of $M$, for $0\leq i < l$, $x_i \rightarrow x_{i+1} \in M$ because $P_i = (x_0 = r_2, \ldots, x_i, x_{i+1})$ is an undirected path in $M_2$ from $r_2$ to $(x_i, x_{i+1})$, and $P_i$ does not contain $y$ either. If $l > 1$ then $u\rightarrow v$ is part of an induced subgraph  $x_{l-2} \rightarrow x_{l-1} \rightarrow x_l$ in $M$.
And, if $l=1$ then $r_2 = u$, and $u\rightarrow v$ is part of an induced subgraph $r_1\rightarrow r_2\rightarrow v$ in $M$.
 Thus, we show that if $u,v \in V_{G_2}$ then $u\rightarrow v$ obeys \cref{item-2-of-thm:nes-and-suf-cond-for-tree-graph-to-be-an-MEC} of \cref{thm:nes-and-suf-cond-for-tree-graph-to-be-an-MEC}.
 
 We now show that $r_1\rightarrow r_2$ in $M$ also obeys \cref{item-2-of-thm:nes-and-suf-cond-for-tree-graph-to-be-an-MEC} of \cref{thm:nes-and-suf-cond-for-tree-graph-to-be-an-MEC}. $r_1\rightarrow r_2$ is part of an induced subgraph $r_1\rightarrow
  r_2 \leftarrow y$ in $M$. Thus, we show that each edge of $M$ obeys \cref{item-2-of-thm:nes-and-suf-cond-for-tree-graph-to-be-an-MEC} of \cref{thm:nes-and-suf-cond-for-tree-graph-to-be-an-MEC}.

The above discussion shows that $M$ obeys \cref{item-1-of-thm:nes-and-suf-cond-for-tree-graph-to-be-an-MEC,item-2-of-thm:nes-and-suf-cond-for-tree-graph-to-be-an-MEC} of \cref{thm:nes-and-suf-cond-for-tree-graph-to-be-an-MEC}. This implies $M$ is an MEC of $G$.

We now show that $M$ obeys \cref{eq:v-structure-M1-M0-3}. From the construction of $M$, directed edges of $M_1$ and $M_2$ are directed edges of $M$. Also, $r_1\rightarrow r_2 \leftarrow y$ is a v-structure in $M$. This implies  $\mathcal{V}(M_1) \cup \mathcal{V}(M_2) \cup \{r_1\rightarrow r_2 \leftarrow y\} \subseteq \mathcal{V}(M)$. For the completeness of the proof, we show that if $u\rightarrow v \leftarrow w$ is a v-structure in $M$ then either it is a v-structure of $M_1$ or it is a v-structure of $M_2$, or $u= r_1$, $v = r_2$ and $w = y$, i.e., $\mathcal{V}(M) \subseteq \mathcal{V}(M_1) \cup \mathcal{V}(M_2) \cup \{r_1\rightarrow r_2 \leftarrow y\}$. This further implies $\mathcal{V}(M) = \mathcal{V}(M_1) \cup \mathcal{V}(M_2)\cup \{r_1\rightarrow r_2 \leftarrow y\}$, i.e., $M$ obeys \cref{eq:v-structure-M1-M0-3}.

Suppose there exists a v-structure $u\rightarrow v \leftarrow w$ in $M$. From the construction of $M$, there are following possibilities: (a) $u,v,w \in V_{G_1}$, or (b) $u = r_1$, $v =r_2$, and $w = y$ is a neighbor of $r_2$ in $V_{G_2}$, or (c) $u,v,w \in V_{G_2}$. 

    Suppose $u,v, w \in V_{G_1}$. From the construction of $M$, $M_1$ is an induced subgraph of $M$. This implies that if $u,v,w \in V_{G_1}$ then $u\rightarrow v \leftarrow w$ is a v-structure in $M_1$. 
    
    Suppose $u = r_1$, $v =r_2$, and $w$ is a neighbor of $r_2$ in  $V_{G_2}$. Since $M_2 \in \setofMECs{G_2, r_2, 0}$, either $r_2\rightarrow w \in M_2$ or $r_2-w \in M_2$. From the construction of $M$ (\cref{item-1-of-item-3-of-M1-M0,item-2-of-item-3-of-M1-M0,item-3-of-item-3-of-M1-M0}), if $r_2\rightarrow w \in M_2$ then $r_2\rightarrow w \in M$. But, we have $r_2\leftarrow w \in M$. This implies that $r_2-w \in M_2$. From \cref{item-2-of-item-3-of-M1-M0,item-3-of-item-3-of-M1-M0}, if we have $w \rightarrow r_2 \in M$ then $w = y$, otherwise, we have $r_2\rightarrow w \in M$ at \cref{item-3-of-item-3-of-M1-M0} as $r_2-w$ is an undirected path from $r_2$ to $(r_2, w)$, and the path does not contain $y$. This implies $u = r_1$, $v = r_2$, and $w = y$.  
    
    Suppose $u,v, w \in V_{G_2}$. From the construction of $M$, if $a\rightarrow b \in M_2$ then $a\rightarrow b \in M$. This implies either $u-v \in M_2$ or $u\rightarrow v \in M_2$, and either $v-w \in M_2$ or $v\leftarrow w \in M_2$. Since $M_2$ is an MEC of a tree graph $G_2$, from \cref{item-1-of-thm:nes-and-suf-cond-for-tree-graph-to-be-an-MEC} of \cref{thm:nes-and-suf-cond-for-tree-graph-to-be-an-MEC}, either $u\rightarrow v\leftarrow w \in M_2$ or $u-v-w \in M_2$ (recall that $M_2$ is an MEC). 
    If $u\rightarrow v\leftarrow w \in M_2$ then the v-structure is a v-structure of $M_2$. Suppose $u-v-w \in M_2$. Then, if $u\rightarrow v \in M$ then it must have been directed either at \cref{item-2-of-item-3-of-M1-M0} or at \cref{item-2-of-item-3-of-M1-M0}. Suppose $u\rightarrow v \in M$ at \cref{item-2-of-item-3-of-M1-M0}. Then, $u = y$ and $v =r_2$. But, then $v-w$ is a path from $(r_2)$ to $(v, w)$ (and the path does not contain $y$). And, we have $v\rightarrow w \in M$ at \cref{item-3-of-item-3-of-M1-M0}, a contradiction, as the v-structure implies $w\rightarrow v \in M$. This implies that this case cannot occur. Suppose  $u\rightarrow v \in M$ at \cref{item-3-of-item-3-of-M1-M0}.  Then, there exists a path from $r_2$ to $(u,v)$, and the path does not contain $y$.
    Since $u-v-w \in M_2$, the existence of a path from $r_2$ to $(u,v)$ in $M_2$ implies the existence of a path from $r_2$ to $(v,w)$ in $M_2$, and the new path also does not contain $y$. Also, since $M_2$ is an MEC of a tree graph $G_2$, there cannot be two different paths from $r_2$ to $w$. 
    This implies at \cref{item-3-of-item-3-of-M1-M0}, if we have $u\rightarrow v \in M$ then we also have $v\rightarrow w\in M$. But, this is a contradiction, as we have assumed that $u\rightarrow v \leftarrow w \in M$.
    This implies there cannot be $u-v-w\in M_2$ and $u\rightarrow v\leftarrow w \in M$. Thus, we show that in all the possible scenarios, if $u\rightarrow v\leftarrow w$ is a v-structure in $M$ then it is a v-structure either in $M_1$ or in $M_2$, or $u = r_1$, $v =r_2$ and $w =  y$. 
    This completes the proof that $M$ obeys \cref{eq:v-structure-M1-M0-3}.

    The above discussion implies that $\mathcal{V}(M[V_{G_1}]) = \mathcal{V}(M_1)$ and $\mathcal{V}(M[V_{G_2}]) = \mathcal{V}(M_2)$. Therefore, from \cref{def:projection}, $\mathcal{P}(M, V_{G_1}, V_{G_2}) = (M_1, M_2)$.

    From \cref{lem:MEC-with-projected-MEC-in-M1-is-in-M1}, $M \in \setofMECs{G, r_1, 1}$. This completes the proof of \cref{item-4-of-lem:relation-between-v-structures-M1-M0} of \cref{lem:relation-between-v-structures-M1-M0}. 
\end{proof}
This completes the proof of \cref{lem:relation-between-v-structures-M1-M0}.
\end{proof}

\Cref{lem:relation-between-v-structures-M1-M0} implies the following corollary:

\begin{corollary}
\label{lem:counting-MECs-of-M1-M0}
For $M_1\in \setofMECs{G_1,r_1, 1}$, and $M_2 \in \setofMECs{G_2,r_2, 0, j}$,
the number of MECs $M$ of $G$ such that $\mathcal{P}(M, V_{G_1}, V_{G_2})= (M_1, M_2)$ is j+ 2. More specifically, all the MECs belong to $\setofMECs{G,r_1, 1}$.
\end{corollary}
\begin{proof}
From \cref{item-1-of-lem:relation-between-v-structures-M1-M0} of \cref{lem:relation-between-v-structures-M1-M0}, for any MEC $M$ of $G$, if $\mathcal{P}(M, V_{G_1}, V_{G_2}) = (M_1, M_2)$ then $M$ has to obey either \cref{eq:v-structure-M1-M0-2} or \cref{eq:v-structure-M1-M0-1} or \cref{eq:v-structure-M1-M0-3}. 

From \cref{item-2-of-lem:relation-between-v-structures-M1-M0} of \cref{lem:relation-between-v-structures-M1-M0}, there exists a unique MEC that satisfies \cref{eq:v-structure-M1-M0-2}, belonging to $\setofMECs{G, r_1, 1}$. Additionally, from \cref{item-3-of-lem:relation-between-v-structures-M1-M0} of \cref{lem:relation-between-v-structures-M1-M0}, another unique MEC, following \cref{eq:v-structure-M1-M0-1}, also belongs to $\setofMECs{G, r_1, 1}$. Furthermore, given $M_2 \in \setofMECs{G_2, r_2, 0, j}$, $j$ instances of $y$ exist such that $y-r_2 \in E_{M_2}$. According to \cref{item-4-of-lem:relation-between-v-structures-M1-M0} of \cref{lem:relation-between-v-structures-M1-M0}, for each $y-r_2 \in E_{M_2}$, there exists a unique MEC satisfying \cref{eq:v-structure-M1-M0-3}. Consequently, there are $j$ MECs following \cref{eq:v-structure-M1-M0-3}, all belong to $\setofMECs{G,r_1, 1}$ as per \cref{item-4-of-lem:relation-between-v-structures-M1-M0}.

Hence, the total count of MECs is $j+2$, and all of them belong to $\setofMECs{G,r_1, 1}$.
\end{proof}

We now deal with the third possibility when for some $0\leq i \leq \delta_1$, $M_1 \in \setofMECs{G_1, r_1, 0, i}$, and $M_2 \in \setofMECs{G_2, r_2, 1}$. 
\begin{lemma}
\label{lem:relation-between-v-structures-M0-M1}
Let for some $0\leq i \leq \delta_1$, $M_1 \in \setofMECs{G_1, r_1, 0, i}$, and $M_2 \in \setofMECs{G_2, r_2, 1}$. Then,
\begin{enumerate}
    \item
    \label{item-1-of-lem:relation-between-v-structures-M0-M1}
    For any $M \in$ MEC$(G)$, if $\mathcal{P}(M,V_{G_1},V_{G_2})=(M_1,M_2)$ then one of the following occurs:
\begin{enumerate}
    \item
    \begin{equation}
    \label{eq:v-structure-M0-M1-1}
        \mathcal{V}(M)=\mathcal{V}(M_1) \cup \mathcal{V}(M_2) \cup \bigcup_{y\rightarrow r_2 \in E_{M_2}}{y\rightarrow r_2\leftarrow r_1}
    \end{equation}

        \item 
    \begin{equation}
        \label{eq:v-structure-M0-M1-2}
        \mathcal{V}(M)=\mathcal{V}(M_1) \cup \mathcal{V}(M_2)
    \end{equation}
    
    \item For some undirected edge $x-r_1\in E_{M_1}$,
    \begin{equation}
    \label{eq:v-structure-M0-M1-3}
        \mathcal{V}(M)=\mathcal{V}(M_1) \cup \mathcal{V}(M_2) \cup \{x\rightarrow r_1\leftarrow r_2\}
    \end{equation}
\end{enumerate}
\item 
\label{item-2-of-lem:relation-between-v-structures-M0-M1}
There exists a unique MEC $M$ of $G$ that obeys \cref{eq:v-structure-M0-M1-1}. Additionally, $M$ belongs to $\setofMECs{G, r_1, 0, i}$, and $\mathcal{P}(M, V_{G_1}, V_{G_2}) = (M_1, M_2)$.

\item
\label{item-3-of-lem:relation-between-v-structures-M0-M1}
There exists a unique MEC $M$ of $G$ that obeys \cref{eq:v-structure-M0-M1-2}. Furthermore, $M$ belongs to $\setofMECs{G, r_1, 1}$, and $\mathcal{P}(M, V_{G_1}, V_{G_2}) = (M_1, M_2)$.

\item
\label{item-4-of-lem:relation-between-v-structures-M0-M1}
For each $x$ such that $x-r_2 \in E_{M_1}$, there exists a unique MEC $M$ of $G$ that obeys \cref{eq:v-structure-M0-M1-3}. Also, $M$ belongs to $\setofMECs{G, r_1, 1}$, and $\mathcal{P}(M, V_{G_1}, V_{G_2}) = (M_1, M_2)$.
\end{enumerate}
\end{lemma}

\begin{proof}
We first prove \cref{item-1-of-lem:relation-between-v-structures-M0-M1} of \cref{lem:relation-between-v-structures-M0-M1}.
\begin{proof}[Proof of \cref{item-1-of-lem:relation-between-v-structures-M0-M1} of \cref{lem:relation-between-v-structures-M0-M1}]
Let $M$ be an MEC of $G$ such that $\mathcal{P}(M, V_{G_1}, V_{G_2}) = (M_1, M_2)$.
    From \cref{obs:v-structures-of-G-neither-in-M1-nor-in-M2}, $M$ has the v-structures of $M_1$ and $M_2$. Other than that, it can only have v-structures that involve an edge with endpoints $r_1$ and $r_2$. 

    Since $M_2 \in \setofMECs{G_2, r_2, 1}$, there exists a $y\in V_{G_2}$ such that $y\rightarrow r_2$ in $M_2$.
From \cref{item-1-of-thm:nes-and-suf-cond-for-tree-graph-to-be-an-MEC} of \cref{thm:nes-and-suf-cond-for-tree-graph-to-be-an-MEC}, there cannot be an induced subgraph $y\rightarrow r_2 - y'$ in $M_2$. Therefore, each edge adjacent to $r_2$ in $M_2$ is directed. This means we can partition the neighbors of $r_2$ in $G_2$ into $Y$ and $Y'$ such that $Y = \{y: y\rightarrow r_2 \in M_2\}$ and $Y' = \{y': r_2\rightarrow y' \in M_2\}$ with $Y$ being non-empty. 
From \cref{lem:directed-edge-in-projection-implies-directed-edge-in-MEC}, for $y\in Y$, $y\rightarrow r_2 \in M$, and for $y' \in Y'$, $r_2\rightarrow y' \in M$. From \cref{item-1-of-thm:nes-and-suf-cond-for-tree-graph-to-be-an-MEC} of \cref{thm:nes-and-suf-cond-for-tree-graph-to-be-an-MEC}, either $r_1\rightarrow r_2 \in M$ or $r_2\rightarrow r_1 \in M$, otherwise, for some $y \in Y$, $M$ contains an induced subgraph $y\rightarrow r_2-r_1$, contradicting \cref{item-1-of-thm:nes-and-suf-cond-for-tree-graph-to-be-an-MEC} of \cref{thm:nes-and-suf-cond-for-tree-graph-to-be-an-MEC}.

Suppose $r_1 \rightarrow r_2 \in M$.  Then, the v-structures with endpoints $r_2$ and $r_1$ are $y\rightarrow r_2\leftarrow r_1$ for each $y \in Y$. In this case, the v-structures of $M$ obey \cref{eq:v-structure-M0-M1-1}.

Suppose $r_2\rightarrow r_1 \in M$. Since $M_1 \in \setofMECs{G_1, r_1, 0, i}$, there cannot be an edge $x\rightarrow r_1$ in $M_1$. This implies that we can partition the neighbors of $r_1$ in $G_1$ into $X$ and $X'$ such that $X = \{x: x-r_1 \in M_1\}$ is a set of $i$ elements (from the definition of $\setofMECs{G_1, r_1, 0, i}$), and $X' = \{x': r_1 \rightarrow x' \in M_1\}$ is a set of $\delta_1 - i$ elements. 
Since $r_2\rightarrow r_1 \in M$, from \cref{item-1-of-thm:nes-and-suf-cond-for-tree-graph-to-be-an-MEC} of \cref{thm:nes-and-suf-cond-for-tree-graph-to-be-an-MEC}, for $x \in X$, either $x\rightarrow r_1 \in M$ or $r_1\rightarrow x \in M$, otherwise, there exists an induced subgraph $r_2\rightarrow r_1-x \in M_1$, contradicting \cref{item-1-of-thm:nes-and-suf-cond-for-tree-graph-to-be-an-MEC} of \cref{thm:nes-and-suf-cond-for-tree-graph-to-be-an-MEC}. 
If for two distinct $x_1, x_2 \in X$,  $x_1\rightarrow r_1, x_2 \rightarrow r_1 \in M$, then we get a v-structure $x_1\rightarrow r_1 \leftarrow x_2$ in $M$ such the nodes of the v-structure are in $V_{G_1}$, but the v-structure is not in $M_1$, implying $\mathcal{V}(M[V_{G_1}]) \neq \mathcal{V}(M_1)$, which further implies $\mathcal{P}(M, V_{G_1}) \neq M_1$, a contradiction.
This implies that there exists at most one $x \in X$ such that $x\rightarrow r_1 \in M$. If there does not exist any $x \in X$ such that $x\rightarrow r_1 \in M$, then there won't be any v-structure in $M$ with endpoints $r_1$ and $r_2$. In this case, $M$ only contains the v-structures of $M_1$ and $M_2$, and obeys \cref{eq:v-structure-M0-M1-2}. If there exists an $x\in X$ such that $x\rightarrow r_1 \in M$, then $r_2\rightarrow r_1 \leftarrow x$ is the only v-structure with endpoints $r_1$ and $r_2$ in $M$. In this case, $M$ obeys \cref{eq:v-structure-M0-M1-3}. This completes the proof of \cref{item-1-of-lem:relation-between-v-structures-M0-M1} of \cref{lem:relation-between-v-structures-M0-M1}.

\end{proof}

As discussed in the introduction, there cannot be two distinct MECs with the same skeleton and the same set of v-structures.  Uniqueness of the MECs present in \cref{item-2-of-lem:relation-between-v-structures-M0-M1,item-3-of-lem:relation-between-v-structures-M0-M1,item-4-of-lem:relation-between-v-structures-M0-M1} of \cref{lem:relation-between-v-structures-M0-M1} comes from this fact. 
We now prove \cref{item-2-of-lem:relation-between-v-structures-M0-M1,item-3-of-lem:relation-between-v-structures-M0-M1,item-4-of-lem:relation-between-v-structures-M0-M1} of \cref{lem:relation-between-v-structures-M0-M1}.  

\begin{proof}[Proof of \cref{item-2-of-lem:relation-between-v-structures-M0-M1} of \cref{lem:relation-between-v-structures-M0-M1}]
    We construct an MEC of $G$ that obeys \cref{item-2-of-lem:relation-between-v-structures-M0-M1} of \cref{lem:relation-between-v-structures-M0-M1}. Let $M = M_1\cup M_2 \cup \{r_1\rightarrow r_2\}$. To prove that $M$ obeys \cref{item-2-of-lem:relation-between-v-structures-M0-M1}, we have to show the following: (a) $M$ is an MEC of $G$, (b) $M$ obeys \cref{eq:v-structure-M0-M1-1}, (c) $\mathcal{P}(M, V_{G_1}, V_{G_2}) = (M_1, M_2)$, and (d) $M \in \setofMECs{G, r_1,0, i}$.   

    We first show that $M$ is an MEC of $G$. From the construction of $M$, it is clear that the skeleton of $M$ is $G$. Thus, the only thing we have to show is that $M$ is an MEC, i.e., it obeys \cref{item-1-of-thm:nes-and-suf-cond-for-tree-graph-to-be-an-MEC,item-2-of-thm:nes-and-suf-cond-for-tree-graph-to-be-an-MEC} of \cref{thm:nes-and-suf-cond-for-tree-graph-to-be-an-MEC}. 

Suppose $M$ contains an induced subgraph of the form $a\rightarrow b -c$. From the construction of $M$, there are 3 possibilities: either (a) $a,b,c \in V_{G_1}$, or (b) $a,b,c \in V_{G_2}$, or (c) $a = r_1$, $b =r_2$, and $c \in V_{G_2}$. We show that none of these possibilities occurs.

Suppose $a, b, c \in V_{G_1}$. From the construction of $M$, $M_1$ is an induced subgraph of $M$. This implies $a\rightarrow b -c \in M_1$. This contradicts \cref{item-1-of-thm:nes-and-suf-cond-for-tree-graph-to-be-an-MEC} of \cref{thm:nes-and-suf-cond-for-tree-graph-to-be-an-MEC}, as $M_1$ is an MEC of a tree graph $G_1$. This implies that this possibility cannot occur. 
Similarly, there cannot be $a, b, c  \in V_{G_2}$, as from the construction of $M$, $M_2$ is also an induced subgraph of $M$. 

Suppose $a = r_1$, $b =r_2$, and $c \in V_{G_2}$. Then, from the construction of $M$, $r_2-c \in M_2$. Since $M_2 \in \setofMECs{G_2, r_2, 1}$, there must exist a node $y$ such that $y\rightarrow r_2 \in M_2$. This implies $M_2$ contains an induced subgraph of the form $y\rightarrow r_2 - c \in M_2$. Since $M_2$ is an MEC of a tree graph $G_2$, this contradicts \cref{item-1-of-thm:nes-and-suf-cond-for-tree-graph-to-be-an-MEC} of \cref{thm:nes-and-suf-cond-for-tree-graph-to-be-an-MEC}. This implies this possibility also cannot occur. This further implies that $M$ cannot have any induced subgraph of the form $a\rightarrow b - c$. 

 We now show that each directed edge $u\rightarrow v$ of $M$ obeys \cref{item-2-of-thm:nes-and-suf-cond-for-tree-graph-to-be-an-MEC} of \cref{thm:nes-and-suf-cond-for-tree-graph-to-be-an-MEC}. Let $u\rightarrow v$ be a directed edge in $M$. From the construction of $M$, either $u\rightarrow v \in M_1$, or $u\rightarrow v \in M_2$, or $u =r_1$ and $v =r_2$. Since $M_1$ is an MEC of a tree graph $G_1$, if $u\rightarrow v \in M_1$ then from \cref{item-2-of-thm:nes-and-suf-cond-for-tree-graph-to-be-an-MEC} of \cref{thm:nes-and-suf-cond-for-tree-graph-to-be-an-MEC}, $u\rightarrow v$ is part of an induced subgraph of $M_1$ of the form either $w\rightarrow u \rightarrow v$ or $w\rightarrow v\leftarrow u$. From the construction of $M$, the induced subgraph of $M_1$ is an induced subgraph of $M$. Thus, when $u\rightarrow v \in M_1$, $u\rightarrow v$ obeys \cref{item-2-of-thm:nes-and-suf-cond-for-tree-graph-to-be-an-MEC} of \cref{thm:nes-and-suf-cond-for-tree-graph-to-be-an-MEC} for $M$. 
 Similarly, if $u\rightarrow v \in M_2$ then $u\rightarrow v$ obeys \cref{item-2-of-thm:nes-and-suf-cond-for-tree-graph-to-be-an-MEC} of \cref{thm:nes-and-suf-cond-for-tree-graph-to-be-an-MEC} for $M$.  We now go through the final possibility when $u =r_1$ and $v =r_2$. Since $M_2 \in \setofMECs{G_2, r_2, 1}$, there exist a node $y\in V_{G_2}$ such that $y\rightarrow r_2 \in M_2$. From the construction of $M$, $y\rightarrow r_2\leftarrow r_1$ is an induced subgraph of $M$. This implies in this case also $u\rightarrow v$ obeys \cref{item-2-of-thm:nes-and-suf-cond-for-tree-graph-to-be-an-MEC} of \cref{thm:nes-and-suf-cond-for-tree-graph-to-be-an-MEC}.  This shows that $M$ obeys \cref{item-2-of-thm:nes-and-suf-cond-for-tree-graph-to-be-an-MEC} of \cref{thm:nes-and-suf-cond-for-tree-graph-to-be-an-MEC}. 

 The above discussion implies that $M$ obeys \cref{item-1-of-thm:nes-and-suf-cond-for-tree-graph-to-be-an-MEC,item-2-of-thm:nes-and-suf-cond-for-tree-graph-to-be-an-MEC} of \cref{thm:nes-and-suf-cond-for-tree-graph-to-be-an-MEC}. 
 This further implies that $M$ is an MEC.

    We now show that the v-structures of $M$ obey \cref{eq:v-structure-M0-M1-1}. From the construction of $M$, $M_1$ and $M_2$ are the induced subgraphs of $M$ on $V_{G_1}$ and $V_{G_2}$. Therefore, $\mathcal{V}(M[V_{G_1}]) = \mathcal{V}(M_1)$, and $\mathcal{V}(M[V_{G_2}]) = \mathcal{V}(M_2)$. Since $r_1 \rightarrow r_2 \in M$, other than the v-structures of $M_1$ and $M_2$, $M$ can have the v-structures with an edge $r_1\rightarrow r_2$. Since $M_2 \in \setofMECs{G_2, r_2, 1}$, there exists a non-empty set $Y = \{y: y\rightarrow r_2 \in M_2\}$. For each $y \in Y$, $y\rightarrow r_2 \leftarrow r_2$ is a v-structure in $M$. This shows that the v-structures in $M$ obey \cref{eq:v-structure-M0-M1-1}. 

    We now show that $\mathcal{P}(M, V_{G_1}, V_{G_2}) = (M_1, M_2)$. From the construction of $M$, $M[V_{G_1}] = M_1$, and $M[V_{G_2}] = M_2$. This implies $\mathcal{V}(M[V_{G_1}]) = \mathcal{V}(M_1)$, and $\mathcal{V}(M[V_{G_2}]) = \mathcal{V}(M_2)$. From \cref{def:projection}, this further implies $\mathcal{P}(M, V_{G_1}, V_{G_2}) = (M_1, M_2)$. 

From the construction of $M$, $M_1$ is an induced subgraph of $M$. 
In $M_1$, $i$ edges adjacent to $r_1$ are undirected and the remaining $\delta_1 - i$ adjacent edges of $r_1$ in $M_1$ are directed outward of $r_1$. The only edge that is adjacent to $r_1$ in $M$ but not adjacent to $r_1$ in $M_1$ is $r_1\rightarrow r_2$, which is directed outward of $r_1$.   Therefore, no edge adjacent to $r_1$ in $M$ is incoming to $r_1$, and $i$ undirected edges adjacent to $M$ are undirected. This implies  $M  \in \setofMECs{G, r_1, 0, i}$. This completes the proof of \cref{item-2-of-lem:relation-between-v-structures-M0-M1} of \cref{lem:relation-between-v-structures-M0-M1}.
\end{proof}

\begin{proof}[Proof of \cref{item-3-of-lem:relation-between-v-structures-M0-M1} of \cref{lem:relation-between-v-structures-M0-M1}]
    We construct an MEC of $G$ using the following steps:
    \begin{enumerate}
        \item
        \label{item-2-of-M0-M1-1}
        Initialize $M = M_1\cup M_2\cup \{r_2\rightarrow r_1\}$.
        \item
        \label{item-2-of-M0-M1-2}
        Update $M$ by replacing  $u-v$ of $M$ with $u\rightarrow v$, if $u-v \in M_1$ and there exists an undirected path from $r_1$ to $(u, v)$ in $M_1$. 
    \end{enumerate}
    We show that $M$ constructed using the above steps will obey \cref{item-3-of-lem:relation-between-v-structures-M0-M1} of \cref{lem:relation-between-v-structures-M0-M1}. To prove this, we have to show the following: (a) $M$ is an MEC of $G$, (b) $M$ obeys \cref{eq:v-structure-M0-M1-2}, (c) $\mathcal{P}(M, V_{G_1}, V_{G_2}) = (M_1, M_2)$, and (d) $M \in \setofMECs{G, r_1, 1}$.

    We first show that $M$ is an MEC, i.e., it obeys \cref{item-1-of-thm:nes-and-suf-cond-for-tree-graph-to-be-an-MEC,item-2-of-thm:nes-and-suf-cond-for-tree-graph-to-be-an-MEC} of \cref{thm:nes-and-suf-cond-for-tree-graph-to-be-an-MEC}. 

    We first show that $M$ obeys \cref{item-1-of-thm:nes-and-suf-cond-for-tree-graph-to-be-an-MEC} of \cref{thm:nes-and-suf-cond-for-tree-graph-to-be-an-MEC}.  Suppose $M$ contains an induced subgraph of the form $a\rightarrow b - c$. From the construction of $M$, there are three possibilities: either (a) $a,b,c \in V_{G_1}$, or (b) $a,b,c \in V_{G_2}$, or (c) $a = r_2$, $b =r_1$, and $c \in V_{G_1}$. 
    One by one, we show that none of the possibilities occurs. 

    Suppose, $M$ contains an induced subgraph of the form $a\rightarrow b-c$ such that $a,b, c \in V_{G_1}$. From the construction of $M$, $b-c \in M_1$ (note that \cref{item-2-of-M0-M1-2} does not change a directed edge in $M_1\cup M_2 \cup \{r_2\rightarrow r_1\}$ into an undirected edge), and either $a\rightarrow b \in M_1$ or $a-b \in M_1$. If $a\rightarrow b \in M_1$ then $a\rightarrow b - c$ is an induced subgraph in the MEC $M_1$, contradicting \cref{item-1-of-thm:nes-and-suf-cond-for-tree-graph-to-be-an-MEC} of \cref{thm:nes-and-suf-cond-for-tree-graph-to-be-an-MEC}. And, if $a-b \in M_1$ then it must have been directed at \cref{item-2-of-M0-M1-2}. This implies there exists a path from $r_2$ to $(a,b)$ in $M_1$. This further implies that there exists a path from $r_2$ to $(b,c)$ in $M_1$. But, then, we have $b\rightarrow c \in M$ at \cref{item-2-of-M0-M1-2}. This shows that the first possibility cannot occur. 
    
    From the construction of $M$, an induced subgraph of $M$ with vertices in $V_{G_2}$ is also an induced subgraph of $M_2$ (note that \cref{item-2-of-M0-M1-2} does not make any change in an edge with the endpoints in $V_{G_2}$). Since $M_2$ is an MEC of a tree graph $G_2$, from \cref{item-1-of-thm:nes-and-suf-cond-for-tree-graph-to-be-an-MEC} of \cref{thm:nes-and-suf-cond-for-tree-graph-to-be-an-MEC}, there cannot be an induced subgraph $u\rightarrow v-w \in M_2$. This implies that the second possibility cannot occur. 

Suppose $r_2 \rightarrow r_1-c$ in $M$ such that $c\in V_{G_1}$. From the construction of $M$, this implies $r_1-c \in M_1$. But, edge $r_1-c$ is a path from $r_1$ to $(r_1,c)$. Then, it must have been directed at \cref{item-2-of-M0-M1-2}. This implies that the third possibility also cannot occur. This shows that $M$ cannot have an induced subgraph of the form $u\rightarrow v- w$. This implies that $M$ obeys \cref{item-1-of-thm:nes-and-suf-cond-for-tree-graph-to-be-an-MEC} of \cref{thm:nes-and-suf-cond-for-tree-graph-to-be-an-MEC}.

We now show that $M$ obeys \cref{item-2-of-thm:nes-and-suf-cond-for-tree-graph-to-be-an-MEC} of \cref{thm:nes-and-suf-cond-for-tree-graph-to-be-an-MEC}. Suppose $u\rightarrow v$ is a directed edge in $M$ then either $u,v \in V_{G_1}$, or $u,v \in V_{G_2}$, or $u=r_1$ and $v = r_2$. We will show that in each of the possibilities, $u\rightarrow v$ in $M$ obeys \cref{item-2-of-thm:nes-and-suf-cond-for-tree-graph-to-be-an-MEC} of \cref{thm:nes-and-suf-cond-for-tree-graph-to-be-an-MEC}.

Suppose $u,v \in V_{G_1}$. From the construction of $M$, there are two cases: either $u\rightarrow v \in M_1$, or $u-v \in M_1$.
Suppose $u\rightarrow v \in M_1$.
Since $M_1$ is an MEC of a tree graph $G_1$, from \cref{item-2-of-thm:nes-and-suf-cond-for-tree-graph-to-be-an-MEC} of \cref{thm:nes-and-suf-cond-for-tree-graph-to-be-an-MEC}, $u\rightarrow v$ is part of an induced subgraph of $M_1$ of the form either $u\rightarrow v\leftarrow w$, or $w\rightarrow u\rightarrow v$. From the construction of $M$,  a directed induced subgraph of $M_1$ is an induced subgraph in $M$, as a directed edge in $M_1$ is a directed edge in $M$. Therefore, $u\rightarrow v$ is strongly protected in $M$.
When $u-v \in M_1$ then it must have been directed in $M$ at \cref{item-2-of-M0-M1-2}. This implies there exists an undirected path from $r_1$ to $(u,v)$ in $M_1$. 
Let the path be $P = (x_0 = r_1, x_1, \ldots, x_{l-1} = u, x_{l} = v)$ for some $l \geq 1$. From the construction of $M$, for $0\leq i < l$, $x_i \rightarrow x_{i+1} \in M$ because $P_i = (x_0 = r_1, \ldots, x_i, x_{i+1})$ is an undirected path in $M_1$ from $r_1$ to $(x_i, x_{i+1})$. If $l > 1$ then $u\rightarrow v$ is part of an induced subgraph  $x_{l-2} \rightarrow x_{l-1} \rightarrow x_l$ in $M$.
And, if $l=1$ then $r_1 = u$, and $u\rightarrow v$ is part of an induced subgraph $r_2\rightarrow r_1\rightarrow v$ in $M$.  Thus, we show that if $u,v \in V_{G_1}$ then $u\rightarrow v$ obeys \cref{item-2-of-thm:nes-and-suf-cond-for-tree-graph-to-be-an-MEC} of \cref{thm:nes-and-suf-cond-for-tree-graph-to-be-an-MEC}.

Suppose $u,v \in V_{G_2}$. From the construction of $M$, $M_2$ is an induced subgraph of $M$ (note that \cref{item-2-of-M0-M1-2} does not change the orientation of any edge of $M_2$). This implies $u\rightarrow v \in M_2$. Since $M_2$ is an MEC of a tree graph $G_2$, from \cref{item-2-of-thm:nes-and-suf-cond-for-tree-graph-to-be-an-MEC} of \cref{thm:nes-and-suf-cond-for-tree-graph-to-be-an-MEC}, $u\rightarrow v$ is part of an induced subgraph of $M_2$ of the form either $u\rightarrow v \leftarrow w$ or  $w\rightarrow u\rightarrow v$. Since from the construction of $M$, an induced subgraph of $M_2$ is an induced subgraph of $M$. Therefore, $u\rightarrow v$ in $M$ obeys \cref{item-2-of-thm:nes-and-suf-cond-for-tree-graph-to-be-an-MEC} of \cref{thm:nes-and-suf-cond-for-tree-graph-to-be-an-MEC}. 

We now show that $r_2\rightarrow r_1$ in $M$ also obeys \cref{item-2-of-thm:nes-and-suf-cond-for-tree-graph-to-be-an-MEC} of \cref{thm:nes-and-suf-cond-for-tree-graph-to-be-an-MEC}.
Since $M_2 \in \setofMECs{G_2, r_2, 1}$, there must exist $y\rightarrow r_2 \in M_2$. From the construction of $M$, $y\rightarrow r_2 \in M$. Then, $r_2\rightarrow r_1$ is part of an induced subgraph of the form $y\rightarrow r_2\rightarrow r_1$. Thus, we show that each edge of $M$ obeys \cref{item-2-of-thm:nes-and-suf-cond-for-tree-graph-to-be-an-MEC} of \cref{thm:nes-and-suf-cond-for-tree-graph-to-be-an-MEC}.

The above discussion shows that $M$ is an MEC as it obeys \cref{item-1-of-thm:nes-and-suf-cond-for-tree-graph-to-be-an-MEC,item-2-of-thm:nes-and-suf-cond-for-tree-graph-to-be-an-MEC} of \cref{thm:nes-and-suf-cond-for-tree-graph-to-be-an-MEC}. From the construction of $M$, the skelton of $M$ is $G$. This implies $M$ is an MEC of $G$.

We now show that $M$ obeys \cref{eq:v-structure-M0-M1-2}. From the construction of $M$, directed edges of $M_1$ and $M_2$ are directed edges of $M$. This implies the v-structures of $M_1$ and $M_2$ are the v-structures of $M$, i.e., $\mathcal{V}(M_1) \cup \mathcal{V}(M_2) \subseteq \mathcal{V}(M)$. For the completeness of the proof, we show that if $u\rightarrow v \leftarrow w$ is a v-structure in $M$ then either it is a v-structure of $M_1$ or it is a v-structure of $M_2$, i.e., $\mathcal{V}(M) \subseteq \mathcal{V}(M_1) \cup \mathcal{V}(M_2)$. This implies $\mathcal{V}(M) = \mathcal{V}(M_1) \cup \mathcal{V}(M_2)$, i.e., $M$ obeys \cref{eq:v-structure-M0-M1-2}.

Suppose there exists a v-structure $u\rightarrow v \leftarrow w$ in $M$. We show that either $u\rightarrow v \leftarrow w$ is a v-structure in $M_1$, or it is a v-structure in $M_2$. From the construction of $M$, there are following possibilities: (a) $u,v,w \in V_{G_1}$, or (b) $u = r_2$, $v =r_1$, and $w$ is a neighbor of $r_1$ in $V_{G_1}$, or (c) $u,v,w \in V_{G_2}$. 

Suppose $u,v, w \in V_{G_1}$. From the construction of $M$, if $a\rightarrow b \in M_1$ then $a\rightarrow b \in M$. This implies either $u-v \in M_1$, or $u\rightarrow v \in M_1$, and either $v-w \in M_1$ or $v\leftarrow w \in M_1$. From \cref{item-1-of-thm:nes-and-suf-cond-for-tree-graph-to-be-an-MEC} of \cref{thm:nes-and-suf-cond-for-tree-graph-to-be-an-MEC}, either $u\rightarrow v\leftarrow w \in M_1$ or $u-v-w \in M_1$ (recall that $M_1$ is an MEC). 
    If $u\rightarrow v\leftarrow w \in M_1$ then the v-structure is a v-structure of $M_1$. Suppose $u-v-w \in M_1$. Then, if $u\rightarrow v \in M$ then it must have been directed at \cref{item-2-of-M0-M1-2}, i.e., there exists an undirected path from $r_1$ to $(u,v)$.
    Since $u-v-w \in M_1$, the existence of a path from $r_1$ to $(u,v)$ in $M_1$ implies the existence of a path from $r_1$ to $(v,w)$ in $M_1$. Also, since $M_1$ is an MEC of a tree graph $G_1$, there cannot be two different paths from $r_1$ to $w$. 
    This implies at \cref{item-2-of-M0-M1-2}, if we have $u\rightarrow v \in M$ then we also have $v\rightarrow w\in M$. But, this is a contradiction, as we have assumed that $u\rightarrow v \leftarrow w \in M$.
    This implies there cannot be $u-v-w\in M_1$ and $u\rightarrow v\leftarrow w \in M$. 
    
    Suppose $u = r_2$, $v =r_1$, and $w$ is a neighbor of $r_1$ in  $V_{G_1}$. Since $M_1 \in \setofMECs{G_1, r_1, 0}$, either $r_1\rightarrow w \in M_1$ or $r_2-w \in M_2$. From the construction of $M$ (\cref{item-2-of-M0-M1-1,item-2-of-M0-M1-2}), if $r_1\rightarrow w \in M_1$ then $r_1\rightarrow w \in M$. But, we have $r_1\leftarrow w \in M$. This implies that $r_1-w \in M_1$. But, then, from \cref{item-2-of-M0-M1-2}, $r_1 \rightarrow w \in M$. This implies that this case cannot occur.
    
    We now move to the final possibility, when $u,v,w \in V_{G_2}$. From the construction of $M$, $M_2$ is an induced subgraph of $M$. This implies that if $u,v,w \in V_{G_2}$ then $u\rightarrow v \leftarrow w$ is a v-structure in $M_2$. Thus, we show that in all the possible scenarios, if $u\rightarrow v\leftarrow w$ is a v-structure in $M$ then it is a v-structure either in $M_1$ or in $M_2$. 
    This completes the proof that $M$ obeys \cref{eq:v-structure-M0-M1-2}.

    The above discussion implies that $\mathcal{V}(M[V_{G_1}]) = \mathcal{V}(M_1)$ and $\mathcal{V}(M[V_{G_2}]) = \mathcal{V}(M_2)$. Therefore, from \cref{def:projection}, $\mathcal{P}(M, V_{G_1}, V_{G_2}) = (M_1, M_2)$.

    Since there is an incoming edge $r_2\rightarrow r_1$ in $M$, $M \in \setofMECs{G, r_1, 1}$. This completes the proof of \cref{item-3-of-lem:relation-between-v-structures-M0-M1} of \cref{lem:relation-between-v-structures-M0-M1}.
\end{proof}

\begin{proof}[Proof of \cref{item-4-of-lem:relation-between-v-structures-M0-M1} of \cref{lem:relation-between-v-structures-M0-M1}]
    We construct an MEC $M$ of $G$ using the following steps:
    \begin{enumerate}
        \item
        \label{item-1-of-item-3-of-M0-M1}
        Initialize $M = M_1\cup M_2\cup \{r_2\rightarrow r_1\}$.
        \item
        \label{item-2-of-item-3-of-M0-M1}
        Pick an edge $x-r_1 \in M_1$. Replace the edge $x-r_1$ in $M$  with $x\rightarrow r_1$.
        \item
        \label{item-3-of-item-3-of-M0-M1}
        Update $M$ by replacing  $u-v$ of $M$ with $u\rightarrow v$, if $u-v \in M_1$ and there exists an undirected path from $r_1$ to $(u, v)$ in $M_1$ such that the path does not contain $x$.
    \end{enumerate}

    We show that $M$ constructed using the above steps will obey \cref{item-4-of-lem:relation-between-v-structures-M0-M1} of \cref{lem:relation-between-v-structures-M0-M1}. To prove this, we have to show the following: (a) $M$ is an MEC of $G$, (b) $M$ obeys \cref{eq:v-structure-M0-M1-3}, (c) $\mathcal{P}(M, V_{G_1}, V_{G_2}) = (M_1, M_2)$, and (d) $M \in \setofMECs{G, r_1, 1}$.

    We first show that $M$ is an MEC i.e., it obeys \cref{item-1-of-thm:nes-and-suf-cond-for-tree-graph-to-be-an-MEC,item-2-of-thm:nes-and-suf-cond-for-tree-graph-to-be-an-MEC} of \cref{thm:nes-and-suf-cond-for-tree-graph-to-be-an-MEC}. 

    We first show that $M$ obeys \cref{item-1-of-thm:nes-and-suf-cond-for-tree-graph-to-be-an-MEC} of \cref{thm:nes-and-suf-cond-for-tree-graph-to-be-an-MEC}.  Suppose $M$ contains an induced subgraph of the form $a\rightarrow b - c$. From the construction of $M$, there are three possibilities: either (a) $a,b,c \in V_{G_1}$, or (b) $a,b,c \in V_{G_2}$, or (c) $a = r_2$, $b =r_1$, and $c \in V_{G_1}$. 
    One by one, we show that none of the possibilities occurs.

    Suppose $a,b, c \in V_{G_1}$. From the construction of $M$, $b-c \in M_1$ (note that \cref{item-2-of-item-3-of-M0-M1,item-3-of-item-3-of-M0-M1} do not change a directed edge in $M_1\cup M_2 \cup \{r_2\rightarrow r_1\}$ into an undirected edge), and either $a\rightarrow b \in M_1$ or $a-b \in M_1$. If $a\rightarrow b \in M_1$ then $a\rightarrow b - c$ is an induced subgraph in the MEC $M_1$, contradicting \cref{item-1-of-thm:nes-and-suf-cond-for-tree-graph-to-be-an-MEC} of \cref{thm:nes-and-suf-cond-for-tree-graph-to-be-an-MEC}. And, if $a-b \in M_1$ then it must have been directed at either at \cref{item-2-of-item-3-of-M0-M1} (first case) or at \cref{item-3-of-item-3-of-M0-M1} (second case). If it has been directed at \cref{item-2-of-item-3-of-M0-M1} then $a = x, b = r_1$ and $c \neq x$ is a neighbor of $r_1$. This implies $r_1 -c$ is a path from $r_1$ to $(r_1, c)$. But, then at \cref{item-3-of-item-3-of-M0-M1}, $r_1-c$ has been converted into $r_1\rightarrow c$. This implies that the first case cannot occur. We move to the second case when $a-b$ has been directed at \cref{item-3-of-item-3-of-M0-M1}. This implies there exists a path from $r_1$ to $(a,b)$ in $M_1$. This further implies that there exists a path from $r_1$ to $(b,c)$ in $M_1$. But, then, we have $b\rightarrow c \in M$ at \cref{item-3-of-item-3-of-M0-M1}, a contradiction (as from our assumption $c\rightarrow b$ is part of the v-structure of $M$).
     This implies that the first possibility cannot occur. 

Suppose $a,b,c \in V_{G_2}$. From the construction of $M$, an induced subgraph of $M$ with vertices in $V_{G_2}$ is also an induced subgraph of $M_2$ (note that \cref{item-2-of-item-3-of-M0-M1,item-3-of-item-3-of-M0-M1} does not make any change in an edge with the endpoints in $V_{G_2}$). Since $M_2$ is an MEC of a tree graph $G_2$, from \cref{item-1-of-thm:nes-and-suf-cond-for-tree-graph-to-be-an-MEC} of \cref{thm:nes-and-suf-cond-for-tree-graph-to-be-an-MEC}, there cannot be an induced subgraph $u\rightarrow v-w \in M_2$. This shows that the second possibility cannot occur.

Suppose $a= r_2$, $b =r_1$, and $c\in V_{G_1}$. From the construction of $M$, this implies $r_1-c \in M_1$. There are two cases, either $c = x$, or $c \neq x$. If $c = x$, then at \cref{item-2-of-item-3-of-M0-M1}, we have $c\rightarrow b \in M$. Therefore, this case cannot occur. Suppose $c\neq x$.  But, edge $r_1-c$ is a path from $r_1$ to $(r_1,c)$. Then, it must have been directed at \cref{item-3-of-item-3-of-M0-M1}. This implies that the third possibility also cannot occur. This shows that $M$ cannot have an induced subgraph of the form $u\rightarrow v- w$. This implies that $M$ obeys \cref{item-1-of-thm:nes-and-suf-cond-for-tree-graph-to-be-an-MEC} of \cref{thm:nes-and-suf-cond-for-tree-graph-to-be-an-MEC}.

We now show that $M$ obeys \cref{item-2-of-thm:nes-and-suf-cond-for-tree-graph-to-be-an-MEC} of \cref{thm:nes-and-suf-cond-for-tree-graph-to-be-an-MEC}. Suppose $u\rightarrow v$ is a directed edge in $M$ then either $u,v \in V_{G_1}$, or $u,v \in V_{G_2}$, or $u=r_2$ and $v = r_1$. We show that in each of the possibilities $u\rightarrow v$ in $M$ obeys \cref{item-2-of-thm:nes-and-suf-cond-for-tree-graph-to-be-an-MEC} of \cref{thm:nes-and-suf-cond-for-tree-graph-to-be-an-MEC}.

Suppose $u,v \in V_{G_1}$. From the construction of $M$, there are two cases: either $u\rightarrow v \in M_1$ or $u-v \in M_1$.

Suppose $u\rightarrow v \in M_1$. 
Since $M_1$ is an MEC of a tree graph $G_1$, from \cref{item-2-of-thm:nes-and-suf-cond-for-tree-graph-to-be-an-MEC} of \cref{thm:nes-and-suf-cond-for-tree-graph-to-be-an-MEC}, $u\rightarrow v$ is part of an induced subgraph of $M_1$ of the form either $u\rightarrow v\leftarrow w$, or $w\rightarrow u\rightarrow v$. From the construction of $M$,  a directed induced subgraph of $M_1$ is an induced subgraph in $M$. Therefore, $u\rightarrow v$ is strongly protected in $M$.

Suppose $u-v \in M_1$. Then, $u-v$ must have been directed in $M$ either at \cref{item-2-of-item-3-of-M0-M1} or at \cref{item-3-of-item-3-of-M0-M1}. Suppose $u-v \in M_1$ and it gets directed in $M$ at \cref{item-2-of-item-3-of-M0-M1}. Then, $u = x$ and $v = r_2$. In that case, $u\rightarrow v$ is part of an induced subgraph $r_2\rightarrow r_1 \leftarrow x$, obeying \cref{item-2-of-thm:nes-and-suf-cond-for-tree-graph-to-be-an-MEC} of \cref{thm:nes-and-suf-cond-for-tree-graph-to-be-an-MEC}. 
Suppose $u-v \in M_1$ and it gets directed in $M$ at \cref{item-3-of-item-3-of-M0-M1}. This implies there exists an undirected path from $r_1$ to $(u,v)$ in $M_1$, and the path does not contain $x$. 
Let the path be $P = (x_0 = r_1, x_1, \ldots, x_{l-1} = u, x_{l} = v)$ for some $l \geq 1$. From the construction of $M$, for $0\leq i < l$, $x_i \rightarrow x_{i+1} \in M$ because $P_i = (x_0 = r_1, \ldots, x_i, x_{i+1})$ is an undirected path in $M_1$ from $r_1$ to $(x_i, x_{i+1})$, and $P_i$ does not contain $x$ either. If $l > 1$ then $u\rightarrow v$ is part of an induced subgraph  $x_{l-2} \rightarrow x_{l-1} \rightarrow x_l$ in $M$.
And, if $l=1$ then $r_1 = u$, and $u\rightarrow v$ is part of an induced subgraph $r_2\rightarrow r_1\rightarrow v$ in $M$.
 Thus, we show that if $u,v \in V_{G_1}$ then $u\rightarrow v$ obeys \cref{item-2-of-thm:nes-and-suf-cond-for-tree-graph-to-be-an-MEC} of \cref{thm:nes-and-suf-cond-for-tree-graph-to-be-an-MEC}.

Suppose $u,v \in V_{G_2}$. From the construction of $M$, $M_2$ is an induced subgraph of $M$ (note that \cref{item-2-of-item-3-of-M0-M1,item-3-of-item-3-of-M0-M1} do not change the orientation of any edge of $M_2$). This implies $u\rightarrow v \in M_2$. Since $M_2$ is an MEC of a tree graph $G_2$, from \cref{item-2-of-thm:nes-and-suf-cond-for-tree-graph-to-be-an-MEC} of \cref{thm:nes-and-suf-cond-for-tree-graph-to-be-an-MEC}, $u\rightarrow v$ is part of an induced subgraph of $M_2$ of the form either $u\rightarrow v \leftarrow w$ or  $w\rightarrow u\rightarrow v$. Since from the construction of $M$, an induced subgraph of $M_2$ is an induced subgraph of $M$. Therefore, $u\rightarrow v$ in $M$ obeys \cref{item-2-of-thm:nes-and-suf-cond-for-tree-graph-to-be-an-MEC} of \cref{thm:nes-and-suf-cond-for-tree-graph-to-be-an-MEC}. 
 
 We now show that $r_2\rightarrow r_1$ in $M$ also obeys \cref{item-2-of-thm:nes-and-suf-cond-for-tree-graph-to-be-an-MEC} of \cref{thm:nes-and-suf-cond-for-tree-graph-to-be-an-MEC}. $r_2\rightarrow r_1$ is part of an induced subgraph $r_2\rightarrow
  r_1 \leftarrow x$ in $M$. Thus, we show that each edge of $M$ obeys \cref{item-2-of-thm:nes-and-suf-cond-for-tree-graph-to-be-an-MEC} of \cref{thm:nes-and-suf-cond-for-tree-graph-to-be-an-MEC}.

The above discussion shows that $M$ is an MEC as it obeys \cref{item-1-of-thm:nes-and-suf-cond-for-tree-graph-to-be-an-MEC,item-2-of-thm:nes-and-suf-cond-for-tree-graph-to-be-an-MEC} of \cref{thm:nes-and-suf-cond-for-tree-graph-to-be-an-MEC}. From the construction of $M$, the skelton of $M$ is $G$. This implies $M$ is an MEC of $G$.

We now show that $M$ obeys \cref{eq:v-structure-M0-M1-3}. From the construction of $M$, directed edges of $M_1$ and $M_2$ are directed edges of $M$. Also, $r_2\rightarrow r_1 \leftarrow x$ is a v-structure in $M$. This implies  $\mathcal{V}(M_2) \cup \mathcal{V}(M_1) \cup \{r_2\rightarrow r_1 \leftarrow x\} \subseteq \mathcal{V}(M)$. For the completeness of the proof, we show that if $u\rightarrow v \leftarrow w$ is a v-structure in $M$ then either it is a v-structure of $M_1$ or it is a v-structure of $M_2$, or  $u =r_2$, $v = r_1$, and $w =  x$, i.e., $\mathcal{V}(M) \subseteq \mathcal{V}(M_1) \cup \mathcal{V}(M_2) \cup \{r_2\rightarrow r_1 \leftarrow x\}$. This further implies $\mathcal{V}(M) = \mathcal{V}(M_1) \cup \mathcal{V}(M_2)\cup \{r_2\rightarrow r_1 \leftarrow x\}$, i.e., $M$ obeys \cref{eq:v-structure-M0-M1-3}.

Suppose there exists a v-structure $u\rightarrow v \leftarrow w$ in $M$.  From the construction of $M$, there are following possibilities: (a) $u,v,w \in V_{G_1}$, or (b) $u = r_2$, $v =r_1$, and $w$ is a neighbor of $r_1$ in $V_{G_1}$, or (c) $u,v,w \in V_{G_2}$. 

Suppose $u,v, w \in V_{G_1}$. From the construction of $M$, if $a\rightarrow b \in M_1$ then $a\rightarrow b \in M$. This implies either $u-v \in M_1$ or $u\rightarrow v \in M_1$, and either $v-w \in M_1$ or $v\leftarrow w \in M_1$. From \cref{item-1-of-thm:nes-and-suf-cond-for-tree-graph-to-be-an-MEC} of \cref{thm:nes-and-suf-cond-for-tree-graph-to-be-an-MEC}, either $u\rightarrow v\leftarrow w \in M_1$ or $u-v-w \in M_1$. 
    If $u\rightarrow v\leftarrow w \in M_1$ then the v-structure is a v-structure of $M_1$. Suppose $u-v-w \in M_1$. Then, if $u\rightarrow v \in M$ then it must have been directed either at \cref{item-2-of-item-3-of-M0-M1} or at \cref{item-3-of-item-3-of-M0-M1}. 
    
    Suppose $u\rightarrow v \in M$ has been directed at \cref{item-2-of-item-3-of-M0-M1}. Then, $u = x$ and $v =r_1$. But, then $v-w$ is a path from $r_1$ to $(v, w)$ (and the path does not contain $x$). And, we have $v\rightarrow w \in M$ at \cref{item-3-of-item-3-of-M0-M1}, a contradiction, as the v-structure implies $w\rightarrow v \in M$. This implies that this case cannot occur.
    
    Suppose  $u\rightarrow v \in M$ has been directed at \cref{item-3-of-item-3-of-M0-M1}.  Then, there exists a path from $r_1$ to $(u,v)$, and the path does not contain $x$.
    Since $u-v-w \in M_1$, the existence of a path from $r_1$ to $(u,v)$ in $M_1$ implies the existence of a path from $r_1$ to $(v,w)$ in $M_1$, and the new path also does not contain $x$. Also, since $M_1$ is an MEC of a tree graph $G_1$, there cannot be two different paths from $r_1$ to $w$. 
    This implies at \cref{item-2-of-item-3-of-M0-M1}, if we have $u\rightarrow v \in M$ then we also have $v\rightarrow w\in M$. But, this is a contradiction, as we have assumed that $u\rightarrow v \leftarrow w \in M$.
    This implies there cannot be $u-v-w\in M_1$ and $u\rightarrow v\leftarrow w \in M$. 

    Suppose $u = r_2$, $v =r_1$, and $w$ is a neighbor of $r_1$ in  $V_{G_1}$. Since $M_1 \in \setofMECs{G_1, r_1, 0}$, either $r_1\rightarrow w \in M_1$ or $r_1-w \in M_1$. From the construction of $M$ (\cref{item-1-of-item-3-of-M0-M1,item-2-of-item-3-of-M0-M1,item-3-of-item-3-of-M0-M1}), if $r_1\rightarrow w \in M_1$ then $r_1\rightarrow w \in M$. But, we have $r_1\leftarrow w \in M$. This implies that $r_1-w \in M_1$. From \cref{item-2-of-item-3-of-M0-M1,item-3-of-item-3-of-M0-M1}, if we have $w \rightarrow r_1 \in M$ then $w = x$, otherwise, we have $r_1\rightarrow w \in M$ at \cref{item-3-of-item-3-of-M0-M1} as $r_1-w$ is an undirected path from $r_1$ to $(r_1, w)$, and the path does not contain $x$. This implies $u= r_2$, $v = r_1$, and $w =  x$. 
    
    We now move to the final possibility. Suppose $u,v,w \in V_{G_2}$. From the construction of $M$, $M_2$ is an induced subgraph of $M$. This implies that if $u,v,w \in V_{G_2}$ then $u\rightarrow v \leftarrow w$ is a v-structure in $M_2$. 
    Thus, we show that in all the possible scenarios, if $u\rightarrow v\leftarrow w$ is a v-structure in $M$ then it is a v-structure either in $M_1$ or in $M_2$, or $u= r_2$, $v =r_1$ and $w =  x$. 
    This completes the proof that $M$ obeys \cref{eq:v-structure-M0-M1-3}.

    The above discussion implies that $\mathcal{V}(M[V_{G_1}]) = \mathcal{V}(M_1)$ and $\mathcal{V}(M[V_{G_2}]) = \mathcal{V}(M_2)$. Therefore, from \cref{def:projection}, $\mathcal{P}(M, V_{G_1}, V_{G_2}) = (M_1, M_2)$.

   Since there is an incoming edge $r_2\rightarrow r_1$ in $M$, $M \in \setofMECs{G, r_1, 1}$. This completes the proof of \cref{item-4-of-lem:relation-between-v-structures-M0-M1} of \cref{lem:relation-between-v-structures-M0-M1}. 
\end{proof}

This completes the proof of \cref{lem:relation-between-v-structures-M0-M1}.
\end{proof}

\Cref{lem:relation-between-v-structures-M0-M1} implies the following corollary:

\begin{corollary}
\label{lem:counting-MECs-of-M0-M1}
For $M_1\in \setofMECs{G_1,r_1, 0, i}$, and $M_2 \in \setofMECs{G_2,r_2, 1}$,
the number of MECs $M$ of $G$ such that $\mathcal{P}(M, V_{G_1}, V_{G_2})= (M_1, M_2)$ is i+ 2. More specifically, $i+1$ MECs belong to $\setofMECs{G,r_1, 1}$, and one MEC belong to $\setofMECs{G, r_1, 0, i}$.
\end{corollary}
\begin{proof}
From \cref{item-1-of-lem:relation-between-v-structures-M0-M1} of \cref{lem:relation-between-v-structures-M0-M1}, for any MEC $M$ of $G$, if $\mathcal{P}(M, V_{G_1}, V_{G_2}) = (M_1, M_2)$ then $M$ has to obey either \cref{eq:v-structure-M0-M1-1,eq:v-structure-M0-M1-2,eq:v-structure-M0-M1-3}.
    From \cref{item-2-of-lem:relation-between-v-structures-M0-M1} of \cref{lem:relation-between-v-structures-M0-M1}, there exists a unique MEC that satisfies \cref{eq:v-structure-M0-M1-1}, belonging to $\setofMECs{G,r_1, 0, i}$. As indicated in \cref{item-3-of-lem:relation-between-v-structures-M0-M1} of \cref{lem:relation-between-v-structures-M0-M1}, another unique MEC, following \cref{eq:v-structure-M0-M1-2},  belonging to $\setofMECs{G,r_1, 1}$.

Given $M_1 \in \setofMECs{G_1, r_1, 0, i}$, $i$ instances of $x$ exist such that $x-r_1 \in E_{M_1}$. According to \cref{item-4-of-lem:relation-between-v-structures-M0-M1} of \cref{lem:relation-between-v-structures-M0-M1}, for each $x-r_1 \in E_{M_1}$, there exists a unique MEC satisfying \cref{eq:v-structure-M0-M1-3}. Consequently, there are $i$ MECs following \cref{eq:v-structure-M0-M1-3}, all belong to $\setofMECs{G,r_1, 1}$ as per \cref{item-3-of-lem:relation-between-v-structures-M0-M1}.

Hence, the total count of MECs is $i+2$. Out of which, $i+1$ MECs belong to $\setofMECs{G,r_1, 1}$, and one MEC belong to $\setofMECs{G, r_1, 0, i}$.
\end{proof}

We now deal with the last possibility, when for some $0 \leq i \leq \delta_1$, $M_1 \in \setofMECs{G_1, r_1, 0, i}$, and for some $0 \leq j \leq \delta_2$, $M_2 \in \setofMECs{G_2, r_2, 0, j}$

\begin{lemma}
\label{lem:relation-between-v-structures-M0-M0}
Let  for some $0\leq i \leq \delta_1$, $M_1 \in \setofMECs{G_1, r_1, 0, i}$, and for some $0\leq j \leq \delta_2$, $M_2 \in \setofMECs{G_2, r_2, 0, j}$. 
\begin{enumerate}
    \item
    \label{item-1-of-lem:relation-between-v-structures-M0-M0}
    Let $M \in$ MEC$(G)$ such that $\mathcal{P}(M,V_{G_1},V_{G_2})=(M_1,M_2)$. Then,
    one of the following occurs:
\begin{enumerate}
    \item
    \begin{equation}
    \label{eq:v-structure-M0-M0-1}
        \mathcal{V}(M)=\mathcal{V}(M_1) \cup \mathcal{V}(M_2)
    \end{equation}

        \item 
    For some undirected edge $x-r_1\in E_{M_1}$,
    \begin{equation}
    \label{eq:v-structure-M0-M0-2}
        \mathcal{V}(M)=\mathcal{V}(M_1) \cup \mathcal{V}(M_2) \cup \{x\rightarrow r_1\leftarrow r_2\}
    \end{equation}
    
    \item For some undirected edge $y-r_2\in E_{M_2}$,
    \begin{equation}
    \label{eq:v-structure-M0-M0-3}
        \mathcal{V}(M)=\mathcal{V}(M_1) \cup \mathcal{V}(M_2) \cup \{y\rightarrow r_2\leftarrow r_1\}
    \end{equation}
\end{enumerate}
\item 
\label{item-2-of-lem:relation-between-v-structures-M0-M0}
There exists a unique MEC $M$ of $G$ that obeys \cref{eq:v-structure-M0-M0-1}. Additionally, $M$ belongs to $\setofMECs{G, r_1, 0, i+1}$, and  $\mathcal{P}(M, V_{G_1}, V_{G_2}) = (M_1, M_2)$.

\item
\label{item-3-of-lem:relation-between-v-structures-M0-M0}
For each $x$ such that $x-r_1 \in E_{M_1}$, there exists a unique MEC $M$ of $G$ that obeys \cref{eq:v-structure-M0-M0-2}. Furthermore, $M$ belongs to $\setofMECs{G, r_1, 1}$, and $\mathcal{P}(M, V_{G_1}, V_{G_2}) = (M_1, M_2)$.

\item
\label{item-4-of-lem:relation-between-v-structures-M0-M0}
For each $y$ such that $y-r_2 \in E_{M_2}$, there exists a unique MEC $M$ of $G$ that obeys \cref{eq:v-structure-M0-M0-3}. Also, $M$ belongs to $\setofMECs{G, r_1, 0, i}$, and $\mathcal{P}(M, V_{G_1}, V_{G_2}) = (M_1, M_2)$.
\end{enumerate}
\end{lemma}

\begin{proof}
We first prove \cref{item-1-of-lem:relation-between-v-structures-M0-M0} of \cref{lem:relation-between-v-structures-M0-M0}.

\begin{proof}[Proof of \cref{item-1-of-lem:relation-between-v-structures-M0-M0} of \cref{lem:relation-between-v-structures-M0-M0}]
Let $M$ be an MEC of $G$ such that $\mathcal{P}(M, V_{G_1}, V_{G_2}) = (M_1, M_2)$.
    From \cref{obs:v-structures-of-G-neither-in-M1-nor-in-M2}, $M$ has the v-structures of $M_1$ and $M_2$. Other than that, it can only have the v-structures that involve an edge with endpoints $r_1$ and $r_2$. There are three possibilities: either $r_1-r_2 \in M$ or $r_2\rightarrow r_1 \in M$ or $r_1\rightarrow r_2 \in M$. 
    
    Suppose $r_1-r_2 \in M$. Then, there cannot be any v-structure with an edge with endpoints $r_1$ and $r_2$, because a v-structure contains directed edges, and $r_1-r_2$ is undirected. This implies that if $r_1-r_2 \in M$ then the only v-structures $M$ can have are the v-structures of $M_1$ and $M_2$, i.e., $M$ follows \cref{eq:v-structure-M0-M0-1}. 

Suppose  $r_2 \rightarrow r_1 \in M$.
Since $M_1 \in \setofMECs{G_1, r_1, 0, i}$,  we can partition the set of neighbors of $r_1$ in $G_1$ into $X$ and $X'$ such that $X = \{x: x-r_1 \in M_1\}$ is a set of $i$ nodes, and $X' = \{x': r_1\rightarrow x' \in M_1\}$ of $\delta_1 - i$ nodes. From \cref{lem:directed-edge-in-projection-implies-directed-edge-in-MEC}, for $x' \in X'$, $r_1\rightarrow x' \in M$. Since $r_2\rightarrow r_1 \in M$, from \cref{item-1-of-thm:nes-and-suf-cond-for-tree-graph-to-be-an-MEC} of \cref{thm:nes-and-suf-cond-for-tree-graph-to-be-an-MEC}, for $x \in X$, either $r_1\rightarrow x \in M$ or $x\rightarrow r_1 \in M$, otherwise, $r_2\rightarrow r_1 -x$ is an induced subgraph in $M$, contradicting \cref{item-1-of-thm:nes-and-suf-cond-for-tree-graph-to-be-an-MEC} of \cref{thm:nes-and-suf-cond-for-tree-graph-to-be-an-MEC}. 
If for two distinct $x_1, x_2 \in X$,  $x_1\rightarrow r_1, x_2 \rightarrow r_1 \in M$, then we get a v-structure $x_1\rightarrow r_1 \leftarrow x_2$ in $M$ such the nodes of the v-structure are in $V_{G_1}$, but the v-structure is not in $M_1$, implying $\mathcal{V}(M[V_{G_1}]) \neq \mathcal{V}(M_1)$, which further implies $\mathcal{P}(M, V_{G_1}) \neq M_1$, a contradiction. This implies that there exists at most one $x \in X$ such that $x\rightarrow r_1 \in M$. If there is no $x\in X$ such that $x\rightarrow r_1 \in M$ then there won't be any v-structure in $M$ with endpoints $r_1$ and $r_2$, and in this case, $M$ contains only the v-structures of $M_1$ and $M_2$, as in \cref{eq:v-structure-M0-M0-1}. If there exists an $x \in X$ such that $x\rightarrow r_1 \in M$ then other than the v-structures of $M_1$ and $M_2$, $M$ also contains a v-structure $r_2\rightarrow r_1 \leftarrow x$. In this case, $M$ obeys \cref{eq:v-structure-M0-M0-2}.

Suppose  $r_1 \rightarrow r_2 \in M$.
Since $M_2 \in \setofMECs{G_2, r_2, 0, j}$,  we can partition the set of neighbors of $r_2$ in $G_2$ into $Y$ and $Y'$ such that $Y = \{y: y-r_2 \in M_2\}$ is a set of $j$ nodes, and $Y' = \{y': r_2\rightarrow y' \in M_2\}$ of $\delta_2 - j$ nodes. From \cref{lem:directed-edge-in-projection-implies-directed-edge-in-MEC}, for $y' \in Y'$, $r_2\rightarrow y' \in M$. Since $r_1\rightarrow r_2 \in M$, from \cref{item-1-of-thm:nes-and-suf-cond-for-tree-graph-to-be-an-MEC} of \cref{thm:nes-and-suf-cond-for-tree-graph-to-be-an-MEC}, for $y \in Y$, either $r_2\rightarrow y \in M$ or $y\rightarrow r_2 \in M$, otherwise, $r_1\rightarrow r_2 -y$ is an induced subgraph in $M$, contradicting \cref{item-1-of-thm:nes-and-suf-cond-for-tree-graph-to-be-an-MEC} of \cref{thm:nes-and-suf-cond-for-tree-graph-to-be-an-MEC}. 
If for two distinct $y_1, y_2 \in Y$,  $y_1\rightarrow r_2, y_2 \rightarrow r_2 \in M$, then we get a v-structure $y_1\rightarrow r_2 \leftarrow y_2$ in $M$ such the nodes of the v-structure are in $V_{G_2}$, but the v-structure is not in $M_2$, implying $\mathcal{V}(M[V_{G_2}]) \neq \mathcal{V}(M_2)$, which further implies $\mathcal{P}(M, V_{G_2}) \neq M_2$, a contradiction. This implies that there exists at most one $y \in Y$ such that $y\rightarrow r_2 \in M$. If there is no $y\in Y$ such that $y\rightarrow r_2 \in M$ then there won't be any v-structure in $M$ with endpoints $r_1$ and $r_2$, and in this case, $M$ contains only the v-structures of $M_1$ and $M_2$, as in \cref{eq:v-structure-M0-M0-1}. If there exists a $y \in Y$ such that $y\rightarrow r_2 \in M$ then other than the v-structures of $M_1$ and $M_2$, $M$ also contains a v-structure $r_1\rightarrow r_2 \leftarrow y$. In this case, $M$ obeys \cref{eq:v-structure-M0-M0-3}.

This implies in all the possibilities, $M$ either follows \cref{eq:v-structure-M0-M0-1} or it follows \cref{eq:v-structure-M0-M0-2} or it follows \cref{eq:v-structure-M0-M0-3}. This completes the proof of \cref{item-1-of-lem:relation-between-v-structures-M0-M0}.
\end{proof}

As discussed in the introduction, there cannot be two distinct MECs with the same skeleton and the same set of v-structures.  Uniqueness of the MECs present in \cref{item-2-of-lem:relation-between-v-structures-M0-M0,item-3-of-lem:relation-between-v-structures-M0-M0,item-4-of-lem:relation-between-v-structures-M0-M0} of \cref{lem:relation-between-v-structures-M0-M0} comes from this fact. 
We now prove \cref{item-2-of-lem:relation-between-v-structures-M0-M0,item-3-of-lem:relation-between-v-structures-M0-M0,item-4-of-lem:relation-between-v-structures-M0-M0} of \cref{lem:relation-between-v-structures-M0-M0}. 

\begin{proof}[Proof of \cref{item-2-of-lem:relation-between-v-structures-M0-M0} of \cref{lem:relation-between-v-structures-M0-M0}]
    We construct an MEC of $G$ that obeys \cref{item-2-of-lem:relation-between-v-structures-M0-M0} of \cref{lem:relation-between-v-structures-M0-M0}. 
    Let $M = M_1\cup M_2 \cup \{r_1 - r_2\}$. To prove that $M$ obeys \cref{item-2-of-lem:relation-between-v-structures-M0-M0}, we have to show the following: (a) $M$ is an MEC of $G$, (b) $M$ obeys \cref{eq:v-structure-M0-M0-1}, (c) $\mathcal{P}(M, V_{G_1}, V_{G_2}) = (M_1, M_2)$, and (d) $M \in \setofMECs{G, r_1,0, i+1}$.   

    We first show that $M$ is an MEC of $G$. From the construction of $M$, it is clear that the skeleton of $M$ is $G$. Thus, the only thing we have to show is that $M$ is an MEC, i.e., it obeys \cref{item-1-of-thm:nes-and-suf-cond-for-tree-graph-to-be-an-MEC,item-2-of-thm:nes-and-suf-cond-for-tree-graph-to-be-an-MEC} of \cref{thm:nes-and-suf-cond-for-tree-graph-to-be-an-MEC}. 

We start with proving \cref{item-1-of-thm:nes-and-suf-cond-for-tree-graph-to-be-an-MEC} of \cref{thm:nes-and-suf-cond-for-tree-graph-to-be-an-MEC}. Suppose $M$ contains an induced subgraph of the form $a\rightarrow b -c$. Since $M_1 \in \setofMECs{G_1, r_1, 0, i}$, there cannot be an edge of the form $x\rightarrow r_1$ in $M_1$. Similarly, since $M_2 \in \setofMECs{G_2, r_2, 0, j}$, there cannot be an edge of the form $y\rightarrow r_2 \in M_2$. Since $M = M_1\cup M_2 \cup \{r_1 - r_2\}$, there cannot be an edge of the form $x\rightarrow r_1$ and $y\rightarrow r_2$ in $M$. This implies either (a) $a,b,c \in V_{G_1}$, or (b) $a,b,c \in V_{G_2}$. If $a, b, c \in V_{G_1}$ then from the construction of $M$, $a\rightarrow b -c \in M_1$. Since $M_1$ is an MEC of a tree graph $T_1$, this contradicts \cref{item-1-of-thm:nes-and-suf-cond-for-tree-graph-to-be-an-MEC} of \cref{thm:nes-and-suf-cond-for-tree-graph-to-be-an-MEC}. This implies that this case cannot occur. Similarly, there cannot be $a, b, c  \in V_{G_2}$. This implies $M$ cannot have any induced subgraph of the form $a\rightarrow b - c$. 

 We now show that each directed edge $u\rightarrow v$ of $M$ obeys \cref{item-2-of-thm:nes-and-suf-cond-for-tree-graph-to-be-an-MEC} of \cref{thm:nes-and-suf-cond-for-tree-graph-to-be-an-MEC}. Let $u\rightarrow v$ be a directed edge in $M$. From the construction of $M$, either $u\rightarrow v \in M_1$, or $u\rightarrow v \in M_2$. Since $M_1$ is an MEC of a tree graph $G_1$, if $u\rightarrow v \in M_1$ then from \cref{item-2-of-thm:nes-and-suf-cond-for-tree-graph-to-be-an-MEC} of \cref{thm:nes-and-suf-cond-for-tree-graph-to-be-an-MEC}, $u\rightarrow v$ is part of an induced subgraph of $M_1$ of the form either $w\rightarrow u \rightarrow v$ or $w\rightarrow v\leftarrow u$. From the construction of $M$, the induced subgraph of $M_1$ is an induced subgraph of $M$. Thus, when $u\rightarrow v \in M_1$, $u\rightarrow v$ obeys \cref{item-2-of-thm:nes-and-suf-cond-for-tree-graph-to-be-an-MEC} of \cref{thm:nes-and-suf-cond-for-tree-graph-to-be-an-MEC} for $M$. 
 Similarly, if $u\rightarrow v \in M_2$ then $u\rightarrow v$ obeys \cref{item-2-of-thm:nes-and-suf-cond-for-tree-graph-to-be-an-MEC} of \cref{thm:nes-and-suf-cond-for-tree-graph-to-be-an-MEC} for $M$.    This shows that $M$ obeys \cref{item-2-of-thm:nes-and-suf-cond-for-tree-graph-to-be-an-MEC} of \cref{thm:nes-and-suf-cond-for-tree-graph-to-be-an-MEC}. 

 The above discussion implies that $M$ obeys \cref{item-1-of-thm:nes-and-suf-cond-for-tree-graph-to-be-an-MEC,item-2-of-thm:nes-and-suf-cond-for-tree-graph-to-be-an-MEC} of \cref{thm:nes-and-suf-cond-for-tree-graph-to-be-an-MEC}. 
 This further implies that $M$ is an MEC.

    We now show that the v-structures of $M$ obey \cref{eq:v-structure-M0-M0-1}. From the construction of $M$, $M_1$ and $M_2$ are the induced subgraphs of $M$ on $V_{G_1}$ and $V_{G_2}$. Therefore, $\mathcal{V}(M[V_{G_1}]) = \mathcal{V}(M_1)$, and $\mathcal{V}(M[V_{G_2}]) = \mathcal{V}(M_2)$. Since $r_1 - r_2 \in M$, other than the v-structures of $M_1$ and $M_2$, $M$ cannot have any v-structures.  This shows that the v-structures in $M$ obey \cref{eq:v-structure-M0-M0-1}. 

    We now show that $\mathcal{P}(M, V_{G_1}, V_{G_2}) = (M_1, M_2)$. From the construction of $M$, $M[V_{G_1}] = M_1$, and $M[V_{G_2}] = M_2$. This implies $\mathcal{V}(M[V_{G_1}]) = \mathcal{V}(M_1)$, and $\mathcal{V}(M[V_{G_2}]) = \mathcal{V}(M_2)$. From \cref{def:projection}, this further implies $\mathcal{P}(M, V_{G_1}, V_{G_2}) = (M_1, M_2)$. 

From the construction of $M$, $M_1$ is an induced subgraph of $M$. 
In $M_1$, $i$ edges adjacent to $r_1$ are undirected and the remaining $\delta_1 - i$ adjacent edges of $r_1$ in $M_1$ are directed outward of $r_1$. Since $M_1$ is an induced subgraph of $M$, the $i$ adjacent undirected edges of $M_1$ are undirected in $M$ as well. Other than that, it has an additional undirected edge $r_1-r_2$, adjacent to $r_1$. This implies the total number of undirected edges adjacent to $r_1$ in $M$ is $i+1$. Therefore, no edge adjacent to $r_1$ in $M$ is incoming to $r_1$, and $i+1$ undirected edges adjacent to $M$ are undirected. This implies  $M  \in \setofMECs{G, r_1, 0, i+1}$. This completes the proof of \cref{item-2-of-lem:relation-between-v-structures-M0-M0} of \cref{lem:relation-between-v-structures-M0-M0}.
   
\end{proof}

\begin{proof}[Proof of \cref{item-3-of-lem:relation-between-v-structures-M0-M0} of \cref{lem:relation-between-v-structures-M0-M0}]
    We construct an MEC of $G$ using the following steps:
    \begin{enumerate}
        \item
        \label{item-1-of-item-2-of-M0-M0}
        Initialize $M = M_1\cup M_2\cup \{r_2\rightarrow r_1\}$.
        \item
        \label{item-2-of-item-2-of-M0-M0}
        Pick an edge $x-r_1 \in M_1$. Replace the edge $x-r_1$ in $M$  with $x\rightarrow r_1$.
        \item
        \label{item-3-of-item-2-of-M0-M0}
        Update $M$ by replacing  $u-v$ of $M$ with $u\rightarrow v$, if $u-v \in M_1$ and there exists an undirected path from $r_1$ to $(u, v)$ in $M_1$ such that the path does not contain $x$.
    \end{enumerate}

    We show that $M$ constructed using the above steps will obey \cref{item-3-of-lem:relation-between-v-structures-M0-M0} of \cref{lem:relation-between-v-structures-M0-M0}. To prove this, we have to show the following: (a) $M$ is an MEC of $G$, (b) $M$ obeys \cref{eq:v-structure-M0-M0-2}, (c) $\mathcal{P}(M, V_{G_1}, V_{G_2}) = (M_1, M_2)$, and (d) $M \in \setofMECs{G, r_1, 1}$.

    We first show that $M$ is an MEC i.e., it obeys \cref{item-1-of-thm:nes-and-suf-cond-for-tree-graph-to-be-an-MEC,item-2-of-thm:nes-and-suf-cond-for-tree-graph-to-be-an-MEC} of \cref{thm:nes-and-suf-cond-for-tree-graph-to-be-an-MEC}. 

    We start with showing that $M$ obeys \cref{item-1-of-thm:nes-and-suf-cond-for-tree-graph-to-be-an-MEC} of \cref{thm:nes-and-suf-cond-for-tree-graph-to-be-an-MEC}.  Suppose $M$ contains an induced subgraph of the form $a\rightarrow b - c$. From the construction of $M$, there are three possibilities: either (a) $a,b,c \in V_{G_1}$, or (b) $a,b,c \in V_{G_2}$, or (c) $a = r_2$, $b =r_1$, and $c \in V_{G_1}$. 
    One by one, we show that none of the possibilities occurs.

    Suppose $a,b, c \in V_{G_1}$. From the construction of $M$, $b-c \in M_1$ (note that \cref{item-2-of-item-3-of-M0-M0,item-3-of-item-3-of-M0-M0} do not change a directed edge in $M_1\cup M_2 \cup \{r_2\rightarrow r_1\}$ into an undirected edge), and either $a\rightarrow b \in M_1$ or $a-b \in M_1$. If $a\rightarrow b \in M_1$ then $a\rightarrow b - c$ is an induced subgraph in the MEC $M_1$, contradicting \cref{item-1-of-thm:nes-and-suf-cond-for-tree-graph-to-be-an-MEC} of \cref{thm:nes-and-suf-cond-for-tree-graph-to-be-an-MEC}. And, if $a-b \in M_1$ then it must have been directed at either at \cref{item-2-of-item-3-of-M0-M0} (first case) or at \cref{item-3-of-item-3-of-M0-M0} (second case). If it has been directed at \cref{item-2-of-item-3-of-M0-M0} then $a = x, b = r_1$ and $c \neq x$ is a neighbor of $r_1$. This implies $r_1 -c$ is a path from $r_1$ to $(r_1, c)$. But, then at \cref{item-3-of-item-3-of-M0-M0}, $r_1-c$ has been converted into $r_1\rightarrow c$. This implies that the first case cannot occur. We now move to the second case. Suppose $a-b$ has been directed at \cref{item-3-of-item-3-of-M0-M0}. This implies there exists a path from $r_1$ to $(a,b)$ in $M_1$. This further implies that there exists a path from $r_1$ to $(b,c)$ in $M_1$. But, then, we have $b\rightarrow c \in M$ at \cref{item-3-of-item-3-of-M0-M0}.
     This implies that there cannot be an induced subgraph $a\rightarrow b - c$ in $M$ with $a,b,c \in V_{G_1}$.

Suppose $a,b,c \in V_{G_2}$. From the construction of $M$, an induced subgraph of $M$ with vertices in $V_{G_2}$ is also an induced subgraph of $M_2$ (note that \cref{item-2-of-item-2-of-M0-M0,item-3-of-item-2-of-M0-M0} does not make any change in an edge with the endpoints in $V_{G_2}$). Since $M_2$ is an MEC of a tree graph $G_2$, from \cref{item-1-of-thm:nes-and-suf-cond-for-tree-graph-to-be-an-MEC} of \cref{thm:nes-and-suf-cond-for-tree-graph-to-be-an-MEC}, there cannot be an induced subgraph $u\rightarrow v-w \in M_2$. This shows that there cannot be an induced subgraph $a\rightarrow b - c$ in $M$ with $a,b,c \in V_{G_2}$.

Suppose $a = r_2$, $b = r_1$, and $c\in V_{G_1}$. From the construction of $M$, this implies $r_1-c \in M_1$. There are two cases, either $c = x$, or $c \neq x$. If $c = x$, then at \cref{item-2-of-item-3-of-M0-M0}, we have $c\rightarrow b \in M$. Therefore, this case cannot occur. Suppose $c\neq x$.  But, edge $r_1-c$ is a path from $r_1$ to $(r_1,c)$. Then, it must have been directed at \cref{item-3-of-item-3-of-M0-M0}. This implies that this possibility also cannot occur. This shows that $M$ cannot have an induced subgraph of the form $u\rightarrow v- w$. This implies that $M$ obeys \cref{item-1-of-thm:nes-and-suf-cond-for-tree-graph-to-be-an-MEC} of \cref{thm:nes-and-suf-cond-for-tree-graph-to-be-an-MEC}.

We now show that $M$ obeys \cref{item-2-of-thm:nes-and-suf-cond-for-tree-graph-to-be-an-MEC} of \cref{thm:nes-and-suf-cond-for-tree-graph-to-be-an-MEC}. Suppose $u\rightarrow v$ is a directed edge in $M$ then either $u,v \in V_{G_1}$, or $u,v \in V_{G_2}$, or $u=r_2$ and $v = r_1$. We show that in each of the possibilities $u\rightarrow v$ in $M$ obeys \cref{item-2-of-thm:nes-and-suf-cond-for-tree-graph-to-be-an-MEC} of \cref{thm:nes-and-suf-cond-for-tree-graph-to-be-an-MEC}.

Suppose $u,v \in V_{G_1}$. From the construction of $M$, a directed edge of $M_1$ is a directed edge in $M$. Therefore, since $u\rightarrow v \in M$, either $u\rightarrow v \in M_1$ or $u-v \in M_1$. Suppose $u\rightarrow v \in M_1$.  
Since $M_1$ is an MEC of a tree graph $G_1$, from \cref{item-2-of-thm:nes-and-suf-cond-for-tree-graph-to-be-an-MEC} of \cref{thm:nes-and-suf-cond-for-tree-graph-to-be-an-MEC}, $u\rightarrow v$ is part of an induced subgraph of $M_1$ of the form either $u\rightarrow v \leftarrow w$ or  $w\rightarrow u\rightarrow v$. Since from the construction of $M$, a directed induced subgraph of $M_1$ is an induced subgraph of $M$. Therefore, $u\rightarrow v$ in $M$ obeys \cref{item-2-of-thm:nes-and-suf-cond-for-tree-graph-to-be-an-MEC} of \cref{thm:nes-and-suf-cond-for-tree-graph-to-be-an-MEC}. 

When $u-v \in M_1$ then it must have been directed in $M$ either at \cref{item-2-of-item-3-of-M0-M0} or at \cref{item-3-of-item-3-of-M0-M0}. Suppose $u-v \in M_1$ and it gets directed in $M$ at \cref{item-2-of-item-3-of-M0-M0}. Then, $u = x$ and $v = r_2$. In that case, $u\rightarrow v$ is part of an induced subgraph $r_2\rightarrow r_1 \leftarrow x$, obeying \cref{item-2-of-thm:nes-and-suf-cond-for-tree-graph-to-be-an-MEC} of \cref{thm:nes-and-suf-cond-for-tree-graph-to-be-an-MEC}. Suppose $u-v \in M_1$ and it gets directed in $M$ at \cref{item-3-of-item-3-of-M0-M0}. This implies there exists an undirected path from $r_2$ to $(u,v)$ in $M_1$, and the path does not contain $x$. 
Let the path be $P = (x_0 = r_1, x_1, \ldots, x_{l-1} = u, x_{l} = v)$ for some $l \geq 1$. From the construction of $M$, for $0\leq i < l$, $x_i \rightarrow x_{i+1} \in M$ because $P_i = (x_0 = r_1, \ldots, x_i, x_{i+1})$ is an undirected path in $M_1$ from $r_1$ to $(x_i, x_{i+1})$, and $P_i$ does not contain $x$ either. If $l > 1$ then $u\rightarrow v$ is part of an induced subgraph  $x_{l-2} \rightarrow x_{l-1} \rightarrow x_l$ in $M$.
And, if $l=1$ then $r_1 = u$, and $u\rightarrow v$ is part of an induced subgraph $r_2\rightarrow r_1\rightarrow v$ in $M$.
 Thus, we show that if $u,v \in V_{G_1}$ then $u\rightarrow v$ obeys \cref{item-2-of-thm:nes-and-suf-cond-for-tree-graph-to-be-an-MEC} of \cref{thm:nes-and-suf-cond-for-tree-graph-to-be-an-MEC}.

Suppose $u,v \in V_{G_2}$. From the construction of $M$, $M_2$ is an induced subgraph of $M$. This implies $u\rightarrow v \in M_2$. Since $M_2$ is an MEC of a tree graph $G_2$, from \cref{item-2-of-thm:nes-and-suf-cond-for-tree-graph-to-be-an-MEC} of \cref{thm:nes-and-suf-cond-for-tree-graph-to-be-an-MEC}, $u\rightarrow v$ is part of an induced subgraph of $M_2$ of the form either $u\rightarrow v \leftarrow w$ or  $w\rightarrow u\rightarrow v$. Since from the construction of $M$, an induced subgraph of $M_2$ is an induced subgraph of $M$. Therefore, $u\rightarrow v$ in $M$ obeys \cref{item-2-of-thm:nes-and-suf-cond-for-tree-graph-to-be-an-MEC} of \cref{thm:nes-and-suf-cond-for-tree-graph-to-be-an-MEC}. 
 
 We now show that $r_2\rightarrow r_1$ in $M$ also obeys \cref{item-2-of-thm:nes-and-suf-cond-for-tree-graph-to-be-an-MEC} of \cref{thm:nes-and-suf-cond-for-tree-graph-to-be-an-MEC}. $r_2\rightarrow r_1$ is part of an induced subgraph $r_2\rightarrow
  r_1 \leftarrow x$ in $M$. Thus, we show that each edge of $M$ obeys \cref{item-2-of-thm:nes-and-suf-cond-for-tree-graph-to-be-an-MEC} of \cref{thm:nes-and-suf-cond-for-tree-graph-to-be-an-MEC}.

The above discussion shows that $M$ is an MEC as it obeys \cref{item-1-of-thm:nes-and-suf-cond-for-tree-graph-to-be-an-MEC,item-2-of-thm:nes-and-suf-cond-for-tree-graph-to-be-an-MEC} of \cref{thm:nes-and-suf-cond-for-tree-graph-to-be-an-MEC}. From the construction of $M$, the skelton of $M$ is $G$. This implies $M$ is an MEC of $G$.

We now show that $M$ obeys \cref{eq:v-structure-M0-M0-2}. From the construction of $M$, directed edges of $M_1$ and $M_2$ are directed edges of $M$. Also, $r_2\rightarrow r_1 \leftarrow x$ is a v-structure in $M$. This implies  $\mathcal{V}(M_1) \cup \mathcal{V}(M_2) \cup \{r_2\rightarrow r_1 \leftarrow x\} \subseteq \mathcal{V}(M)$. For the completeness of the proof, we show that if $u\rightarrow v \leftarrow w$ is a v-structure in $M$ then either it is a v-structure of $M_1$ or it is a v-structure of $M_2$, or $u = r_2$, $v = r_1$ and $w =  x$, i.e., $\mathcal{V}(M) \subseteq \mathcal{V}(M_1) \cup \mathcal{V}(M_2) \cup \{r_2\rightarrow r_1 \leftarrow x\}$. This further implies $\mathcal{V}(M) = \mathcal{V}(M_1) \cup \mathcal{V}(M_2)\cup \{r_2\rightarrow r_1 \leftarrow x\}$, i.e., $M$ obeys \cref{eq:v-structure-M0-M0-2}.

Suppose there exists a v-structure $u\rightarrow v \leftarrow w$ in $M$.  From the construction of $M$, there are following possibilities: (a) $u,v,w \in V_{G_1}$, or (b) $u = r_2$, $v =r_1$, and $w$ is a neighbor of $r_1$ in $V_{G_1}$, or (c) $u,v,w \in V_{G_2}$. 

Suppose $u,v, w \in V_{G_1}$. From the construction of $M$, if $a\rightarrow b \in M_1$ then $a\rightarrow b \in M$. This implies either $u-v \in M_1$ or $u\rightarrow v \in M_1$, and either $v-w \in M_1$ or $v\leftarrow w \in M_1$. From \cref{item-1-of-thm:nes-and-suf-cond-for-tree-graph-to-be-an-MEC} of \cref{thm:nes-and-suf-cond-for-tree-graph-to-be-an-MEC}, either $u\rightarrow v\leftarrow w \in M_1$ or $u-v-w \in M_1$ (recall that $M_1$ is an MEC of a tree graph $G_1$). 
    If $u\rightarrow v\leftarrow w \in M_1$ then the v-structure is a v-structure of $M_1$. Suppose $u-v-w \in M_1$. Then, if $u\rightarrow v \in M$ then it must have been directed either at \cref{item-2-of-item-3-of-M0-M0} or at \cref{item-2-of-item-3-of-M0-M0}. Suppose $u\rightarrow v \in M$ at \cref{item-2-of-item-3-of-M0-M0}. Then, $u = x$ and $v =r_1$. But, then $v-w$ is a path from $r_1$ to $(v, w)$ (and the path does not contain $x$). And, we have $v\rightarrow w \in M$ at \cref{item-3-of-item-3-of-M0-M0}, a contradiction, as the v-structure implies $w\rightarrow v \in M$. This implies that this case cannot occur. Suppose  $u\rightarrow v \in M$ at \cref{item-3-of-item-3-of-M0-M0}.  Then, there exists a path from $r_1$ to $(u,v)$, and the path does not contain $x$.
    Since $u-v-w \in M_1$, the existence of a path from $r_1$ to $(u,v)$ in $M_1$ implies the existence of a path from $r_1$ to $(v,w)$ in $M_1$, and the new path also does not contain $x$. Also, since $M_1$ is an MEC of a tree graph $G_1$, there cannot be two different paths from $r_1$ to $w$. 
    This implies at \cref{item-2-of-item-3-of-M0-M0}, if we have $u\rightarrow v \in M$ then we also have $v\rightarrow w\in M$. But, this is a contradiction, as we have assumed that $u\rightarrow v \leftarrow w \in M$.
    This implies there cannot be $u-v-w\in M_1$ and $u\rightarrow v\leftarrow w \in M$. 

    Suppose $u = r_2$, $v =r_1$, and $w$ is a neighbor of $r_1$ in  $V_{G_1}$. Since $M_1 \in \setofMECs{G_1, r_1, 0}$, either $r_1\rightarrow w \in M_1$ or $r_1-w \in M_1$. From the construction of $M$, if $r_1\rightarrow w \in M_1$ then $r_1\rightarrow w \in M$. But, we have $r_1\leftarrow w \in M$. This implies that $r_1-w \in M_1$. From \cref{item-2-of-item-3-of-M0-M0,item-3-of-item-3-of-M0-M0}, if we have $w \rightarrow r_1 \in M$ then $w = x$, otherwise, we have $r_1\rightarrow w \in M$ at \cref{item-3-of-item-3-of-M0-M0} as $r_1-w$ is an undirected path from $r_1$ to $(r_1, w)$, and the path does not contain $x$. This implies $u= r_2$, $v = r_1$, and $w = x$. 
    
    We now move to the final possibility. Suppose $u,v,w \in V_{G_2}$. From the construction of $M$, $M_2$ is an induced subgraph of $M$. This implies that if $u,v,w \in V_{G_2}$ then $u\rightarrow v \leftarrow w$ is a v-structure in $M_2$. 
    Thus, we show that in all the possible scenarios, if $u\rightarrow v\leftarrow w$ is a v-structure in $M$ then it is a v-structure either in $M_2$ or in $M_1$, or $v =r_1$ and $\{u,w\} = \{r_2, x\}$. 
    This completes the proof that $M$ obeys \cref{eq:v-structure-M0-M0-2}.

    The above discussion implies that $\mathcal{V}(M[V_{G_1}]) = \mathcal{V}(M_1)$ and $\mathcal{V}(M[V_{G_2}]) = \mathcal{V}(M_2)$. Therefore, from \cref{def:projection}, $\mathcal{P}(M, V_{G_1}, V_{G_2}) = (M_1, M_2)$.

   Since there is an incoming edge $r_2\rightarrow r_1$ in $M$, $M \in \setofMECs{G, r_1, 1}$. This completes the proof of \cref{item-3-of-lem:relation-between-v-structures-M0-M0} of \cref{lem:relation-between-v-structures-M0-M0}. 
\end{proof}

\begin{proof}[Proof of \cref{item-4-of-lem:relation-between-v-structures-M0-M0} of \cref{lem:relation-between-v-structures-M0-M0}]
    We construct an MEC of $G$ using the following steps:
    \begin{enumerate}
        \item
        \label{item-1-of-item-3-of-M0-M0}
        Initialize $M = M_1\cup M_2\cup \{r_1\rightarrow r_2\}$.
        \item
        \label{item-2-of-item-3-of-M0-M0}
        Pick an edge $y-r_2 \in M_2$. Replace the edge $y-r_2$ in $M$  with $y\rightarrow r_2$.
        \item
        \label{item-3-of-item-3-of-M0-M0}
            Update $M$ by replacing  $u-v$ of $M$ with $u\rightarrow v$, if $u-v \in M_2$ and there exists an undirected path from $r_2$ to $(u, v)$ in $M_2$ such that the path does not contain $y$.
    \end{enumerate}

    We show that $M$ constructed using the above steps will obey \cref{item-4-of-lem:relation-between-v-structures-M0-M0} of \cref{lem:relation-between-v-structures-M0-M0}. To prove this, we have to show the following: (a) $M$ is an MEC of $G$, (b) $M$ obeys \cref{eq:v-structure-M0-M0-3}, (c) $\mathcal{P}(M, V_{G_1}, V_{G_2}) = (M_1, M_2)$, and (d) $M \in \setofMECs{G, r_1, 0, i}$.

    We first show that $M$ is an MEC i.e., it obeys \cref{item-1-of-thm:nes-and-suf-cond-for-tree-graph-to-be-an-MEC,item-2-of-thm:nes-and-suf-cond-for-tree-graph-to-be-an-MEC} of \cref{thm:nes-and-suf-cond-for-tree-graph-to-be-an-MEC}. 

    We start with proving that $M$ obeys \cref{item-1-of-thm:nes-and-suf-cond-for-tree-graph-to-be-an-MEC} of \cref{thm:nes-and-suf-cond-for-tree-graph-to-be-an-MEC}.  Suppose $M$ contains an induced subgraph of the form $a\rightarrow b - c$. From the construction of $M$, there are three possibilities: either (a) $a,b,c \in V_{G_1}$, or (b) $a,b,c \in V_{G_2}$, or (c) $a = r_1$, $b =r_2$, and $c \in V_{G_2}$. 
    One by one, we show that none of the possibilities occurs. 
    
    Suppose $a,b,c \in V_{G_1}$. From the construction of $M$, an induced subgraph of $M$ with vertices in $V_{G_1}$ is also an induced subgraph of $M_1$ (note that \cref{item-2-of-item-3-of-M0-M0,item-3-of-item-3-of-M0-M0} does not make any change in an edge with the endpoints in $V_{G_1}$). Since $M_1$ is an MEC of a tree graph $G_1$, from \cref{item-1-of-thm:nes-and-suf-cond-for-tree-graph-to-be-an-MEC} of \cref{thm:nes-and-suf-cond-for-tree-graph-to-be-an-MEC}, there cannot be an induced subgraph $a\rightarrow b-c \in M_1$. This implies that the first possibility cannot occur.

Suppose $a,b, c \in V_{G_2}$. From the construction of $M$, $b-c \in M_2$ (note that \cref{item-2-of-item-3-of-M0-M0,item-3-of-item-3-of-M0-M0} do not change a directed edge in $M_1\cup M_2 \cup \{r_1\rightarrow r_2\}$ into an undirected edge), and either $a\rightarrow b \in M_2$ or $a-b \in M_2$. If $a\rightarrow b \in M_2$ then $a\rightarrow b - c$ is an induced subgraph in the MEC $M_2$, contradicting \cref{item-1-of-thm:nes-and-suf-cond-for-tree-graph-to-be-an-MEC} of \cref{thm:nes-and-suf-cond-for-tree-graph-to-be-an-MEC}. And, if $a-b \in M_2$ then it must have been directed at either at \cref{item-2-of-item-3-of-M0-M0} (first case) or at \cref{item-3-of-item-3-of-M0-M0} (second case). If it has directed at \cref{item-2-of-item-3-of-M0-M0} then $a = y, b = r_2$ and $c \neq y$ is a neighbor of $r_2$. This implies $r_2 -c$ is a path from $r_2$ to $(r_2, c)$. But, then at \cref{item-3-of-item-3-of-M0-M0}, $r_2-c$ has been converted into $r_2\rightarrow c$. This implies that the first case cannot occur. We move to the second case when $a-b$ has been directed at \cref{item-3-of-item-3-of-M0-M0}. This implies there exists a path from $r_2$ to $(a,b)$ in $M_2$. This further implies that there exists a path from $r_2$ to $(b,c)$ in $M_2$. But, then, we have $b\rightarrow c \in M$ at \cref{item-3-of-item-3-of-M0-M0}. This shows that the second case also cannot occur. This further implies the second possibility also cannot occur. We now move to the third possibility.

Suppose $a = r_1$, $b =r_2$, and $c \in V_{G_2}$. From the construction of $M$, this implies $r_2-c \in M_2$. There are two cases, either $c = y$, or $c \neq y$. If $c = y$, then at \cref{item-2-of-item-3-of-M0-M0}, we have $c\rightarrow b \in M$. Therefore, this case cannot occur. Suppose $c\neq y$.  But, edge $r_2-c$ is a path from $r_2$ to $(r_2,c)$. Then, it must have been directed at \cref{item-3-of-item-3-of-M0-M0}. This implies that the third possibility also cannot occur. This shows that $M$ cannot have an induced subgraph of the form $u\rightarrow v- w$. This implies that $M$ obeys \cref{item-1-of-thm:nes-and-suf-cond-for-tree-graph-to-be-an-MEC} of \cref{thm:nes-and-suf-cond-for-tree-graph-to-be-an-MEC}.

We now show that $M$ obeys \cref{item-2-of-thm:nes-and-suf-cond-for-tree-graph-to-be-an-MEC} of \cref{thm:nes-and-suf-cond-for-tree-graph-to-be-an-MEC}. Suppose $u\rightarrow v$ is a directed edge in $M$ then either $u,v \in V_{G_1}$, or $u,v \in V_{G_2}$, or $u=r_1$ and $v = r_2$. We show that in each of the possibilities $u\rightarrow v$ in $M$ obeys \cref{item-2-of-thm:nes-and-suf-cond-for-tree-graph-to-be-an-MEC} of \cref{thm:nes-and-suf-cond-for-tree-graph-to-be-an-MEC}.

Suppose $u,v \in V_{G_1}$, then from the construction of $M$, $u\rightarrow v \in M_1$. Since $M_1$ is an MEC of a tree graph $G_1$, from \cref{item-2-of-thm:nes-and-suf-cond-for-tree-graph-to-be-an-MEC} of \cref{thm:nes-and-suf-cond-for-tree-graph-to-be-an-MEC}, $u\rightarrow v$ is part of an induced subgraph of $M_1$ of the form either $u\rightarrow v \leftarrow w$ or  $w\rightarrow u\rightarrow v$. Since from the construction of $M$, an induced subgraph of $M_1$ is an induced subgraph of $M$. Therefore, $u\rightarrow v$ in $M$ obeys \cref{item-2-of-thm:nes-and-suf-cond-for-tree-graph-to-be-an-MEC} of \cref{thm:nes-and-suf-cond-for-tree-graph-to-be-an-MEC}. 

Suppose $u,v \in V_{G_2}$. From the construction of $M$, either $u\rightarrow v \in M_2$ or $u-v \in M_2$. Suppose $u\rightarrow v \in M_2$. Since $M_2$ is an MEC of a tree graph $G_2$, from \cref{item-2-of-thm:nes-and-suf-cond-for-tree-graph-to-be-an-MEC} of \cref{thm:nes-and-suf-cond-for-tree-graph-to-be-an-MEC}, $u\rightarrow v$ is part of an induced subgraph of $M_2$ of the form either $u\rightarrow v \leftarrow w$ or  $w\rightarrow u\rightarrow v$.  From the construction of $M$, a directed induced subgraph of $M_2$ is an induced subgraph of $M$ (as a directed edge of $M_2$ is a directed edge of $M$, and the skeletons of $M_2$ and $M[V_{G_2}]$ are the same).  Therefore, in this case also $u\rightarrow v$ in $M$ obeys \cref{item-2-of-thm:nes-and-suf-cond-for-tree-graph-to-be-an-MEC} of \cref{thm:nes-and-suf-cond-for-tree-graph-to-be-an-MEC}.

Suppose $u-v \in M_2$. Then, it must have been directed in $M$ either at \cref{item-2-of-item-3-of-M0-M0} or at \cref{item-3-of-item-3-of-M0-M0}. Suppose $u-v \in M_2$ and it gets directed in $M$ at \cref{item-2-of-item-3-of-M0-M0}. Then, $u = y$ and $v = r_2$. In that case, $u\rightarrow v$ is part of an induced subgraph $r_1\rightarrow r_2 \leftarrow y$, obeying \cref{item-2-of-thm:nes-and-suf-cond-for-tree-graph-to-be-an-MEC} of \cref{thm:nes-and-suf-cond-for-tree-graph-to-be-an-MEC}. Suppose $u-v$ gets directed in $M$ at \cref{item-3-of-item-3-of-M0-M0}. This implies there exists an undirected path from $r_2$ to $(u,v)$ in $M_2$, and the path does not contain $y$. 
Let the path be $P = (x_0 = r_2, x_1, \ldots, x_{l-1} = u, x_{l} = v)$ for some $l \geq 1$. From the construction of $M$, for $0\leq i < l$, $x_i \rightarrow x_{i+1} \in M$ because $P_i = (x_0 = r_2, \ldots, x_i, x_{i+1})$ is an undirected path in $M_2$ from $r_2$ to $(x_i, x_{i+1})$, and $P_i$ does not contain $y$ either. If $l > 1$ then $u\rightarrow v$ is part of an induced subgraph  $x_{l-2} \rightarrow x_{l-1} \rightarrow x_l$ in $M$.
And, if $l=1$ then $r_2 = u$, and $u\rightarrow v$ is part of an induced subgraph $r_1\rightarrow r_2\rightarrow v$ in $M$.
 Thus, we show that if $u,v \in V_{G_2}$ then $u\rightarrow v$ obeys \cref{item-2-of-thm:nes-and-suf-cond-for-tree-graph-to-be-an-MEC} of \cref{thm:nes-and-suf-cond-for-tree-graph-to-be-an-MEC}.

 We now show that $r_1\rightarrow r_2$ in $M$ also obeys \cref{item-2-of-thm:nes-and-suf-cond-for-tree-graph-to-be-an-MEC} of \cref{thm:nes-and-suf-cond-for-tree-graph-to-be-an-MEC}. $r_1\rightarrow r_2$ is part of an induced subgraph $r_1\rightarrow
  r_2 \leftarrow y$ in $M$. Thus, we show that each edge of $M$ obeys \cref{item-2-of-thm:nes-and-suf-cond-for-tree-graph-to-be-an-MEC} of \cref{thm:nes-and-suf-cond-for-tree-graph-to-be-an-MEC}.

The above discussion shows that $M$ is an MEC as it obeys \cref{item-1-of-thm:nes-and-suf-cond-for-tree-graph-to-be-an-MEC,item-2-of-thm:nes-and-suf-cond-for-tree-graph-to-be-an-MEC} of \cref{thm:nes-and-suf-cond-for-tree-graph-to-be-an-MEC}. From the construction of $M$, the skelton of $M$ is $G$. This implies $M$ is an MEC of $G$.

We now show that $M$ obeys \cref{eq:v-structure-M0-M0-3}. From the construction of $M$, directed edges of $M_1$ and $M_2$ are directed edges of $M$. Also, $r_1\rightarrow r_2 \leftarrow y$ is a v-structure in $M$. This implies  $\mathcal{V}(M_1) \cup \mathcal{V}(M_2) \cup \{r_1\rightarrow r_2 \leftarrow y\} \subseteq \mathcal{V}(M)$. For the completeness of the proof, we show that if $u\rightarrow v \leftarrow w$ is a v-structure in $M$ then either it is a v-structure of $M_1$ or it is a v-structure of $M_2$, or $u = r_1, v = r_2$ and $w = y$, i.e., $\mathcal{V}(M) \subseteq \mathcal{V}(M_1) \cup \mathcal{V}(M_2) \cup \{r_1\rightarrow r_2 \leftarrow y\}$. This further implies $\mathcal{V}(M) = \mathcal{V}(M_1) \cup \mathcal{V}(M_2)\cup \{r_1\rightarrow r_2 \leftarrow y\}$, i.e., $M$ obeys \cref{eq:v-structure-M0-M0-3}.

Suppose there exists a v-structure $u\rightarrow v \leftarrow w$ in $M$.  From the construction of $M$, there are following possibilities: (a) $u,v,w \in V_{G_1}$, or (b) $u = r_1$, $v =r_2$, and $w$ is a neighbor of $r_2$ in $V_{G_2}$, or (c) $u,v,w \in V_{G_2}$. 
    From the construction of $M$, $M_1$ is an induced subgraph of $M$ (note that \cref{item-2-of-item-3-of-M0-M0,item-3-of-item-3-of-M0-M0} do not change the orientation of any edge of $M_1$). This implies that if $u,v,w \in V_{G_1}$ then $u\rightarrow v \leftarrow w$ is a v-structure in $M_1$.

    Suppose $u = r_1$, $v =r_2$, and $w$ is a neighbor of $r_2$ in  $V_{G_2}$. Since $M_2 \in \setofMECs{G_2, r_2, 0}$, either $r_2\rightarrow w \in M_2$ or $r_2-w \in M_2$. From the construction of $M$ (\cref{item-1-of-item-3-of-M0-M0,item-2-of-item-3-of-M0-M0,item-3-of-item-3-of-M0-M0}), if $r_2\rightarrow w \in M_2$ then $r_2\rightarrow w \in M$. But, we have $r_2\leftarrow w \in M$. This implies that $r_2-w \in M_2$. From \cref{item-2-of-item-3-of-M0-M0,item-3-of-item-3-of-M0-M0}, if we have $w \rightarrow r_2 \in M$ then $w = y$, otherwise, we have $r_2\rightarrow w \in M$ at \cref{item-3-of-item-3-of-M0-M0} as $r_2-w$ is an undirected path from $r_2$ to $(r_2, w)$, and the path does not contain $y$. This implies $u = r_1, v = r_2$, and $w = y$. 
     
    We now move to the final possibility.
    Suppose $u,v, w \in V_{G_2}$. From the construction of $M$, if $a\rightarrow b \in M_2$ then $a\rightarrow b \in M$. This implies either $u-v \in M_2$ or $u\rightarrow v \in M_2$, and either $v-w \in M_2$ or $v\leftarrow w \in M_2$. From \cref{item-1-of-thm:nes-and-suf-cond-for-tree-graph-to-be-an-MEC} of \cref{thm:nes-and-suf-cond-for-tree-graph-to-be-an-MEC}, either $u\rightarrow v\leftarrow w \in M_2$ or $u-v-w \in M_2$. 
    
    Suppose $u\rightarrow v\leftarrow w \in M_2$. Then, the v-structure is a v-structure of $M_2$.
    
    Suppose $u-v-w \in M_2$. Then, if $u\rightarrow v \in M$ then it must have been directed either at \cref{item-2-of-item-3-of-M0-M0} or at \cref{item-2-of-item-3-of-M0-M0}. 
    
    Suppose $u\rightarrow v$ is added to $M$ at \cref{item-2-of-item-3-of-M0-M0}. Then, $u = y$ and $v =r_2$. But, then $v-w$ is a path from $(r_2)$ to $(v, w)$ (and the path does not contain $y$). And, we have $v\rightarrow w \in M$ at \cref{item-3-of-item-3-of-M0-M0}, a contradiction, as the v-structure implies $w\rightarrow v \in M$. This implies that this case cannot occur.
    
    Suppose  $u\rightarrow v$ is added to $M$ at \cref{item-3-of-item-3-of-M0-M0}.  Then, there exists a path from $r_2$ to $(u,v)$, and the path does not contain $y$.
    Since $u-v-w \in M_2$, the existence of a path from $r_2$ to $(u,v)$ in $M_2$ implies the existence of a path from $r_2$ to $(v,w)$ in $M_2$, and the new path also does not contain $y$. Also, since $M_2$ is an MEC of a tree graph $G_2$, there cannot be two different paths from $r_2$ to $w$. 
    This implies at \cref{item-3-of-item-3-of-M0-M0}, if we have $u\rightarrow v \in M$ then we also have $v\rightarrow w\in M$. But, this is a contradiction, as we have assumed that $u\rightarrow v \leftarrow w \in M$.
    This implies there cannot be $u-v-w\in M_2$ and $u\rightarrow v\leftarrow w \in M$. Thus, we show that in all the possible scenarios, if $u\rightarrow v\leftarrow w$ is a v-structure in $M$ then it is a v-structure either in $M_1$ or in $M_2$, or $u = r_1, v =r_2$ and $w = y$. 
    This completes the proof that $M$ obeys \cref{eq:v-structure-M0-M0-3}.

    The above discussion implies that $\mathcal{V}(M[V_{G_1}] = \mathcal{V}(M_1)$ and $\mathcal{V}(M[V_{G_2}] = \mathcal{V}(M_2)$. Therefore, from \cref{def:projection}, $\mathcal{P}(M, V_{G_1}, V_{G_2}) = (M_1, M_2)$.

    From the construction of $M$, $M_1$ is an induced subgraph of $M$. 
In $M_1$, $i$ edges adjacent to $r_1$ are undirected and the remaining $\delta_1 - i$ adjacent edges of $r_1$ in $M_1$ are directed outward of $r_1$. The only edge that is adjacent to $r_1$ in $M$ but not adjacent to $r_1$ in $M_1$ is $r_1\rightarrow r_2$, which is directed outward of $r_1$.   Therefore, no edge adjacent to $r_1$ in $M$ is incoming to $r_1$, and $i$ undirected edges adjacent to $M$ are undirected. This implies  $M  \in \setofMECs{G, r_1, 0, i}$.

    This completes the proof of \cref{item-4-of-lem:relation-between-v-structures-M0-M0} of \cref{lem:relation-between-v-structures-M0-M0}.
\end{proof}
This completes the proof of \cref{lem:relation-between-v-structures-M0-M0}.
\end{proof}

\Cref{lem:relation-between-v-structures-M0-M1} implies the following corollary:

\begin{corollary}
\label{lem:counting-MECs-of-M0-M0}
For $M_1\in \setofMECs{G_1,r_1, 0, i}$, and $M_2 \in \setofMECs{G_2,r_2, 0, j}$,
the number of MECs $M$ of $G$ such that $\mathcal{P}(M, V_{G_1}, V_{G_2})= (M_1, M_2)$ is $i+ j+1$. More specifically, $i$ MECs belong to $\setofMECs{G,r_1, 1}$, and $j$ MECs belong to $\setofMECs{G, r_1, 0, i}$, and one MEC belongs to $\setofMECs{G, r_1, 0, i+1}$.
\end{corollary}
\begin{proof}
From \cref{item-1-of-lem:relation-between-v-structures-M0-M0} of \cref{lem:relation-between-v-structures-M0-M0}, for any MEC $M$ of $G$, if $\mathcal{P}(M, V_{G_1}, V_{G_2}) = (M_1, M_2)$ then $M$ has to obey either \cref{eq:v-structure-M0-M0-1,eq:v-structure-M0-M0-2,eq:v-structure-M0-M0-3}.
    From \cref{item-2-of-lem:relation-between-v-structures-M0-M0} of \cref{lem:relation-between-v-structures-M0-M0}, there exists a unique MEC that satisfies \cref{eq:v-structure-M0-M0-1}, belonging to $\setofMECs{G,r_1, 0, i+1}$.

Given $M_1 \in \setofMECs{G_1, r_1, 0, i}$, $i$ instances of $x$ exist such that $x-r_1 \in E_{M_1}$. According to \cref{item-3-of-lem:relation-between-v-structures-M0-M0} of \cref{lem:relation-between-v-structures-M0-M0}, for each $x-r_1 \in E_{M_1}$, there exists a unique MEC satisfying \cref{eq:v-structure-M0-M0-2}. Consequently, there are $i$ MECs following \cref{eq:v-structure-M0-M0-2}, all belong to $\setofMECs{G,r_1, 1}$ as per \cref{item-3-of-lem:relation-between-v-structures-M0-M0}.

Given $M_2 \in \setofMECs{G_2, r_2, 0, j}$, $j$ instances of $y$ exist such that $y-r_2 \in E_{M_2}$. According to \cref{item-4-of-lem:relation-between-v-structures-M0-M0} of \cref{lem:relation-between-v-structures-M0-M0}, for each $y-r_2 \in E_{M_2}$, there exists a unique MEC satisfying \cref{eq:v-structure-M0-M0-3}. Consequently, there are $j$ MECs following \cref{eq:v-structure-M0-M0-3}, all belong to $\setofMECs{G,r_1, 0, i}$ as per \cref{item-4-of-lem:relation-between-v-structures-M0-M0}.

Hence, the total count of MECs is $i+j +1$. Out of which, $i$ MECs belong to $\setofMECs{G,r_1, 1}$, and $j$ MECs belong to $\setofMECs{G, r_1, 0, i}$, and one MEC belongs to $\setofMECs{G, r_1, 0, i+1}$.
\end{proof}

From \cref{lem:counting-MECs-of-M1-M1}, corresponding to each tuple $(M_1,M_2)$ such that $M_1\in \setofMECs{G_1, r_1, 1}$ and $M_2 \in \setofMECs{G_2, r_2, 1}$, there are two MECs of $G$ that belong to $\setofMECs{G, r_1, 1}$. From \cref{lem:counting-MECs-of-M1-M0}, corresponding to each tuple $(M_1,M_2)$ such that $M_1\in \setofMECs{G_1, r_1, 1}$ and $M_2 \in \setofMECs{G_2, r_2, 0, j}$, there are $j+2$ MECs of $G$ that belong to $\setofMECs{G, r_1, 1}$.
From \cref{lem:counting-MECs-of-M0-M1}, corresponding to each tuple $(M_1,M_2)$ such that $M_1\in \setofMECs{G_1, r_1, 0, i}$ and $M_2 \in \setofMECs{G_2, r_2, 1}$, there are $i+1$ MECs of $G$ that belong to $\setofMECs{G, r_1, 1}$. From \cref{lem:counting-MECs-of-M0-M0}, corresponding to each tuple $(M_1,M_2)$ such that $M_1\in \setofMECs{G_1, r_1, 0, i}$ and $M_2 \in \setofMECs{G_2, r_2, 0, j}$, there are $i$ MECs of $G$ that belong to $\setofMECs{G, r_1, 1}$. This makes the total number of MECs of $G$ that belongs to $\setofMECs{G, r_1, 1}$  equal to as in  \cref{item-1-of-recursion} of \cref{lem:recursive-method-to-count-MECs-of-tree}.

From \cref{lem:counting-MECs-of-M1-M1,lem:counting-MECs-of-M1-M0}, corresponding to each tuple $(M_1,M_2)$ such that $M_1\in \setofMECs{G_1, r_1, 1}$ and $M_2 \in \setofMECs{G_2}$, no MEC of $G$ belongs to $\setofMECs{G, r_1, 0}$.
From \cref{lem:counting-MECs-of-M0-M1}, corresponding to each tuple $(M_1,M_2)$ such that $M_1\in \setofMECs{G_1, r_1, 0, i}$ and $M_2 \in \setofMECs{G_2, r_2, 1}$, there are one MEC of $G$ that belongs to $\setofMECs{G, r_1, 0, i}$. From \cref{lem:counting-MECs-of-M0-M0},  corresponding to each tuple $(M_1,M_2)$ such that $M_1\in \setofMECs{G_1, r_1, 0, i}$ and $M_2 \in \setofMECs{G_2, r_2, 0, j}$, there are $j$ MECs of $G$ that belong to $\setofMECs{G, r_1, 0, i}$. From \cref{lem:counting-MECs-of-M0-M0}, for each tuple $(M_1,M_2)$ such that $M_1\in \setofMECs{G_1, r_1, 0, i-1}$ and $M_2 \in \setofMECs{G_2, r_2, 0}$, there are one MEC of $G$ that belong to $\setofMECs{G, r_1, 0, i}$. This makes total number of MECs of $G$ that belongs to $\setofMECs{G, r_1, 0, i}$  equal to as in  \cref{item-2a-of-recursion,item-2-of-recursion} of \cref{lem:recursive-method-to-count-MECs-of-tree}.

This completes the proof of \cref{lem:recursive-method-to-count-MECs-of-tree}.
\end{proof} \subsection{Omitted proof of \cref{subsection:chordal-graph}}
\label{subsection:omitted-proof-of-chordal-graph}

We first go through some results of \cite{sharma2023fixedparameter}:

\begin{definition}[Triangle Free Path, \cite{sharma2023fixedparameter}]
    \label{def:tfp}
    For a graph $G$, a path $P = (u_1, u_2, \ldots, u_l)$ is said to be a \tfp{} of $G$ if for any $1\leq i \leq l-2$, there does not exist an edge between $u_i$ and $u_{i+2}$ in $G$, i.e., neither $u_i-u_{i+2} \in E_G$ nor $u_i\rightarrow u_{i+2} \in E_G$ nor $u_i\leftarrow u_{i+2} \in E_G$.
\end{definition}

\begin{lemma}[\cite{sharma2023fixedparameter}]
\label{lem:tfp-is-cp-for-chordal-graph}
    For an undirected chordal graph $G$, if $P$ is a \tfp{} of $G$ then $P$ is a \cp{} of $G$.
\end{lemma}

\begin{observation}[Observation 2.30 of \cite{sharma2023fixedparameter}]
\label{obs:undirected-tfp-in-M-is-a-cp}
Let $M$ be an MEC, and $P = (u_1=x, u_2=y, \ldots, u_{l-1}=u, u_l=v)$ be a \tfp{} from $(x,y)$ to $(u,v)$ in $M$. If $u-v \in M$ then  $P$ is an undirected \cp{} in $M$, more specifically $x-y \in M$.
\end{observation}

\begin{lemma}[Corollary 2.50 of \cite{sharma2023fixedparameter}]
\label{obs:cond-for-ud-path-in-M_a-to-be-ud-path-in-M}
Let $M$ be an MEC and $M'$ be a projection of $M$.
Let $P = (u_1,u_2,\ldots, u_l)$ be a \tfp{} in $M'$. If $u_2\rightarrow u_1 \notin M$ then $P$ is also a \tfp{} in $M$.
\end{lemma}

\begin{lemma}[Corollary 2.51 of \cite{sharma2023fixedparameter}]
\label{obs:every-edge-of-triangle-free-path-is-directed}
Let $M$ be an MEC and $M'$ be a projection of $M$.
If $P = (u_1 = x, u_2 = y, u_3, \ldots, u_{l-1} = u, u_l =v)$ is a \tfp{} from $(x,y)$ to $(u,v)$ in $M'$, and $x\rightarrow y \in M$ then for all $1\leq i < l$, $u_i\rightarrow u_{i+1} \in M$.
\end{lemma}

\begin{proof}[Proof of \cref{item-3-of-obs:nes-cond-for-chordal-graph} of \cref{obs:nes-cond-for-chordal-graph}]
W.l.o.g., let us assume $i = 1$.
We start with the following observation.
\begin{observation}
    \label{obs:cp-in-M1-is-a-cp-in-O1}
    Let  $P = (u_0 = x, u_1, \ldots, u_l = y)$ be a \cp{} in $G_1$ from $x$ to $y$ such that $l \geq 2$, and $x, y \in r_1\cup N(r_1, G_1)$. Then,  each node of $P$ is in $r_1\cup N(r_1, G_1)$.
\end{observation}
\begin{proof}

There are two possibilities: either $x \in r_1$, or $x \notin r_1$. 
    \begin{enumerate}
       \item Suppose $x \in r_1$. Then, we claim that $y \in N(r_1, G_1)\setminus r_1$. 
    
       Suppose $y\in r_1$. Then, since $r_1$ is a clique of $G_1$, $x-y \in G_1$.  But, this implies $P$ is not a \cp{} in $G_1$, a contradiction. This implies $y \notin r_1$. 
        This further implies $y \in N(r_1, G_1) \setminus r_1$. 

        Since $y \in N(r_1, G_1) \setminus r_1$, there must exist a node $z \in r_1$ such that $y \in N(z, G_1)$, i.e., $z-y \in G_1$. Since $r_1$ is a clique in $G_1$, and $x, z \in r_1$, $x-z \in G_1$. This implies $u_0-z, u_l-z \in G_1$.
        We claim that each node in $P$ is a neighbor of $z$ in $G_1$.

        Suppose there exists a node $u_i$ which is not a neighbor of $z$. Pick the least such $i$. Since $u_0-z \in G_1$, $i > 0$. Pick the least $j$ such that $j> i$ and $u_j-z \in G_1$ (there must exist such $j$ as $u_l-z \in G_1$). Then $(z, u_{i-1}, u_i, \ldots, u_j, z)$ is a \cc{} in $G_1$ (since $P$ is a \cp{}, there cannot be an edge between two non-adjacent nodes $u_a$ and $u_b$ of $P$), a contradiction, as $G_1$ is a chordal graph. This implies each node in $P$ is a neighbor of $z$ in $G_1$. Since $z\in r_1$, this further implies each node of $P$ is in $r_1\cup N(r_1, G_1)$. 
        
       \item Suppose $x \notin r_1$.   Since $x\in r_1\cup N(r_1, G_1)$ and $x\notin r_1$, $x\in N(r_1, G_1)\setminus r_1$.
       Then, there must exist a node $z \in r_1$ such that $x\in N(z, G_1)$, i.e., $u_0-z \in G_1$. There are two possibilities: either $z-y \in G_1$ or $z-y\notin G_1$.

       Suppose $z-y \in G_1$. This implies $u_0-z \in G_1$ and $u_l-z \in G_1$. Then,  similar to the previous case, each node of the \cp{} $P$ is in $r_1\cup N(r_1, G_1)$.

       Suppose $z-y\notin G_1$. This implies $y \notin r_1$. Because, if $y\in r_1$ then since $r_1$ is a clique of $G_1$, we have $z-y \in G_1$, contradicting our assumption. Since $y \in r_1\cup N(r_1, G_1)$ and $y\notin r_1$, $y \in N(r_1, G_1)\setminus r_1$. Then, there must exist a node $z' \in r_1$ such that $y \in N(z', G_1)$. We show that each node of the \cp{} $P$ is either a neighbor of $z$ or a neighbor of $z'$.

       Suppose there exists a node $u_i$ in $P$ which is not a neighbor of $z$ in $G_1$, i.e., $u_i-z\notin G_1$.  Pick the least such $i$. $i > 0$ as $u_0-z \in G_1$. This implies $u_{i-1}-z \in G_1$. We claim that for any $j \geq i$, $z-u_j \notin G_1$. Suppose there exists a $j > i$ such that $z-u_j \in G_1$ ($j\neq i$ as we have assumed $u_i-z\notin G_1$). Pick the least $j$. Then $(z, u_{i-1}, u_i, \ldots, u_j, z)$ is a \cc{} in $G_1$ (there cannot be an edge between two non-adjacent nodes $u_a$ and $u_b$ of $P$, as $P$ is a \cp{}), a contradiction, as $G_1$ is a chordal graph. We further claim that for each $j \geq i-1$, $z'-u_j \in G_1$. 

       Suppose there exists $j \geq i-1$ such that $z'-u_j \notin G_1$. Pick the least such $j$. Now, pick the least $k >j$ such that $z'-u_k \in G_1$ (there must exist such $k$ because $u_l-z' \in G_1$). If $j = i-1$ then $C = (z', z, u_{i-1}, u_i, \ldots, u_k, z')$ is a \cc{} in $G_1$, and if $j>i-1$ then $C = (z', u_{j-1}, u_j, \ldots, u_k, z')$ is a \cc{} in $G_1$. Both give a contradiction, as $G_1$ is a chordal graph. This implies that for each $j \geq i-1$, $z'-u_j \in G_1$. 
       
       This further implies each node of $P$ is either a neighbor of $z$ or a neighbor of $z'$. Since both $z$ and $z'$ are in $r_1$,  nodes of $P$ are in $r_1\cup N(r_1, G_1)$. 
   \end{enumerate}

   Therefore, in all the possibilities, nodes of $P$ are in $r_1\cup N(r_1, G_1)$. This completes the proof of \cref{obs:cp-in-M1-is-a-cp-in-O1}.
\end{proof}
From \cref{item-3-of-lem:sharma2023results} of \cref{lem:sharma2023results}, for $u-v \in O_1$, if $u\rightarrow v \in O$ then either 
(a) $u\rightarrow v$ is strongly protected in $O$ (\cref{item-3a-of-lem:sharma2023results} of \cref{lem:sharma2023results}), or 
(b) there exists $x-y \in O_1$ such that $x\rightarrow y\in O$, and there exists a \tfp{} from $(x,y)$ to $(u,v)$ in $M_1$ with $(x,y) \neq (u,v)$ (\cref{item-3b-of-lem:sharma2023results} of \cref{lem:sharma2023results}), or
(c) there exists $x-y \in O_1$ such that $x\rightarrow y \in O$, and there exist \tfps{} from $(x, y)$ to $v$, and from $(v,u)$ to $x$ with $y\neq v$ and $u\neq x$ (\cref{item-3c-of-lem:sharma2023results} of \cref{lem:sharma2023results}). 
We show that (b) and (c) implies (a). 
\begin{enumerate}
    \item 
    Suppose $u-v \in O_1$ and $u\rightarrow v \in O$ because $u\rightarrow v$ obeys \cref{item-3b-of-lem:sharma2023results} of \cref{lem:sharma2023results}, i.e., there exists $x-y \in O_1$ such that $x\rightarrow y \in O$, and there exists a \tfp{} $P = (a_0 = x, a_1 =y, \ldots, a_{l-1} = u, a_l =v)$ from $(x,y)$ to $(u,v)$ in $M_1$ with $(x,y)\neq (u,v)$. Since $(x,y)\neq (u,v)$, $l \geq 2$. 
 Also, since $O_1$ is an induced subgraph of $M_1$, $u-v, x-y \in M_1$. From \cref{obs:undirected-tfp-in-M-is-a-cp}, $P$ is an undirected \cp{} in $M_1$. Since the skeleton of $M_1$ is $G_1$, $P$ is also a \cp{} in $G_1$. Also, since $u-v, x-y \in O_1$, $u,v,x,y \in V_{O_1} = r_1\cup N(r_1, G_1)$.
 From \cref{obs:cp-in-M1-is-a-cp-in-O1}, all the nodes of $P$ is in $r_1\cup N(r_1, G_1)$. Since $O$ is an induced subgraph of $M$, and $x\rightarrow y \in O$, $x\rightarrow y \in M$.
 Since $M_1$ is a projection of $M$, and $P$ is a \tfp{} in $M_1$ from $(x,y)$ to $(u,v)$ such that $x\rightarrow y \in M$, from \cref{obs:every-edge-of-triangle-free-path-is-directed}, for $0\leq i < l$, $a_i \rightarrow a_{i+1} \in M$. This further implies $a_{l-2}\rightarrow a_{l-1} \rightarrow a_l$ is an induced subgraph of $M$ (as the skeletons of $M[V_{G_1}]$ and $M_1$ are the same). Since nodes of $P$ are in $r_1\cup N(r_1, G_1) \subseteq V_O$, and $O$ is an induced subgraph of $M$, $a_{l-2}\rightarrow a_{l-1} \rightarrow a_l$ is an induced subgraph of $O$.  This implies $u\rightarrow v$ is strongly protected in $O$, as it is part of the induced subgraph  $a_{l-2}\rightarrow a_{l-1} \rightarrow a_l \in O$, as shown in \cref{fig:strongly-protected-edge}.a (recall that $a_{l-1} = u$ and $a_l =v$). 
 \item 
 Suppose $u-v \in O_1$ and $u\rightarrow v \in O$ because $u\rightarrow v$ obeys \cref{item-3c-of-lem:sharma2023results} of \cref{lem:sharma2023results}, i.e., there exists $x-y \in O_1$ such that $x\rightarrow y \in O$, and there exist \tfps{} $P_1 = (a_0 = x, a_1 =y, \ldots,  a_l =v)$ from $(x,y)$ to $v$ in $M_1$ with $y\neq v$, and $P_2 = (b_0 = v, b_1 = u, \ldots, b_m = x)$ from $(v,u)$ to $x$ with $u\neq x$. Since $y\neq v$ and $u \neq x$, $l,m \geq 2$. 
   Both $P_1$ and $P_2$ must be undirected paths in $M_1$, otherwise, concatenation of $P_1$ and $P_2$ will give a directed cycle in $M_1$, contradicting \cref{item-1-theorem-nec-suf-cond-for-MEC} of \cref{thm:nes-and-suf-cond-for-chordal-graph-to-be-an-MEC}. This implies nodes of $P_1$ and $P_2$ belong to the same undirected connected component $\mathcal{C}$ of $M_1$. From \cref{item-2-theorem-nec-suf-cond-for-MEC} of \cref{thm:nes-and-suf-cond-for-chordal-graph-to-be-an-MEC}, $\mathcal{C}$ is chordal. This implies $P_1$ and $P_2$ are \tfp{} in $\mathcal{C}$. From \cref{lem:tfp-is-cp-for-chordal-graph}, $P_1$ and $P_2$ are \cps{} in $\mathcal{C}$, and in $M_1$, as $\mathcal{C}$ is an \ucc{} of $M_1$.
 Since the skeleton of $M_1$ is the skeleton of $G_1$, $P_1$ and $P_2$ are \cps{} in $G_1$.
 Also, since $u-v, x-y \in O_1$, $u,v,x,y \in V_{O_1} = r_1\cup N(r_1, G_1) $. From \cref{obs:cp-in-M1-is-a-cp-in-O1}, all the nodes of $P_1$ and $P_2$ belong to $r_1\cup N(r_1, G_1)$. Since $P_1$ and $P_2$ are \cps{} in $M_1$, and $O_1 = M_1[r_1\cup N(r_1, G_1)]$, $P_1$ and $P_2$ are \cps{} in $O_1$.

 We claim that $u-a_{l-1} \in O_1$. We first show that $u-a_{l-1} \in G_1$.
Suppose $u-a_{l-1} \notin G_1$, Then, we claim that for all $i \leq l-1$, $u-a_i \notin G_1$. 
Suppose for some $j < l-1$, $u-a_j \in G_1$. Pick the highest such $j$. Then, $C = (u, a_j, a_{j+1}, \ldots, a_{l-1}, a_l= v, u)$ is a \cc{} in $G_1$, a contradiction, as $G_1$ is a chordal graph. Then, we claim that $C = (b_m = a_0, a_1, \ldots, a_l = b_0, b_1, \ldots, b_m)$ is a \cc{}. Suppose $C$ is not a \cc{} in $G_1$. Then, since $P_1$ and $P_2$ is a \cp{} in $G_1$, there must exist $i$ and $j$ such that $0 < i < l$, $0 < j < m$ and $a_i-b_j \in G_1$. Pick the highest such $i$ for which there exists a $j$ such that $0 < j < m$ and $a_i-b_j \in G_1$. Then, pick the smallest such $j$. As shown earlier, for all $i \leq l-1$, $u-a_i \notin G_1$. Therefore $j > 1$ (remember $u = b_1$). But, then $C = (a_i, a_{i+1}, \ldots, a_l=b_0, b_1, \ldots, b_j, a_i)$ is a \cc{} in $G_1$. In any way, we get a \cc{} in $G_1$, a contradiction, as $G_1$ is a chordal graph. This implies $u-a_{l-1} \in G_1$.

Since $u, a_{l-1}$ are in the same \ucc{} of $M_1$ (an MEC of $G_1$), and $u-a_{l-1} \in G_1$ (the skeleton of $M_1$), $u-a_{l-1} \in M_1$ (edge between two nodes of the same \ucc{} of an MEC must be undirected, otherwise, it will create a directed cycle and violate \cref{item-1-theorem-nec-suf-cond-for-MEC} of \cref{thm:nes-and-suf-cond-for-chordal-graph-to-be-an-MEC}). Since $O_1$ is an induced subgraph of $M_1$, and $u, a_{l-1} \in V_{O_1}$ (as shown earlier), $u-a_{l-1} \in O_1$.

   We now prove that $u\rightarrow v$ is strongly protected in $O$.
   We first show that for each $0\leq i < l$, $a_i \rightarrow a_{i+1} \in O$.
   Since $O$ is the projection of $M$, and $u\rightarrow v, x\rightarrow y \in O$, $u\rightarrow v, x\rightarrow y \in M$.
 Since $M_1$ is a projection of $M$, and $P_1$ is a \tfp{} in $M_1$ from $(x,y)$ to $v$ such that $x\rightarrow y \in M$, from \cref{obs:every-edge-of-triangle-free-path-is-directed}, for $0\leq i < l$, $a_i \rightarrow a_{i+1} \in M$.
 Since nodes of $P_1$  belong to $r_1\cup N(r_1, G_1) \subseteq V_O$, and $O$ is an induced subgraph of $M$, for $0\leq i < l$, $a_i \rightarrow a_{i+1} \in O$.

Since from the construction, the skeletons of $O[V_{O_1}]$ and $O_1$ are the same, either $u\rightarrow a_{l-1} \in O$, or $u-a_{l-1} \in O$, or $u\leftarrow a_{l-1} \in O$. We show that in each possibility, $u\rightarrow v$ is strongly protected in $O$.
If $u\rightarrow a_{l-1} \in O$ then $u\rightarrow v$ is strongly protected in $O$ as it is part of an induced subgraph $u\rightarrow a_{l-1} \rightarrow v \leftarrow u$, as shown in \cref{fig:strongly-protected-edge}.c. We now show that in other possibilities also $u\rightarrow v$ is strongly protected.

Suppose either $u-a_{l-1} \in O$ or $u\leftarrow a_{l-1} \in O$.
We claim that $b_2\rightarrow b_1 \in O$. 
Suppose $b_2 \rightarrow b_1 \notin O$.  Since $P_2$ is a \tfp{} in $M_1$, the subpath $P_2' = (b_2 = u, b_3, \ldots, b_m)$ of $P$ is a \tfp{} in $M_1$.  
If $b_2\rightarrow b_1 \notin M$ then from \cref{obs:cond-for-ud-path-in-M_a-to-be-ud-path-in-M}, $P_2'$ is a \tfp{} in $M$. As discussed earlier, each edge of the \tfp{} $P_1$ is directed in $O$. This implies a subpath of $P$, $P_1' = (a_0, a_1, \ldots, a_{l-1})$, is a directed \tfp{} in $O$. Also, from our assumption, either $a_{l-1}-b_2 \in O$ or $a_{l-1}\rightarrow b_2 \in O$. 
Since $O$ is an induced subgraph of $M$, concatenaiton of $P_1'$, $a_{l-1}-b_2$ (or $a_{l-1} \rightarrow b_2$), and $P_2'$ makes a directed cycle in $M$, contradicting \cref{item-1-theorem-nec-suf-cond-for-MEC} of \cref{thm:nes-and-suf-cond-for-chordal-graph-to-be-an-MEC}. This implies if $u-a_{l-1} \in O$ or $u\leftarrow a_{l-1} \in O$ then $b_2\rightarrow b_1 \in O$. But, then, $u\rightarrow v$ is strongly protected in $O$ due to part of an induced subgraph of the form $b_2\rightarrow (b_1 = u) \rightarrow (b_0 = v)$. 
\end{enumerate}

Thus, we show that for chordal graph $G$, \cref{item-3b-of-lem:sharma2023results,item-3c-of-lem:sharma2023results} of \cref{lem:sharma2023results} imply \cref{item-3a-of-lem:sharma2023results} of \cref{lem:sharma2023results}. 
This completes the proof \cref{item-3-of-obs:nes-cond-for-chordal-graph} of \cref{obs:nes-cond-for-chordal-graph}.
\end{proof}

\begin{proof}[Proof of \cref{obs:M-is-a-chain-graph}]
We start with the following claim and observations:
\begin{claim}
    \label{claim:no-directed-cycle-with-nodes-in-ucc}
    For $i\in \{1,2\}$, for any three nodes $u$, $v$, and $w$ belonging to an \ucc{} $\mathcal{C}$ of $M_i$, there does not exist a directed cycle $(u, v, w, u)$ in $M$.
\end{claim}

\begin{observation}
            \label{obs:directed-edge-in-M1-is-directed-in-M}
            For $i\in \{1,2\}$, if $u\rightarrow v \in M_i$ then $u\rightarrow v \in M$.
        \end{observation}
        \begin{proof}
        W.l.o.g., let us assume $a = 1$.
        Suppose $u\rightarrow v \in M_1$. 
            From \cref{def:Markov-union-of-graphs}, at step \ref{step-1-of-construction-of-M} of the construction of $M$, $u\rightarrow v \in M$. At step \ref{step-2-of-construction-of-M}, we replace an undirected edge with a directed one. This implies once an edge is directed in $M$ at step \ref{step-1-of-construction-of-M}, it remains directed in $M$. This implies $u\rightarrow v \in M$. This completes the proof of \cref{obs:directed-edge-in-M1-is-directed-in-M}.
        \end{proof}
        \begin{corollary}
            \label{corr:directed-edge-in-M-is-either-directed-or-undirected-in-M1}
            For $i\in \{1,2\}$, $u,v \in V_{G_i}$, if $u\rightarrow v \in M$ then either $u-v \in M_i$ or $u\rightarrow v \in M_i$.
        \end{corollary}
        \begin{proof}
         W.l.o.g., let us assume $a = 1$.
            Suppose $u,v \in V_{G_1}$ and $u\rightarrow v \in M$.
            From the construction of $M$, $\skel{M[V_{G_1}]} = \skel{M_1}$. This implies either $u-v \in M_1$, or $u\rightarrow v \in M_1$, or $v\rightarrow u \in M_1$. But, if $v\rightarrow u \in M_1$ then from \cref{obs:directed-edge-in-M1-is-directed-in-M}, $v\rightarrow u \in M$, a contradiction, as $u\rightarrow v \in M$. Thus, the options that remain are $u-v \in M_1$ and $u\rightarrow v \in M_1$. This completes the proof of \cref{corr:directed-edge-in-M-is-either-directed-or-undirected-in-M1}.
        \end{proof}
        \begin{corollary}
            \label{corr:undirected-edge-in-M-is-undirected-in-M1}
            For $i\in \{1,2\}$, for $u,v \in V_{G_i}$, if $u-v \in M$ then $u-v \in M_i$.
        \end{corollary}
        \begin{proof}
        W.l.o.g., let us assume $i = 1$.
            Suppose $u,v \in V_{G_1}$ and $u- v \in M$.
            From the construction of $M$, $\skel{M[V_{G_1}]} = \skel{M_1}$. This implies either $u-v \in M_1$, or $u\rightarrow v \in M_1$, or $v\rightarrow u \in M_1$. But, if $u\rightarrow v \in M_1$ or $v\rightarrow u \in M_1$ then from \cref{obs:directed-edge-in-M1-is-directed-in-M}, $u\rightarrow v \in M$ or $v\rightarrow u \in M$, a contradiction, as $u- v \in M$. Thus, the option that remains is $u-v \in M_1$.
            This completes the proof of \cref{corr:undirected-edge-in-M-is-undirected-in-M1}.
        \end{proof}
The above claim, observation, and corollaries imply the following:
\begin{observation}
    \label{obs:no-directed-cycle-in-Mi}
    For $i\in \{1,2\}$, for any three nodes $u$, $v$, and $w$ belonging to $V_{G_i}$, there does not exist a directed cycle  $(u, v, w, u)$ in $M$.
\end{observation}
\begin{proof}
    Contrary to \cref{obs:no-directed-cycle-in-Mi}, suppose for some $i\in \{1,2\}$, there exist nodes $u$, $v$, and $w$ belonging to $V_{G_i}$ such that $C = (u, v, w, u)$ is a directed cycle in $M$. From \cref{corr:directed-edge-in-M-is-either-directed-or-undirected-in-M1,corr:undirected-edge-in-M-is-undirected-in-M1}, $C$ is a cycle in $M_i$. Since $M_i$ is an MEC, from \cref{item-1-theorem-nec-suf-cond-for-MEC} of \cref{thm:nes-and-suf-cond-for-chordal-graph-to-be-an-MEC}, none of the edges of the cycle is directed in $M_i$. This implies all the node belongs to the same \ucc{} of $M_i$. This contradicts \cref{claim:no-directed-cycle-with-nodes-in-ucc}. This implies our assumption is wrong. This proves \cref{obs:no-directed-cycle-in-Mi}.
\end{proof}

\Cref{obs:no-directed-cycle-in-Mi} implies the following:
\begin{observation}
    \label{obs:no-directed-cycle-of-length-three}
    There do not exist three nodes $u$, $v$, and $w$ in $M$ such that $(u, v, w, u)$ forms a directed cycle in $M$.
\end{observation}
\begin{proof}
    Contrary to \cref{obs:no-directed-cycle-of-length-three}, let us assume that there exist $u$, $v$, and $w$ in $M$ such that $(u, v, w, u)$ forms a directed cycle in $M$.    
    We claim that either all the nodes of the cycle belong to $V_{G_1}$ (i.e., $u, v, w \in V_{G_1}$) or all the nodes belong to $V_{G_2}$ (i.e., $u, v, w \in V_{G_2}$).
    
    Suppose neither all the nodes of the cycle are in $V_{G_1}$ nor all the nodes are in $V_{G_2}$. This implies there exists a node in the cycle that belongs to $V_{G_1}\setminus V_{G_2}$ and another node in the cycle that belongs to $V_{G_2}\setminus V_{G_1}$. However, from the construction, there are no edges between $V_{G_1}\setminus V_{G_2}$ and $V_{G_2}\setminus V_{G_1}$ in $G$ because $I = V_{G_1}\cap V_{G_2}$ is a vertex separator that separates $V_{G_1}\setminus I$ and $V_{G_2}\setminus I$ in $G$. According to \cref{obs:skeleton-of-M-is-G}, $G$ is the skeleton of $M$. This implies there cannot be an edge between nodes in $V_{G_1}\setminus V_{G_2}$ and $V_{G_2}\setminus V_{G_1}$. But, in $M$, there exists an edge between each pair of nodes in the cycle. This implies that the scenario where a node of the cycle is in $V_{G_1}\setminus V_{G_2}$ and another in $V_{G_2}\setminus V_{G_1}$ cannot occur. Therefore, all the nodes of the cycle must either be in $V_{G_1}$ or in $V_{G_2}$.
    
    From \cref{obs:no-directed-cycle-in-Mi}, there cannot exist a directed cycle $(u, v, w, u)$ in $M$ such that all the nodes are either in $V_{G_1}$ or $V_{G_2}$. Hence, our assumption is wrong. This completes the proof of \cref{obs:no-directed-cycle-of-length-three}.
\end{proof}
We now demonstrate that \cref{obs:no-directed-cycle-of-length-three} implies the absence of any directed cycle in $M$. Let's assume the existence of a directed cycle in $M$. Select the smallest-length directed cycle $C = (u_0, u_1, \ldots, u_l = u_0)$ in $M$. Our claim is that $l = 2$.

Suppose $l > 2$. As $G$ is the skeleton of $M$, $C = (u_0, u_1, \ldots, u_l, u_{l+1} = u_0)$ forms a cycle in $G$. Since $G$ is chordal, an edge in $G$ must exist between two non-adjacent nodes $u_i$ and $u_j$ of $C$. Without loss of generality, let's assume $i < j$. Then, either $u_i-u_j \in M$ or $u_i\rightarrow u_j \in M$ or $u_j\rightarrow u_i \in M$. Consequently, either $(u_0, u_1, \ldots, u_i, u_j, u_{j+1}, \ldots, u_l, u_{l+1} = u_0)$ or $(u_j, u_i, u_{i+1}, \ldots, u_j)$ forms a directed cycle in $M$ with a smaller length than cycle $C$. This contradicts our choice of picking the smallest directed cycle. Thus, $l = 2$.

From \cref{obs:no-directed-cycle-of-length-three}, there cannot exist a directed cycle of length two in $M$. Therefore, no directed cycle exists in $M$, indicating that $M$ is a chain graph. This completes the proof of \cref{obs:M-is-a-chain-graph}.
\end{proof}

\begin{proof}[Proof of \cref{claim:no-directed-cycle-with-nodes-in-ucc}]
The following two claims prove \cref{claim:no-directed-cycle-with-nodes-in-ucc}.
\begin{claim}
    \label{claim:no-directed-cycle-with-nodes-in-ucc-without-ri}
    For $i\in \{1,2\}$, for any three nodes $u$, $v$, and $w$ belonging to an \ucc{} $\mathcal{C}$ of $M_i$ such that $V_{\mathcal{C}}\cap r_i = \emptyset$, there does not exist a directed cycle $(u, v, w, u)$ in $M$.
\end{claim}
\begin{claim}
    \label{claim:no-directed-cycle-with-nodes-in-ucc-with-ri}
    For $i\in \{1,2\}$, for any three nodes $u$, $v$, and $w$ belonging to an \ucc{} $\mathcal{C}$ of $M_i$ such that $V_{\mathcal{C}}\cap r_i \neq \emptyset$, there does not exist a directed cycle $(u, v, w, u)$ in $M$.
\end{claim}

For any $i\in \{1,2\}$, an \ucc{} of $M_i$ can be of two types: either $V_{\mathcal{C}}\cap r_i = \emptyset$ or $V_{\mathcal{C}}\cap r_i \neq \emptyset$. \Cref{claim:no-directed-cycle-with-nodes-in-ucc-without-ri,claim:no-directed-cycle-with-nodes-in-ucc-with-ri} show that in both of the possibilities, \cref{claim:no-directed-cycle-with-nodes-in-ucc} is valid. This completes the proof of \cref{claim:no-directed-cycle-with-nodes-in-ucc}.
\end{proof}

\begin{proof}[Proof of \cref{claim:no-directed-cycle-with-nodes-in-ucc-without-ri}]
Contrary to \cref{claim:no-directed-cycle-with-nodes-in-ucc-without-ri}, let us assume that for some $i \in \{1,2\}$, there exists $u$, $v$, and $w$ belonging to an \ucc{} $\mathcal{C}$ of $M_i$ such that $V_{\mathcal{C}}\cap r_i = \emptyset$, and $(u,v,w,u)$ is a directed cycle in $M$. From the construction of $M$, $\skel{M[V_{\mathcal{C}}]} = \mathcal{C}$. This implies $u-v, v-w, w-c \in \mathcal{C}$.
Before moving ahead, we go through the following claim.

\begin{claim}
\label{claim:ucc-of-Mi-not-in-r1-is-a-ucc-in-M}
For $i\in \{1,2\}$, let $\mathcal{C}$ be an \ucc{} of $M_i$ such that $V_{\mathcal{C}} \cap r_i = \emptyset$. Then, for $u-v \in \mathcal{C}$, $u-v \in M$.
\end{claim}

From \cref{claim:ucc-of-Mi-not-in-r1-is-a-ucc-in-M},  $u-v, v-w, w-c \in M$. This implies $(u, v, w, u)$ is an undirected cycle in $M$, contradicting our assumption. This implies that our assumption is wrong. This completes the proof of \cref{claim:no-directed-cycle-with-nodes-in-ucc-without-ri}.
\end{proof}

\begin{proof}[Proof of \cref{claim:ucc-of-Mi-not-in-r1-is-a-ucc-in-M}]
Contrary to  \cref{claim:ucc-of-Mi-not-in-r1-is-a-ucc-in-M}, let us assume that for some $i \in \{1,2\}$, there exists an \ucc{} $\mathcal{C}$ of $M_i$ such that $V_{\mathcal{C}}\cap r_i =\emptyset$, and there exists $u-v \in \mathcal{C}$ such that $u\rightarrow v \in M$. W.l.o.g., let us assume $i = 1$.

Before moving ahead, we go through the following observation:
  \begin{observation}
      \label{obs:for-u-v-in-C-if-u->v-in-M-then-either-step-1-or-step-2-obeyed}
      For $i \in \{1,2\}$, let $\mathcal{C}$ be an \ucc{} of $M_i$, and $u-v$ be an edge of $\mathcal{C}$ such that $u\rightarrow v \in M$. Then, $u\rightarrow v$ has been added to $M$ either at step \ref{step-1-of-construction-of-M} of the construction of $M$, or at step \ref{step-2-of-construction-of-M} of the construction of $M$ when $\mathcal{C}$ has been considered.
  \end{observation}

From \cref{obs:for-u-v-in-C-if-u->v-in-M-then-either-step-1-or-step-2-obeyed}, if $u\rightarrow v$ has been added to $M$ then it must have been added either at step \ref{step-1-of-construction-of-M} of the construction of $M$, or at step \ref{step-2-of-construction-of-M} of the construction of $M$, when $\mathcal{C}$ has been considered.
Since  $V_{\mathcal{C}}\cap r_1 = \emptyset$, $\mathcal{C}$ has not been considered at step \ref{step-2-of-construction-of-M} of the construction of $M$. This implies $u\rightarrow v$ has been added to $M$ at step \ref{step-1-of-construction-of-M} of the construction of $M$. This implies $u\rightarrow v \in U_M(M_1, M_2, O)$.

Since $\mathcal{C}$ is an \ucc{} of $M_1$, each edge of $\mathcal{C}$ is an undirected edge of $M_1$. This implies $u-v \in M_1$. Since  $u\rightarrow v \in U_M(M_1, M_2, O)$,  the following observation implies $u\rightarrow v\in O$.

  \begin{observation}
        \label{obs:undirected-in-M_1-and-directed-in-M-implies-directed-in-O}
        If $u-v\in M_1$ and $u\rightarrow v \in U_M(M_1, M_2, O)$ (i.e., $u\rightarrow v$ added to $M$ at step \ref{step-1-of-construction-of-M} of the construction of $M$) then $u\rightarrow v \in O$.
    \end{observation}
From \cref{obs:undirected-in-M_1-and-directed-in-M-implies-directed-in-O},  $u\rightarrow v \in O$.
This further implies that for $u-v \in \mathcal{C}$ if $u\rightarrow v \in M$ then $u\rightarrow v \in O[V_{\mathcal{C}}\cap V_O]$, contradicting \cref{obs:for-x-y-in-C-x-y-in-O}.

\begin{observation}
    \label{obs:for-x-y-in-C-x-y-in-O}
    Let $\mathcal{C}$ be an \ucc{} of $M_1$ such that $r_1\cap V_{\mathcal{C}} = \emptyset$. Then, $O[V_{\mathcal{C}}\cap V_O]$ is an undirected graph.
\end{observation}
 This implies our assumption was wrong that $u\rightarrow v \in M$. 
This completes the proof of
\cref{claim:ucc-of-Mi-not-in-r1-is-a-ucc-in-M}.
\end{proof}

  \begin{proof}[Proof of \cref{obs:for-u-v-in-C-if-u->v-in-M-then-either-step-1-or-step-2-obeyed}]
  W.l.o.g., let us assume $i = 1$.
  If $u\rightarrow v$ has been added to $M$ at step \ref{step-1-of-construction-of-M} of the construction of $M$ then we are done. Suppose $u\rightarrow v$ has not been added to $M$ at step \ref{step-1-of-construction-of-M} of the construction of $M$. Then, it must have added to $M$ at step \ref{step-2-of-construction-of-M} when some \ucc{} $\mathcal{C}'$ of $M_1$ or $M_2$ is considered. We show at step \ref{step-2-of-construction-of-M} of the construction of $M$, when any other \ucc{} is considered, $u\rightarrow v$ has not been added to $M$. 

      Since two different \ucc{} of an MEC does not share any vertices, for any other \ucc{} $\mathcal{C}'$ of $M_1$, $V_{\mathcal{C}'}\cap V_{\mathcal{C}} = \emptyset$. This implies $u-v \notin \mathcal{C}'$ and $u\rightarrow v$ has  not been added to $M$ at step \ref{step-2-of-construction-of-M} when $\mathcal{C}' $ has been considered. 
      
      For any \ucc{} $\mathcal{C}'$ of $M_2$, $V_{\mathcal{C}'}\cap V_{\mathcal{C}} \subseteq r_1\cap r_2 \subseteq r_2\cup N(r_2, G)$. This implies if $u-v \in \mathcal{C}'$ then $u,v \in r_2\cup N(r_2, G_2)$. But, at step \ref{step-2-of-construction-of-M} of the construction of $M$, when $\mathcal{C}'$ has been considered, we can added $u\rightarrow v$ to $M$  only if $v\notin r_2\cup N(r_2, G_2)$. This implies $u\rightarrow v$ has not been added to $M$ at step \ref{step-2-of-construction-of-M} of the construction of $M$, when $\mathcal{C}'$ has been considered.
      
      The above discussion implies that for $u-v \in \mathcal{C}$, if $u\rightarrow v$ has been added to $M$ then it must have been added either at step \ref{step-1-of-construction-of-M} of the construction of $M$ or at step \ref{step-2-of-construction-of-M} of the construction of $M$ when $\mathcal{C}$ has been considered. 
  \end{proof}

\begin{proof}[Proof of \cref{obs:undirected-in-M_1-and-directed-in-M-implies-directed-in-O}]
    From \cref{def:Markov-union-of-graphs},  if $u-v\in M_1$ and $u\rightarrow v \in U_M(M_1, M_2,O)$ then either $u\rightarrow v \in M_2$ or $u\rightarrow v \in O$.
    $u\rightarrow v \in M_2$ implies that $u,v \in V_{M_1} \cap V_{M_2} = I$.
    Since $O_2 = M_2[V_{O_2}]$, for $u,v \in I$,  $u\rightarrow v \in M_2$ implies $u\rightarrow v \in O_2$, as from the construction,  $I \subseteq V_{O_2}$.
    Since $O\in \mathcal{E}(O_1, O_2)$, therefore, from \cref{item-1-of-def:extensions} of \cref{def:extension}, we have $u\rightarrow v \in O$.
    Thus, we can say that if $u-v\in M_1$ and $u\rightarrow v \in U_M(M_1, M_2,O)$ then $u\rightarrow v \in O$.
  \end{proof} 

  \begin{proof}[Proof of \cref{obs:for-x-y-in-C-x-y-in-O}]
    Suppose $\mathcal{C}$ is an \ucc{} of $M_1$ such that $r_1\cap V_{\mathcal{C}} = \emptyset$.  Contrary to \cref{obs:for-x-y-in-C-x-y-in-O}, let us assume that $O[V_{\mathcal{C}}\cap V_O]$ is not an undirected graph, i.e., there exists a directed edge $u\rightarrow v \in O[V_{\mathcal{C}}\cap V_O]$.
    Pick $u$ and $v$ using the following steps:
    \begin{enumerate}
        \item
        \label{choice-of-u}
        Pick $u$ such that there is no incoming edge to $u$ in $O[V_{\mathcal{C}}\cap V_O]$. 
        \item
        \label{choice-of-v}
        Pick $v$ from the set $Y = \{y:  u\rightarrow y  \in O[V_{\mathcal{C}}\cap V_O]\}$ such that for any $y \in Y$, $y\rightarrow v \notin O$.
    \end{enumerate}
    The existence of such $u$ and $v$ is because $O$ is a chain graph.
Since $u\rightarrow v \in O[V_{\mathcal{C}}\cap V_O]$,  $u\rightarrow v \in O$. From \cref{item-3-of-def:extension} of \cref{def:extension}, $u\rightarrow v$ must be strongly protected in $O$. 
    We now show that $u\rightarrow v$ is not strongly protected in $O$. This leads to a contradiction, which further implies that $O[V_{\mathcal{C}}\cap V_O]$ is an undirected graph.
 
    We go through the following claims and observations to prove $u\rightarrow v$ is not strongly protected in $O$.
    \begin{claim}
        \label{claim:no-x->u-exists}
        For any $x \in  V_{\mathcal{C}}$, $x\rightarrow u \notin O$.
    \end{claim}
    \begin{proof}[Proof of \cref{claim:no-x->u-exists}]
        Contrary to \cref{claim:no-x->u-exists}, suppose there exists an $x \in  V_{\mathcal{C}}$ such that $x\rightarrow u \in O$. This implies $x \in V_{\mathcal{C}}\cap V_O$, and $x\rightarrow u \in O[V_{\mathcal{C}}\cap V_O]$, contradicting our choice of $u$ (recall that from the choice of $u$ (\cref{choice-of-u}), there is no incoming edge to $u$ in $O[V_{\mathcal{C}}\cap V_O]$). This implies that our assumption is wrong. This completes the proof of \cref{claim:no-x->u-exists}. 
    \end{proof}

    \begin{observation}
    \label{obs:neighbors-of-C-with-no-ri-is-in-Mi}
    Let $\mathcal{C}$ be an \ucc{} of $M_1$ such that $V_{\mathcal{C}}\cap r_1 =\emptyset$. Then, for any $x \in V_{\mathcal{C}}$, $N(x, G) \subseteq V_{G_1}$.
\end{observation}
\begin{proof}[Proof of \cref{obs:neighbors-of-C-with-no-ri-is-in-Mi}]
    Suppose $\mathcal{C}$ is an \ucc{} of $M_1$ such that $V_{\mathcal{C}}\cap r_1 =\emptyset$, $x$ is a node in $V_{\mathcal{C}}$, and $y$ is a neighbor of $x$ in $G$. 
    We show that $y \in V_{G_1}$.

    Since $V_{\mathcal{C}}\cap r_1 =\emptyset$, and $\mathcal{C}$ is an \ucc{} of $M_1$, $x \in V_{G_1}\setminus r_1$.
    Since $I = r_1\cap r_2$ is a vertex separator of $G$ that separates $V_{G_1}\setminus I$ and $V_{G_2}\setminus I$, $x \in V_{G_1}\setminus I$. This implies $y\notin V_{G_2}\setminus I$. Since $V_G$ is a disjoint union of $V_{G_1}$ and $V_{G_2}\setminus I$, this further implies $y \in V_{G_1}$. This implies $N(x, G) \subseteq V_{G_1}$. This completes the proof of \cref{obs:neighbors-of-C-with-no-ri-is-in-Mi}.
\end{proof}

    \begin{observation}
    \label{obs:directed-edge-towards-u-implies-directed-edge-towards-v}
    Let $\mathcal{C}$ be an \ucc{} of $M_1$. For $u-v \in \mathcal{C}$, if $x\rightarrow u \in M_1$ then $x\rightarrow v \in M_1$.
\end{observation}
\begin{proof}[Proof of \cref{obs:directed-edge-towards-u-implies-directed-edge-towards-v}]
    Suppose $\mathcal{C}$ is an \ucc{} of $M_1$, there exists an undirected edge $u-v \in \mathcal{C}$, and there exists a directed edge $x\rightarrow u \in M_1$. There are two possibilities: either $x-v \in G_1$ or $x-v \notin G_1$ (remember $G_1$ is the skeleton of $M_1$). If $x-v \notin G_1$ then  $x\rightarrow u-v$ is an induced subgraph of $M_1$, contradicting \cref{item-3-theorem-nec-suf-cond-for-MEC} of \cref{thm:nes-and-suf-cond-for-chordal-graph-to-be-an-MEC}. This implies $x-v \in G_1$. Then, there must be $x\rightarrow v \in M_1$, otherwise, $(x, u, v, x)$ is a directed cycle in $M$, contradicting \cref{item-1-theorem-nec-suf-cond-for-MEC} of \cref{thm:nes-and-suf-cond-for-chordal-graph-to-be-an-MEC}. This completes the proof of \cref{obs:directed-edge-towards-u-implies-directed-edge-towards-v}.
\end{proof}
    
    \begin{claim}
        \label{claim:no-x->u-implies-y->u}
        For any $x \notin V_{\mathcal{C}}$, $x\rightarrow u \in O$ if, and only if, $x\rightarrow v \in O$.
    \end{claim}
    \begin{proof}[Proof of \cref{claim:no-x->u-implies-y->u}]
    From \cref{obs:neighbors-of-C-with-no-ri-is-in-Mi}, the neighbors of $\mathcal{C}$ in $M$ is in $V_{G_1}$. Since $u \in V_{\mathcal{C}}$, $v \in V_{G_1}$.
        For any $x \notin V_{\mathcal{C}}$, from  \cref{corr:directed-edge-in-M-is-either-directed-or-undirected-in-M1}, if $x\rightarrow u \in M$ then $x\rightarrow v \in M_1$ (if $x-u \in M_1$ then $x$ and $u$ are in the same \ucc{} $\mathcal{C}$ of $M_1$). From \cref{obs:directed-edge-towards-u-implies-directed-edge-towards-v}, for any $x \notin V_{\mathcal{C}}$, if $x\rightarrow u \in M_1$ then $x\rightarrow v \in M_1$. 
    From \cref{obs:directed-edge-in-M1-is-directed-in-M}, if $x\rightarrow v \in M_1$ then $x\rightarrow v \in M$. This implies that for $x \notin V_{\mathcal{C}}$,  if $x\rightarrow u \in M$ then $x\rightarrow v \in M$. 
    Since $O$ is an induced subgraph of $M$, and $v\in V_O$,  for $x \notin V_{\mathcal{C}}$, if $x\rightarrow u \in O$ then $x\rightarrow v \in O$.
    Similarly, we can show that for $x \notin V_{\mathcal{C}}$,  if $x\rightarrow v \in O$ then $x\rightarrow u \in O$. This implies that for $x \notin V_{\mathcal{C}}$,  $x\rightarrow u \in O$ if, and only if,  $x\rightarrow v \in O$.
    \end{proof}

    \begin{claim}
        \label{claim:u->v-does-not-obey-fig-sp-a}
        $u\rightarrow v$ cannot be part of an induced subgraph $w\rightarrow u\rightarrow v$ in $O$ (as shown in \cref{fig:strongly-protected-edge}.a).
    \end{claim}
    \begin{proof}[Proof of \cref{claim:u->v-does-not-obey-fig-sp-a}]
    From \cref{claim:no-x->u-exists}, if $w\in \mathcal{C}$ then $w\rightarrow u \notin O$. And, from \cref{claim:no-x->u-implies-y->u}, if $w\notin \mathcal{C}$ and $w\rightarrow u \in O$ then $w\rightarrow v \in O$. This implies  $w\rightarrow u\rightarrow v$ is not an induced subgraph in $O$. This completes the proof of \cref{claim:u->v-does-not-obey-fig-sp-a}.
    \end{proof}

    \begin{claim}
        \label{claim:u->v-does-not-obey-fig-sp-b}
        $u\rightarrow v$ cannot be part of an induced subgraph $w\rightarrow v\leftarrow u$ in $O$ (as shown in \cref{fig:strongly-protected-edge}.b).
    \end{claim}
    \begin{proof}[Proof of \cref{claim:u->v-does-not-obey-fig-sp-b}]
    Suppose $u\rightarrow v$ is part of an induced subgraph $w\rightarrow v\leftarrow u$ in $O$.
       Since $w \rightarrow v \in O$ and $w\rightarrow u \notin O$, from \cref{claim:no-x->u-implies-y->u}, $w \in \mathcal{C}$. This implies $u,v, w \in V_{M_1}$, as $\mathcal{C}$ is an \ucc{} of $M_1$. Since $u,v,w \in \mathcal{C}$, and from the construction, $\skel{M[V_{G_1}]} = \skel{M_1}$, we have $u-v-w$ is an induced subgraph in $M_1$. This further implies that $u-v-w$ is an induced subgraph in $O_1$ as $V_{O_1} = V_{M_1}\cap V_{O}$, and $O_1$ is an induced subgraph of $M_1$. This contradicts \cref{item-2-of-def:extension} of \cref{def:extension}. This implies that our assumption is wrong. This completes the proof of \cref{claim:u->v-does-not-obey-fig-sp-b}.
    \end{proof}

    \begin{claim}
        \label{claim:u->v-does-not-obey-fig-sp-c}
         $u\rightarrow v$ cannot be part of an induced subgraph $u\rightarrow w \rightarrow v\leftarrow u$ in $O$ (as shown in \cref{fig:strongly-protected-edge}.c).
    \end{claim}
    \begin{proof}[Proof of \cref{claim:u->v-does-not-obey-fig-sp-c}]
        Suppose $u\rightarrow v$ is part of an induced subgraph $u\rightarrow w \rightarrow v\leftarrow u$ in $O$.
        Since $w\rightarrow v \in O$ and $w\rightarrow u \notin O$, from \cref{claim:no-x->u-implies-y->u}, $w \in \mathcal{C}$. This implies $u,v,w \in V_O \cap V_{\mathcal{C}}$. From the choice of $u$ and $v$ that we have (\cref{choice-of-u,choice-of-v}) there cannot exist any $w$ such that $u\rightarrow w \rightarrow v\leftarrow u\in O$. This implies that our assumption is wrong. This completes the proof of \cref{claim:u->v-does-not-obey-fig-sp-c}.
    \end{proof}

    \begin{claim}
        \label{claim:u->v-does-not-obey-fig-sp-d}
         $u\rightarrow v$ cannot be part of an induced subgraph $u\rightarrow v \leftarrow w -u-w'\rightarrow v$ in $O$ (as shown in \cref{fig:strongly-protected-edge}.d).
    \end{claim}
    \begin{proof}[Proof of \cref{claim:u->v-does-not-obey-fig-sp-d}]
         Suppose $u\rightarrow v$ is part of an induced subgraph $u\rightarrow v \leftarrow w -u-w'\rightarrow v$ in $O$.
        Since $w\rightarrow v, w' \rightarrow v\in O$ and $w\rightarrow u, w'\rightarrow u \notin O$, from \cref{claim:no-x->u-implies-y->u}, $w, w' \in \mathcal{C}$. This implies $u,v, w,w' \in V_{M_1}$, as $\mathcal{C}$ is an \ucc{} of $M_1$. Since $w,v,w' \in \mathcal{C}$, and from the construction, $\skel{M[V_{G_1}]} = \skel{M_1}$, we have $w-v-w'$ is an induced subgraph in $M_1$. This further implies that $w-v-w'$ is an induced subgraph in $O_1$ as $V_{O_1} = V_{M_1}\cap V_{O}$, and $O_1$ is an induced subgraph of $M_1$. This contradicts \cref{item-2-of-def:extension} of \cref{def:extension}. This implies that our assumption is wrong. This completes the proof of \cref{claim:u->v-does-not-obey-fig-sp-d}.
    \end{proof}

    \Cref{claim:u->v-does-not-obey-fig-sp-a,claim:u->v-does-not-obey-fig-sp-b,claim:u->v-does-not-obey-fig-sp-c,claim:u->v-does-not-obey-fig-sp-d} imply that $u\rightarrow v$ is not strongly protected in $O$, contradicting \cref{item-3-of-def:extension} of \cref{def:extension}. This implies our assumption that there exists a directed edge in $O[V_{\mathcal{C}}\cap V_O]$ is wrong. This further implies $O[V_{\mathcal{C}}\cap V_O]$ is an undirected graph. This completes the proof of \cref{obs:for-x-y-in-C-x-y-in-O}.
\end{proof}

\begin{proof}[Proof of \cref{claim:no-directed-cycle-with-nodes-in-ucc-with-ri}]
Contrary to \cref{claim:no-directed-cycle-with-nodes-in-ucc-with-ri}, suppose for some $i\in \{1,2\}$, there exist nodes $u$, $v$, and $w$ belonging to an \ucc{} $\mathcal{C}$ of $M_i$ such that $V_{\mathcal{C}}\cap r_i \neq \emptyset$, and  $(u, v, w, u)$ is a directed cycle in $M$. W.l.o.g., let us assume $i = 1$.

\begin{claim}
    \label{claim:no-node-of-cycle-is-in-r_1}
    No node of the cycle belongs to $r_1$.
\end{claim}
\begin{proof}
    If a node of the cycle belongs to $r_1$ then all the nodes of the cycle are in $r_1\cup N(r_1, G_1) \subseteq V_O$. From \cref{obs:O-is-an-induced-subgraph-of-M}, $O$ is an induced subgraph of $M$. This implies $O$ contains the directed cycle. But, this is a contradiction as $O$ is a partial MEC (\cref{def:partial-MEC}), and a partial MEC is a chain graph. This implies no node of the cycle belongs to $r_1$.
\end{proof}

\begin{claim}
    \label{claim:one-node-of-the-cycle-is-not-in-r1-and-its-neighbors}
    At least one node of the cycle is not in $r_1\cup N(r_1, G_1)$.
\end{claim}
\begin{proof}
    Contrary to \cref{claim:one-node-of-the-cycle-is-not-in-r1-and-its-neighbors}, suppose all the nodes of the cycle belong to $r_1 \cup N(r_1, G_1)$. Since $r_1 \cup N(r_1, G_1) \subseteq V_O$. From \cref{obs:O-is-an-induced-subgraph-of-M}, $O$ is an induced subgraph of $M$. This implies $O$ contains the directed cycle. But, this is a contradiction as $O$ is a partial MEC (\cref{def:partial-MEC}), and a partial MEC is a chain graph. This implies our assumption is incorrect. This validates \cref{claim:one-node-of-the-cycle-is-not-in-r1-and-its-neighbors}.
\end{proof}

Since $V_{\mathcal{C}}\cap r_1 \neq \emptyset$, while running step \ref{step-2-of-construction-of-M} of the construction of $M$, at some iteration, $\mathcal{C}$ gets considered at step \ref{step-2-of-construction-of-M} of the construction of $M$. Let $\tau$ be the LBFS ordering of $\mathcal{C}$ constructed at step \ref{step-2-a-of-construction-of-M} of the construction of $M$. Since $\tau$ starts with $V_{\mathcal{C}}\cap r_1$ (from step \ref{step-2-a-of-construction-of-M} of the construction of $M$), the following claims are intuitive:
\begin{claim}
    \label{claim:nodes-of-r1-comes-before-others-in-tau}
    For two nodes $x,y \in \mathcal{C}$, if $x \in V_{\mathcal{C}}\cap r_1$ and $y \notin V_{\mathcal{C}}\cap r_1$ then $\tau(x) < \tau(y)$. 
\end{claim}
\begin{claim}
    \label{claim:nodes-not-in-r1-and-its-neighbor-comes-after-others-in-tau}
    For two nodes $x,y \in \mathcal{C}$, if $x \in V_{\mathcal{C}}\cap (r_1 \cup N(r_1, G_1))$ and $y \notin V_{\mathcal{C}}\cap (r_1\cup N(r_1, G_1))$ then $\tau(x) < \tau(y)$. 
\end{claim}

The following claim comes from the implementation of step \ref{step-2-b-of-construction-of-M} of the construction of $M$.

\begin{claim}
    \label{claim:LBFS-ordering-for-directed-edges}
    For $u-v \in \mathcal{C}$, if $v \notin r_1\cup N(r_1, G_1)$, and $\tau(u) < \tau(v)$ then $v\rightarrow u\notin M$.
\end{claim}
\begin{proof}
Contrary to \cref{claim:LBFS-ordering-for-directed-edges}, suppose $v\rightarrow u \in M$.
From \cref{obs:for-u-v-in-C-if-u->v-in-M-then-either-step-1-or-step-2-obeyed}, either $v\rightarrow u$ has been added to $M$ at step \ref{step-1-of-construction-of-M} of the construction of $M$, or at step \ref{step-2-of-construction-of-M} of the construction of $M$ when $\mathcal{C}$ has been considered.

If $v\rightarrow u$ is added to $M$ at step \ref{step-1-of-construction-of-M} of the construction of $M$ then $v\rightarrow u \in U_M(M_1, M_2, O)$. From \cref{obs:undirected-in-M_1-and-directed-in-M-implies-directed-in-O}, $v\rightarrow u \in O$. This implies $v,u \in V_{M_1} \cap V_O = r_1\cup N(r_1, G_1)$. But, this is a contradiction, as we have given $v \notin r_1\cup N(r_1, G_1)$. This implies $v\rightarrow u$ is not added to $M$ at step \ref{step-1-of-construction-of-M} of the construction of $M$.

Since $\tau(u) < \tau(v)$, the precondition of step \ref{step-2-b-of-construction-of-M} implies $v\rightarrow u$ cannot be added to $M$ at step \ref{step-2-of-construction-of-M}. 

The above discussion implies $v\rightarrow u \notin M$. This completes the proof of \cref{claim:LBFS-ordering-for-directed-edges}.
\end{proof}
W.l.o.g., let us assume that $w$ has the highest rank in $\tau$ among the nodes in the directed cycle $(u,v,w,u)$, i.e., $\tau(u)< \tau(w)$ and $\tau(v) < \tau(w)$.

We first show that all the edges of the cycle $(u,v,w,u)$ are not directed in $M$. All the edges of the cycle $(u,v,w,u)$ are directed in $M$ implies $u\rightarrow v, v\rightarrow w, w\rightarrow u\in M$.
Since $\tau(u) < \tau(w)$, from \cref{claim:LBFS-ordering-for-directed-edges}, $w\rightarrow u \notin M$. This implies all three edges of the cycle are not directed in $M$. 

We now show that no two edges of the cycle $(u,v,w,u)$ are directed.
Since $w\rightarrow v \notin M$ (discussed in the previous paragraph), if two edges are directed then it must be $u\rightarrow v$ and $v\rightarrow w$. 
From \cref{obs:for-u-v-in-C-if-u->v-in-M-then-either-step-1-or-step-2-obeyed}, the two directed edges must have been added to $M$ either at step \ref{step-1-of-construction-of-M} or at step \ref{step-2-of-construction-of-M} when $\mathcal{C}$ has been considered.
Since $w$ is of highest rank in $\tau$ among the nodes in the directed cycle, from
\Cref{claim:one-node-of-the-cycle-is-not-in-r1-and-its-neighbors,claim:nodes-not-in-r1-and-its-neighbor-comes-after-others-in-tau}, $w \in V_{\mathcal{C}}\setminus{r_1\cup N(r_1, G_1)}$.
This implies that while running step \ref{step-2-b-of-construction-of-M}, $u\rightarrow w$ gets added to $M$ due to the preconditions of step \ref{step-2-b-b-of-construction-of-M} get satisfied. This implies $(u,v,w,u)$ is not a cycle in $M$, a contradiction. Therefore, no two edges of the cycle $(u,v,w,u)$ are directed. 

We now show that not even a single edge of cycle $(u,v,w,u)$ of $M$ is directed. Since $w\rightarrow v \notin M$, if an edge of the cycle is directed then there are two possibilities: either $u\rightarrow v \in M$ (that means $u\rightarrow v-w-u$ is a directed cycle in $M$) or $v\rightarrow w \in M$ (that means $u- v\rightarrow w-u$ is a directed cycle in $M$). The following claims show that both possibilities do not occur.

\begin{claim}
    \label{claim:possibility-1}
    There do not exist nodes $u,v,w \in \mathcal{C}$ such that $u\rightarrow v -w-u\in M$, $w \notin r_1\cup N(r_1, G_1)$, and rank of $w$ in $\tau$ is highest among the nodes in the cycle, where $\tau$ is the LBFS ordering of $\mathcal{C}$ that we get at step \ref{step-2-a-of-construction-of-M} of the construction of $M$.
\end{claim}

\begin{claim}
    \label{claim:possibility-2}
    There do not exist nodes $u,v,w \in \mathcal{C}$ such that  $u-v\rightarrow w-u \in M$, and rank of $w$ in $\tau$ is highest among the nodes in the cycle, where $\tau$ is the LBFS ordering of $\mathcal{C}$ that we get at step \ref{step-2-a-of-construction-of-M} of the construction of $M$. 
\end{claim}

\Cref{claim:possibility-1,claim:possibility-2} imply that none of the edges of the cycle $(u,v,w,u)$ in $M$ is directed. This achieves our goal of ensuring that $(u, v, u, v)$ does not form a directed cycle in $M$. This completes the proof of \cref{claim:no-directed-cycle-with-nodes-in-ucc-with-ri}.
\end{proof}

\begin{proof}[Proof of \cref{claim:possibility-1}]
    Suppose there exist nodes $u,v,w \in \mathcal{C}$ such that $u\rightarrow v -w-u\in M$, $w \notin r_1\cup N(r_1, G_1)$, and rank of $w$ in $\tau$ is highest among the nodes in the cycle, where $\tau$ is the LBFS ordering of $\mathcal{C}$ that we get at step \ref{step-2-a-of-construction-of-M} of the construction of $M$. Let $Y = \{y: y\in V_{\mathcal{C}}, \tau(y) < \tau(w)$, and $u\rightarrow y -w-u \in M \}$. We pick a $v \in Y$ such that for all $y \in Y$, $y\rightarrow v \notin M$.
    We first show that such $v$ always exists. 

    Suppose there does not exist a $v \in Y$ such that for all $y\in Y$, $y\rightarrow v \notin M$. Then, there exists a cycle $C_1 = (v_1, v_2, \ldots, v_m, v_{m+1} = v_1)$ in $M$ such that for $1\leq i \leq m$, $v_i \in Y$, and $v_i\rightarrow v_{i+1} \in M$. Since nodes in $Y$ are in $V_{\mathcal{C}}$, we pick the node $v_j$ among the nodes in the cycle $C_1$ which has the highest rank in $\tau$. 
    This implies $\tau(v_j) > \tau(v_{j+1})$. But, then, from \cref{claim:LBFS-ordering-for-directed-edges}, $v_j \in r_1\cup N(r_1, G_1)$, otherwise, we do not have $v_j \rightarrow v_{j+1} \in M$. Since $v_j$ is the highest rank node among the nodes in $C_1$, from \cref{claim:nodes-not-in-r1-and-its-neighbor-comes-after-others-in-tau}, each node of the cycle is in $r_1\cup N(r_1, G_1)$. This implies for $1\leq i \leq m$, $u_i-u_{i+1} \in \mathcal{C}$ does not get directed at step \ref{step-2-of-construction-of-M} of the construction of $M$. Then, from \cref{obs:for-u-v-in-C-if-u->v-in-M-then-either-step-1-or-step-2-obeyed,obs:undirected-in-M_1-and-directed-in-M-implies-directed-in-O}, for $1\leq i \leq m$, $u_i\rightarrow u_{i+1} \in O$. This implies $O$ has a directed cycle, a contradiction, as $O$ is a partial MEC (\cref{def:partial-MEC}). This implies our assumption is not correct, and $v$ always exists. 
    
    The following claim provides us the potential induced subgraph of $M$ that contains  $u\rightarrow v$.
    \begin{claim}
    \label{claim:if-u->v-in-M-then-it-is-sp-using-a-or-c}
For some $x \in V_{\mathcal{C}}$, $u\rightarrow v$ is part of an induced subgraph $x\rightarrow u\rightarrow v$ of $M$, or $u\rightarrow v$ is part of an induced subgraph $u\rightarrow x \rightarrow v \leftarrow u$ of $M$.
\end{claim}
\Cref{claim:if-u->v-in-M-then-it-is-sp-using-a-or-c} provides two possible induced subgraphs of $M$ that can contain $u\rightarrow v$. We go through both of them.
\begin{enumerate}
    \item Suppose there exists $x\in V_{\mathcal{C}}$ such that $u\rightarrow v$ is part of an induced subgraph $x\rightarrow u\rightarrow v$ of $M$. We first claim that $\tau(x) < \tau(w)$.

    Suppose $\tau(w) < \tau(x)$. Since $w \in V_{\mathcal{C}}\setminus{r_1\cap N(r_1, G_1)}$, from \cref{claim:nodes-not-in-r1-and-its-neighbor-comes-after-others-in-tau}, $x\in V_{\mathcal{C}}\setminus{r_1\cap N(r_1, G_1)}$. Also, since $\tau(u) < 
    \tau(w)$, we have $\tau(u) < \tau(x)$. But, then, from \cref{claim:LBFS-ordering-for-directed-edges}, $x\rightarrow u \notin M$, a contradiction, as we have assumed $x\rightarrow u\rightarrow v \in M$. This implies $\tau(x) < \tau(w)$.

    We now claim that $x-w\notin \mathcal{C}$. Suppose $x-w \in \mathcal{C}$. As $\tau(x) < \tau(w)$, $\tau(v) < \tau(w)$, and $x-w, v-w \in \mathcal{C}$, from \cref{def:LBFS}, $x-v \in \mathcal{C}$. From the construction of $M$, $\skel{M[V_{\mathcal{C}}]} = \mathcal{C}$. This implies there is an edge between $x$ and $v$ in $M$. This further implies $x\rightarrow u \rightarrow v$ is not an induced subgraph of $M$, contradicting our assumption. Therefore, $x-w \notin 
    \mathcal{C}$.

    As discussed above, $\tau(x) < \tau(w)$, $\tau(u) < \tau(w)$, and $x\rightarrow u \in M$. This satisfies the precondition of step \ref{step-2-b-a-of-construction-of-M} of the construction of $M$ to make $u\rightarrow w \in M$. But, from our assumption $u-w \in M$, a contradiction. This implies this possibility cannot occur.

    \item Suppose there exists $x\in V_{\mathcal{C}}$ such that $u\rightarrow v$ is part of an induced subgraph $u\rightarrow x \rightarrow v \leftarrow u$ of $M$. We first show that $x \in Y$ (Recall $Y = \{y: y\in V_{\mathcal{C}}, \tau(y) < \tau(w)$, and $u\rightarrow y -w-u \in M \}$). To prove that $x\in Y$, we have to show that $\tau(x) < \tau(w)$, and $x-w \in M$.

    Suppose $\tau(w) < \tau(x)$. Since $\tau(v) < \tau(w)$, we have $\tau(v) < \tau(x)$. Also, since $w\notin V_{\mathcal{C}}\cap (r_1\cup N(r_1, G_1))$, and $\tau(w) < \tau(x)$, from \cref{claim:nodes-not-in-r1-and-its-neighbor-comes-after-others-in-tau}, $x \notin V_{\mathcal{C}}\cap (r_1\cup N(r_1, G_1))$. But, then, from \cref{claim:LBFS-ordering-for-directed-edges}, $x\rightarrow v \notin M$, contradicting our assumption that $u\rightarrow x \rightarrow v \leftarrow u \in M$. This implies $\tau(x) < \tau(w)$.

    Suppose $x-w \notin \skel{M}$. Then, $x\rightarrow v-w$ is an induced subgraph of $M$ such that $\tau(x) < \tau(w)$, and $\tau(v) < \tau(w)$. Then, from step \ref{step-2-b-a-of-construction-of-M} of the construction of $M$, we have $v\rightarrow w \in M$, a contradiction, as we have assumed $u\rightarrow v-w-u$ is a directed cycle in $M$. This implies $x-w \in \skel{M}$. This implies either $x\rightarrow w \in M$ or $w\rightarrow x \in M$ or $x-w \in M$. We show that the first two possibility does not occur. If $x\rightarrow w \in M$ then we have a directed cycle $(u, x, w, u)$ with two directed edges. But, we have shown earlier that there cannot exist a directed cycle of length two with two directed edges. This implies $x\rightarrow w \notin M$. Since $\tau(x) < \tau(w)$, and $w \notin r_1\cup N(r_1, G_1)$, from \cref{claim:LBFS-ordering-for-directed-edges}, $w\rightarrow x \notin M$. Thus, the only option that remains is $x-w \in M$.

    The above discussion implies that $x\in Y$. But, from the choice of $v$, there cannot exist any $x \in Y$ such that $x\rightarrow v \in M$, a contradiction. This implies our assumption is wrong, and there does not any $x\in V_{\mathcal{C}}$ such that $u\rightarrow v$ is part of an induced subgraph $u\rightarrow x\rightarrow v\leftarrow u$ in $M$.
\end{enumerate}

The above discussion implies that $u\rightarrow v$ does not obey \cref{claim:if-u->v-in-M-then-it-is-sp-using-a-or-c}, a contradiction.
This implies our assumption is wrong, and there do not exist nodes $u,v,w \in \mathcal{C}$ such that $u\rightarrow v -w-u\in M$. This completes the proof of \cref{claim:possibility-1}.
\end{proof}

\begin{proof}[Proof of \cref{claim:possibility-2}]
    Suppose there exist $u,v,w \in \mathcal{C}$ such that  $u-v\rightarrow w-u \in M$, and rank of $w$ in $\tau$ is highest among the nodes in the cycle, where $\tau$ is the LBFS ordering of $\mathcal{C}$ that we get at step \ref{step-2-a-of-construction-of-M} of the construction of $M$. We pick $u,v, w$ with least $\tau(w)$, i.e., for any other triplet $u',v',w' \in \mathcal{C}$ such that  $u'-v'\rightarrow w'-u' \in M$, and rank of $w'$ in $\tau$ is highest among the nodes in the cycle $(u',v',w',u')$ then $\tau(w) < \tau(w')$.

The following claim provides the potential induced subgraphs of $M$ that contain  $v\rightarrow w$.  
\begin{claim}
    \label{claim:if-v->w-in-M-then-it-is-sp-using-a-or-c}
    For some $x\in V_{\mathcal{C}}$ such that $\tau(x) < \tau(w)$, $v\rightarrow w$ is part of an induced subgraph $x\rightarrow v\rightarrow w$ or $v\rightarrow x \rightarrow w \leftarrow v$ in $M$.
\end{claim}

\Cref{claim:if-v->w-in-M-then-it-is-sp-using-a-or-c} provides two possible induced subgraphs of $M$ that contains $v\rightarrow w$. We go through both of them.
\begin{enumerate}
    \item Suppose there exists $x\in V_{\mathcal{C}}$ such that $\tau(x) < \tau(w)$, and $v\rightarrow w$ is part of an induced subgraph $x\rightarrow v\rightarrow w$ of $M$. We claim that $x-u \in \mathcal{C}$.

    Suppose $x-u \notin \mathcal{C}$. From \cref{obs:for-u-v-in-C-if-u->v-in-M-then-either-step-1-or-step-2-obeyed}, $x\rightarrow v$ has been added to $M$ either at step \ref{step-1-of-construction-of-M} of the construction of $M$, or at step \ref{step-2-of-construction-of-M} of the construction of $M$ when $\mathcal{C}$ has been considered.
    
    Suppose $x\rightarrow v$ has been added to $M$ at step \ref{step-1-of-construction-of-M} of the construction of $M$.
    This implies $x\rightarrow v \in U_M(M_1, M_2, O)$. From \cref{obs:undirected-in-M_1-and-directed-in-M-implies-directed-in-O}, $x\rightarrow v \in O$. Now, if $u \in V_O$ then $x\rightarrow v-u$ is an induced subgraph in $O$, contradicting $O$ as a partial MEC (\cref{def:partial-MEC}). And, if $u\notin V_O$ then at step \ref{step-2-b-a-of-construction-of-M} of the construction of $M$ when $\mathcal{C}$ has been considered, we get $v\rightarrow u \in M$, contradicting $(u,v,w,u)$ is a cycle. Thus, in each case of this possibility, we get a contradiction. We move to the second possibility.

    Suppose $x\rightarrow v$ has been added to $M$ at step \ref{step-2-of-construction-of-M} of the construction of $M$ when $\mathcal{C}$ has been considered. Then, from the preconditions of steps \ref{step-2-b-of-construction-of-M} of the construction of $M$, we have $\tau(x) < \tau(v)$. Then, $\tau(v) < \tau(u)$, otherwise, if $\tau(u) < \tau(v)$ then from \cref{def:LBFS}, $x-u \in \mathcal{C}$, contradicting our assumption. But, if $\tau(v)< \tau(u)$, then at step \ref{step-2-b-a-of-construction-of-M} of the construction of $M$ when $\mathcal{C}$ has been considered, we get $v\rightarrow u \in M$, contradicting $(u,v,w,u)$ is a cycle. 
    
    Thus, in each case of both possibilities, we get a contradiction. This implies $x-u \in \mathcal{C}$.

    From the construction of $M$, $\skel{M[V_{\mathcal{C}}]} = \mathcal{C}$. Therefore, we have three possibilities: either $u\rightarrow x \in M$ or $x\rightarrow u \in M$ or $x-u \in M$. We show that none of the possibilities occurs. 
    
    If $u\rightarrow x \in M$ then we get a directed cycle $(u,x,v,u)$ in $M$ with two directed edges. We have shown earlier that $M$ does not have a directed cycle of length three with two directed edges. This implies $u\rightarrow x\notin M$. 

    Suppose $x\rightarrow u \in M$. Since $x\rightarrow v\rightarrow w$ is an induced subgraph in $M$, $x-w \notin \skel{M}$. This implies $x\rightarrow u-w$ is an induced subgraph in $M$ with $w \notin r_1\cup N(r_1, G)$, $\tau(x)< \tau(w)$, and $\tau(u) < \tau(w)$. This satisfies the precondition of step \ref{step-2-b-a-of-construction-of-M} of the construction of $M$ for adding $u\rightarrow w$. And, therefore, while running step \ref{step-2-of-construction-of-M} of the construction of $M$ for the \ucc{} $\mathcal{C}$, we add $u\rightarrow w$ at step \ref{step-2-b-a-of-construction-of-M}. But, this creates a contradiction, as from our assumption $u-w \in M$ (as $u-v\rightarrow w-u$ is a cycle in $M$).

    If $x-u \in M$ then we get a directed cycle $(x, v, u, x)$ in $M$ such that the rank of all the three nodes of the cycle in $\tau$ is less than $\tau(w)$. From \cref{claim:one-node-of-the-cycle-is-not-in-r1-and-its-neighbors}, one node of the cycle does not belong to $r_1\cup N(r_1, G_1)$. From \cref{claim:nodes-not-in-r1-and-its-neighbor-comes-after-others-in-tau}, among the nodes of the cycle, the node with the highest rank does not belong to $r_1\cup N(r_1, G)$. From \cref{claim:LBFS-ordering-for-directed-edges}, $x$ cannot be the highest rank node. This implies either $u$ or $v$ is the highest rank node. If $u$ is the highest rank node then $M$ contains a cycle $x\rightarrow v-u-x$ such that the rank of all the three nodes of the cycle in $\tau$ is less than $\tau(u)$. \Cref{claim:possibility-1} implies that such a cycle does not exist. 
    And, if $v$ is the highest rank node then we get a directed cycle $u-x\rightarrow v-x$ in $M$ with the highest rank node is $\tau(v)$. Since $\tau(v) < \tau(w)$, this contradicts our choice of $w$. This implies this case also cannot occur. 

    The above discussion implies that in each case of this possibility, we get a contradiction. This implies our assumption is wrong.  We now move to the second possibility of \cref{claim:if-v->w-in-M-then-it-is-sp-using-a-or-c}.

    \item Suppose there exists $x\in V_{\mathcal{C}}$ such that $\tau(x) < \tau(w)$, and $v\rightarrow w$ is part of an induced subgraph $v\rightarrow x \rightarrow w \leftarrow v$ of $M$. 

    Given $\tau(x) < \tau(w)$, $\tau(u) < \tau(w)$, and $x-w, u-w \in \mathcal{C}$, from \cref{def:LBFS}, $x-u \in \mathcal{C}$.
    From the construction of $M$, $\skel{M[V_{\mathcal{C}}]} = \mathcal{C}$. Therefore, we have three possibilities: either $u\rightarrow x \in M$ or $x\rightarrow u \in M$ or $u-x \in M$. We show that none of the possibilities occurs.

    Suppose $u\rightarrow x \in M$. Then, while running step \ref{step-2-of-construction-of-M} when $\mathcal{C}$ has been considered, we add $u\rightarrow w$ in $M$ as all the preconditions of step \cref{step-2-b-b-of-construction-of-M} is obeyed. But, this contradicts our assumption that $(u,v,w,u)$ is a cycle in $M$. This implies $u\rightarrow x \notin M$.

    Suppose $x\rightarrow u \in M$. Then, $M$ has a directed cycle $(v, x, u, v)$ with two directed edges. As discussed earlier, there cannot exist a directed cycle of length three with two directed edges. This implies $x\rightarrow u \notin M$.

    Suppose $u-x \in M$. Then, we get a directed cycle $(u, v, x, u)$ in $M$ such that the rank of all the three nodes of the cycle in $\tau$ is less than $\tau(w)$.
     From \cref{claim:one-node-of-the-cycle-is-not-in-r1-and-its-neighbors}, one node of the cycle does not belong to $r_1\cup N(r_1, G_1)$. 
     From \cref{claim:nodes-not-in-r1-and-its-neighbor-comes-after-others-in-tau}, among the nodes of the cycle, the node with the highest rank does not belong to $r_1\cup N(r_1, G)$.
     From \cref{claim:LBFS-ordering-for-directed-edges}, $u$ cannot be the highest rank node.
     This implies either $u$ or $x$ is the highest rank node.
     If $u$ is the highest rank node then $M$ contains a cycle $v\rightarrow x-u-v$ such that $u$ has the highest rank in $\tau$ among nodes in the cycle. \Cref{claim:possibility-1} implies that such a cycle does not exist. And, if $x$ is the highest rank node then we get a directed cycle $u-v\rightarrow x-u$ in $M$ with the highest rank node being $x$. Since $\tau(x) < \tau(w)$, this contradicts our choice of $w$. This implies this cannot occur.
\end{enumerate}

The above discussion implies $v\rightarrow w$ is part of none of the induced subgraphs mentioned in \cref{claim:if-v->w-in-M-then-it-is-sp-using-a-or-c}. This contradicts \cref{claim:if-v->w-in-M-then-it-is-sp-using-a-or-c}. This implies that our assumption is wrong. This validates \cref{claim:possibility-2}.
\end{proof}

\begin{proof}[Proof of \cref{claim:if-u->v-in-M-then-it-is-sp-using-a-or-c}]
From \cref{obs:for-u-v-in-C-if-u->v-in-M-then-either-step-1-or-step-2-obeyed}, $u\rightarrow v$ has been added to $M$ either at 
step \ref{step-1-of-construction-of-M} of the construction of $M$, or it has been added to $M$ at step \ref{step-2-of-construction-of-M} of the construction of $M$, when $\mathcal{C}$ has been considered.

    If $u\rightarrow v$ is added to $M$ at step \ref{step-2-of-construction-of-M} of the construction of $M$ then $u\rightarrow v$ obeys \cref{claim:if-u->v-in-M-then-it-is-sp-using-a-or-c} (see step \ref{step-2-of-construction-of-M} of the construction of $M$).

    Suppose $u\rightarrow v$ is added to $M$ at step \ref{step-1-of-construction-of-M} of the construction of $M$. This implies $u\rightarrow v \in U_M(M_1, M_2, O)$. From \cref{obs:undirected-in-M_1-and-directed-in-M-implies-directed-in-O}, $u\rightarrow v \in O$. This implies $u,v \in V_{M_1}\cap V_O = V_{O_1} =r_1\cup N(r_1, G_1)$. From \cref{claim:no-node-of-cycle-is-in-r_1}, $u,v \in N(r_1, G_1)\setminus r_1$. Since $u,v \in N(r_1, G_1)\setminus r_1$, and $I = r_1\cap r_2$ is a vertex separator of $G$ that separates $V_{G_1}\setminus I$ and $V_{G_2}\setminus I$, neighbors of $u$ and $v$ belong to $V_{G_1}$.
    Since $u,v \in V_{O_1}$, and $O_1$ is an induced subgraph of $M_1$, $u-v \in O_1$. From \cref{item-3-of-def:extension} of \cref{def:extension}, $u\rightarrow v$ is strongly protected in $O$.  This implies $u\rightarrow v$ is part of one of the induced subgraphs of $O$ as shown in \cref{fig:strongly-protected-edge}.
\begin{enumerate}
    \item Suppose $u\rightarrow v$ is strongly protected in $O$ because it is part of an induced subgraph $x\rightarrow u\rightarrow v$ of $O$ (as shown in \cref{fig:strongly-protected-edge}.a).  From \cref{obs:O-is-an-induced-subgraph-of-M}, $x\rightarrow u \in M$.
    As discussed earlier, the neighbors of $u$ and $v$ are in $V_{G_1}$. This implies $x \in V_{G_1}$.
    From \cref{corr:directed-edge-in-M-is-either-directed-or-undirected-in-M1}, either $x-u \in M_1$ or $x\rightarrow u \in M$.
     From \cref{item-3-theorem-nec-suf-cond-for-MEC} of \cref{thm:nes-and-suf-cond-for-chordal-graph-to-be-an-MEC}, $x\rightarrow u-v$ cannot be an induced subgraph of $M$. Therefore, $x-u \in M_1$. This implies $x \in \mathcal{C}$. This implies $u\rightarrow v$ obeys \cref{claim:if-u->v-in-M-then-it-is-sp-using-a-or-c}.

     \item Suppose $u\rightarrow v$ is strongly protected in $O$ because it is part of an induced subgraph $x\rightarrow v\leftarrow u$ of $O$ (as shown in \cref{fig:strongly-protected-edge}.b). From \cref{obs:O-is-an-induced-subgraph-of-M}, $x\rightarrow v \in M$.
     As discussed earlier, the neighbors of $u$ and $v$ are in $V_{G_1}$. This implies $x \in V_{G_1}$. 
     From \cref{corr:directed-edge-in-M-is-either-directed-or-undirected-in-M1}, either $x-v \in M_1$ or $x\rightarrow v \in M$.
     From \cref{item-3-theorem-nec-suf-cond-for-MEC} of \cref{thm:nes-and-suf-cond-for-chordal-graph-to-be-an-MEC}, $x\rightarrow u-v$ cannot be an induced subgraph of $M$. Therefore, $x-u \in M_1$. Since $u,v, x \in V_{O} \cap V_{M_1} = V_{O_1}$, and $\skel{O_1} = \skel{O[V_{O_1}]}$, $x-v-u$ is an induced subgraph of $O_1$. This implies $\mathcal{V}(O_1)\neq \mathcal{V}(O[V_{O_1}])$ (as $x\rightarrow v \leftarrow u \in O$ and $x-v-u \in O_1$). This contradicts \cref{item-2-of-def:extension} of \cref{def:extension}. This implies this case cannot occur.

     \item Suppose $u\rightarrow v$ is strongly protected in $O$ because it is part of an induced subgraph $u\rightarrow x \rightarrow v\leftarrow u$ of $O$ (as shown in \cref{fig:strongly-protected-edge}.c). From \cref{obs:O-is-an-induced-subgraph-of-M}, $u\rightarrow x \rightarrow v\leftarrow u\in M$.
     As discussed earlier, the neighbors of $u$ and $v$ are in $V_{G_1}$. This implies $x \in V_{G_1}$. 
     From \cref{corr:directed-edge-in-M-is-either-directed-or-undirected-in-M1}, (a) either $u-x \in M_1$ or $u\rightarrow x \in M$, and (b) either $x-v \in M_1$ or $x\rightarrow v \in M_1$. If either $u\rightarrow x \in M_1$ or $x\rightarrow v \in M_1$ then $(u, x, v, u)$ is a directed cycle in $M_1$, contradicting \cref{item-1-theorem-nec-suf-cond-for-MEC} of \cref{thm:nes-and-suf-cond-for-chordal-graph-to-be-an-MEC}. This implies $u-x, x-v \in M_1$, and $x \in \mathcal{C}$. This further implies $u\rightarrow v$ obeys \cref{claim:if-u->v-in-M-then-it-is-sp-using-a-or-c}.

     \item Suppose $u\rightarrow v$ is strongly protected in $O$ because it is part of an induced subgraph $u-x\rightarrow v\leftarrow u-x'\rightarrow v$ of $O$ (as shown in \cref{fig:strongly-protected-edge}.d). From \cref{obs:O-is-an-induced-subgraph-of-M}, $x\rightarrow v \in M$.
     As discussed earlier, the neighbors of $u$ and $v$ are in $V_{G_1}$. This implies $x \in V_{G_1}$. 
     From \cref{corr:directed-edge-in-M-is-either-directed-or-undirected-in-M1}, either $x-v \in M_1$ or $x\rightarrow v \in M$. From \cref{corr:undirected-edge-in-M-is-undirected-in-M1}, $u-x \in M_1$. If $x\rightarrow v \in M_1$ then $(u,x,v,u)$ is a directed cycle in $M_1$, contradicting \cref{item-1-theorem-nec-suf-cond-for-MEC} of \cref{thm:nes-and-suf-cond-for-chordal-graph-to-be-an-MEC}. This implies $x-v \in M_1$. Similarly, $x'-v \in M_1$. Since $u,v, x,x' \in V_{O} \cap V_{M_1} = V_{O_1}$, and $\skel{O_1} = \skel{O[V_{O_1}]}$, $x-v-x'$ is an induced subgraph of $O_1$. This implies $\mathcal{V}(O_1)\neq \mathcal{V}(O[V_{O_1}])$ (as $x\rightarrow v \leftarrow x' \in O$ and $x-v-x' \in O_1$). This contradicts \cref{item-2-of-def:extension} of \cref{def:extension}. This implies this case cannot occur.
\end{enumerate}
The above discussion implies that $u\rightarrow v$ obeys \cref{claim:if-u->v-in-M-then-it-is-sp-using-a-or-c}. This completes the proof of \cref{claim:if-u->v-in-M-then-it-is-sp-using-a-or-c}.    
\end{proof}

\begin{proof}[Proof of \cref{claim:if-v->w-in-M-then-it-is-sp-using-a-or-c}]
    From \cref{obs:for-u-v-in-C-if-u->v-in-M-then-either-step-1-or-step-2-obeyed}, either $v\rightarrow w$ has been added to $M$ at step \ref{step-1-of-construction-of-M} of the construction of $M$, or it has been added to $M$ at step \ref{step-2-of-construction-of-M} of the construction of $M$ when $\mathcal{C}$ has been considered.

    Suppose $v\rightarrow w$ has been added to $M$ at step \ref{step-1-of-construction-of-M} of the construction of $M$. This implies $v\rightarrow w \in U_M(M_1, M_2, O)$. From \cref{obs:undirected-in-M_1-and-directed-in-M-implies-directed-in-O}, $v\rightarrow w \in O$. This implies $w \in V_O\cap V_{M_1} = r_1\cup N(r_1, G_1)$, a contradiction, as we have assumed $w\notin r_1\cup N(r_1, G_1)$. This implies $v\rightarrow w$ has not been added to $M$ at step \ref{step-1-of-construction-of-M} of the construction of $M$.

    If $v\rightarrow w$ is added to $M$ at step \ref{step-2-of-construction-of-M} of the construction of $M$ then $v\rightarrow w$ obeys \cref{claim:if-v->w-in-M-then-it-is-sp-using-a-or-c} (see step \ref{step-2-of-construction-of-M} of the construction of $M$).

    This implies $v\rightarrow w$ obeys \cref{claim:if-v->w-in-M-then-it-is-sp-using-a-or-c}. This completes the proof of \cref{claim:if-v->w-in-M-then-it-is-sp-using-a-or-c}.
\end{proof}

    \begin{proof}[Proof of \cref{obs:ucc-of-M-is-chordal}]
    Suppose there exists an \ucc{} $\mathcal{C}$ of $M$ such that $\mathcal{C}$ is not chordal. Then, there must exist a chordless cycle $C = (u_0, u_1, \ldots, u_l, u_{l+1} = u_0)$ in $\mathcal{C}$. Since $G$ is a chordal graph, $C$ cannot be a chordless cycle in $G$. This implies there must exist an edge $u_i-u_j$ between two non-adjacent nodes of $C$ in $G$. W.l.o.g., we assume $i<j$. 
    Since the skeleton of $M$ is $G$, either $u_i\rightarrow u_j \in M$, or $u_j\rightarrow u_i \in M$, or $u_i-u_j \in M$. 

 If $u_i\rightarrow u_j \in M$ then $(u_0, u_1, \ldots, u_i, u_j, u_{j+1}, \ldots, u_l, u_{l+1} = u_0)$ is a directed cycle in $M$. If $u_j\rightarrow u_i \in M$ then $(u_j, u_i, u_{i+1}, \ldots, u_j)$ is a directed cycle in $M$. Both of them contradict \cref{obs:M-is-a-chain-graph}, which asserts that $M$ is a chain graph.
 
Thus, the only option that remains is $u_i-u_j \in M$. But, then, $C$ is not a chordless cycle in $\mathcal{C}$, contradicting our assumption that $\mathcal{C}$ is not chordal. This completes the proof of \cref{obs:ucc-of-M-is-chordal}.
\end{proof}

\begin{proof}[Proof of \cref{obs:M-does-not-have-a->b-c}]
    Suppose $M$ has an induced subgraph $a\rightarrow b-c$. 
    Since $I = V_{G_1}\cap V_{G_2}$ is a vertex separator of $G$, either $a,b,c \in V_{G_1}$, or $a,b,c \in V_{G_2}$, or $a\in V_{G_1}\setminus I$, $b\in I$ and $c \in V_{G_2}\setminus I$.
    We one by one show that none of the possibilities occurs.

    Suppose $a\rightarrow b-c$ is an induced subgraph of $M$ such that $a,b,c \in V_{G_1}$.
   Since the skeleton of $M_1$ and the skeleton of $M[V_{G_1}]$ both are the same, from \cref{corr:directed-edge-in-M-is-either-directed-or-undirected-in-M1,corr:undirected-edge-in-M-is-undirected-in-M1}, either $a\rightarrow b-c$ or $a-b-c$ is an induced subgraph of $M_1$. Since $M_1$ is an MEC, from \cref{item-3-theorem-nec-suf-cond-for-MEC} of \cref{thm:nes-and-suf-cond-for-chordal-graph-to-be-an-MEC}, $a\rightarrow b-c$ is cannot be an induced subgraph of $M_1$. Thus, the only option that remains is $a-b-c$ is an induced subgraph of $M_1$. This implies $a,b$ and $c$ belong to the same \ucc{} $\mathcal{C}$ of $M_1$. From the construction of $M$, $\skel{M[V_{\mathcal{C}}]}= \mathcal{C}$. This implies $a-b, b-c \in \mathcal{C}$ (there cannot exist a directed edge between two nodes of the same \ucc{} of an MEC, otherwise, we get a directed cycle, contradicting \cref{item-1-theorem-nec-suf-cond-for-MEC} of \cref{thm:nes-and-suf-cond-for-chordal-graph-to-be-an-MEC}). Since $a-b \in \mathcal{C}$, and $a\rightarrow b \in M$, from \cref{claim:ucc-of-Mi-not-in-r1-is-a-ucc-in-M}, $V_{\mathcal{C}} \cap r_1 \neq \emptyset$. From \cref{obs:for-u-v-in-C-if-u->v-in-M-then-either-step-1-or-step-2-obeyed}, $a\rightarrow b$ is added to $M$ either at step \ref{step-1-of-construction-of-M} of the construction of $M$, or at step \ref{step-2-of-construction-of-M} of the construction of $M$ when $\mathcal{C}$ has been considered.

    Suppose $a\rightarrow b$ is added to $M$ at step \ref{step-1-of-construction-of-M} of the construction of $M$. This implies $a\rightarrow b \in U_M(M_1, M_2, O)$. Then, from \cref{obs:undirected-in-M_1-and-directed-in-M-implies-directed-in-O}, $a\rightarrow b \in O$. This implies $a,b \in r_1\cup N(r_1, G_1)$. This further implies $c \notin r_1\cup N(r_1, G_1)$. Otherwise, $a\rightarrow b -c$ is an induced subgraph of $O$, contradicting \cref{def:partial-MEC} (recall that $O$ is a partial MEC). This implies $a,b \in V_{\mathcal{C}}\cap V_{O} = V_{\mathcal{C}}\cap (r_1\cup N(r_1, G_1))$, and $c \in V_{\mathcal{C}}\setminus{r_1\cup N(r_1, G_1)}$. From \cref{claim:nodes-not-in-r1-and-its-neighbor-comes-after-others-in-tau}, $\tau(a) < \tau(c)$, and $\tau(b) < \tau(c)$. Then, at step \ref{step-2-of-construction-of-M} of the construction of $M$, when $\mathcal{C}$ has been considered,  while running step \ref{step-2-b-a-of-construction-of-M}, we get $b\rightarrow c \in M$, a contradiction, since $b-c \in M$. 
    
    Suppose $a\rightarrow b$ is added to $M$ at step \ref{step-2-of-construction-of-M} of the construction of $M$.
    Let $\tau$ be the LBFS ordering of $\mathcal{C}$ returned at step \ref{step-2-a-of-construction-of-M} of the construction of $M$. Then, $\tau(a) < \tau(b)$. There cannot be $\tau(c) < \tau(b)$, otherwise, from \cref{def:LBFS}, we have $a-c \in \mathcal{C}$, which implies $a-b-c$ is not an induced subgraph of $M$, a contradiction. This further implies $\tau(b) < \tau(c)$. Then, while running step \ref{step-2-b-a-of-construction-of-M}, we get $b\rightarrow c \in M$, a contradiction, since $b-c \in M$. In every case, we get a contradiction. This implies that the possibility that $a,b,c, \in V_{G_1}$ cannot occur. 

    Similarly, we can claim that $a,b,c \notin V_{G_2}$. And, if $a=V_{G_1}\setminus{I}$, $b \in I$, and $c\in V_{G_2}\setminus I$ then $a,b,c \in V_O$. From \cref{obs:O-is-an-induced-subgraph-of-M}, $O$ is an induced subgraph of $M$. This implies $a\rightarrow b-c \in O$. But, this is a contradiction, as $O$ is a partial MEC (\cref{def:partial-MEC}). This implies this case also cannot occur. This further implies that there cannot exist an induced subgraph $a\rightarrow b-c$ in $M$. This completes the proof of \cref{obs:M-does-not-have-a->b-c}.
\end{proof}

\begin{proof}[Proof of \cref{obs:directed-edges-of-M-are-sp}]
    Suppose $u\rightarrow v$ is a directed edge in $M$. From the construction of $M$, there are two possibilities: either $u,v \in V_{G_1}$, or $u,v \in V_{G_2}$. W.l.o.g., let us assume $u,v \in V_{G_1}$.

    Suppose $u,v \in V_{G_1}$ and $u\rightarrow v \in M$. From \cref{corr:directed-edge-in-M-is-either-directed-or-undirected-in-M1}, either $u\rightarrow v \in M_1$ or $u-v \in M_1$. We show that in both cases $u\rightarrow v$ is strongly protected in $M$.

    \begin{enumerate}
        \item Suppose $u\rightarrow v \in M_1$. According to \cref{item-4-theorem-nec-suf-cond-for-MEC} of \cref{thm:nes-and-suf-cond-for-chordal-graph-to-be-an-MEC}, $u\rightarrow v$ is strongly protected in $M_1$. From \cref{obs:directed-edge-in-M1-is-directed-in-M}, a directed edge in $M_1$ is a directed edge in $M$. Now, let's consider the induced subgraphs shown in (a), (b), (c), and (d) of \cref{fig:strongly-protected-edge}.

Suppose $u\rightarrow v$ is strongly protected in $M_1$ due to its inclusion in an induced subgraph shown in (a), (b), or (c) of \cref{fig:strongly-protected-edge}. From the construction of $M$,  $\skel{M_1} = \skel{M[V_{M_1}]}$. And, from \cref{obs:directed-edge-in-M1-is-directed-in-M}, a directed edge in $M_1$ is a directed edge in $M$. This implies that the induced subgraph is also an induced subgraph of $M$. This further implies $u\rightarrow v$ is also strongly protected in $M$.

Now, let's focus on the case where $u\rightarrow v$ is strongly protected in $M_1$ because it is part of the induced subgraph shown in \cref{fig:strongly-protected-edge}-(d). In this induced subgraph, the edges $w\rightarrow v$, $w'\rightarrow v$, and $u\rightarrow v$ are directed, indicating that $w\rightarrow v$, $w'\rightarrow v$, and $u\rightarrow v$ are also present in $M$ (from \cref{obs:directed-edge-in-M1-is-directed-in-M}).
To establish the strong protection of $u\rightarrow v$ in $M$, we consider different cases. If $w-u$ and $w'-u$ are present in $M$, then $u\rightarrow v$ is part of the induced subgraph shown in \cref{fig:strongly-protected-edge}-(d), making it strongly protected in $M$. If $u\rightarrow w$ or $u\rightarrow w'$ exists in $M$, then again $u\rightarrow v$ is part of an induced subgraph, either as shown in \cref{fig:strongly-protected-edge}-(c) with $u\rightarrow w\rightarrow v\leftarrow u$, or as shown in \cref{fig:strongly-protected-edge}-(c) with $u\rightarrow w'\rightarrow v\leftarrow u$. In both cases, $u\rightarrow v$ is strongly protected in $M$.

The remaining possibility is that both $w\rightarrow u$ and $w'\rightarrow u$ are present in $M$. We show that this case cannot occur.
Since $w-u-w'$ is present in $M_1$, it implies that $w$, $u$, and $w'$ belong to the same undirected connected component $\mathcal{C}$ of $M_1$. 

Suppose $w,u,w' \in r_1\cup N(r_1, G_1) = V_{O_1}$. Since $O_1 = M_1[r_1\cup N(r_1, G_1)]$, $w-u-w'$ is an induced subgraph in $O_1$. Also, $r_1\cup N(r_1, G_1) \subseteq V_O$, and from \cref{obs:O-is-an-induced-subgraph-of-M}, $O$ is an induced subgraph of $M$. This implies $w \rightarrow u\leftarrow w'$ is an induced subgraph of $O$.  But, this is a contradiction, as since $O \in \mathcal{E}(O_1, O_2)$, from \cref{item-2-of-def:extension},   $\mathcal{V}(O_1)= \mathcal{V}(O[V_{O_1}])$. This implies $w,u, w'$ all are not in $r_1\cup N(r_1, G_1)$.

Suppose all the nodes $w,u,w'$ are not in $r_1\cup N(r_1, G_1)$. If for an edge $x-y$ of $M_1$, $x\rightarrow y$ is added to $M$ at step \ref{step-1-of-construction-of-M} of the construction of $M$ then $x\rightarrow y \in U_M(M_1, M_2, O)$. From \cref{obs:undirected-in-M_1-and-directed-in-M-implies-directed-in-O}, $x\rightarrow y \in O$. This implies $x,y \in V_O\cap V_{M_1} = r_1\cup N(r_1, G_1)$. As per our assumption, not all the nodes $w,u,w'$ are in $r_1\cup N(r_1, G_1)$. This implies at least one edge among $w\rightarrow u$ and $w'\rightarrow u$ added to $M$ at step \ref{step-2-of-construction-of-M} of the construction of $M$.
W.l.o.g., let us assume $w\rightarrow u$ is added to $M$ at step \ref{step-2-of-construction-of-M} of the construction of $M$. Then, from the precondition of \cref{step-2-b-of-construction-of-M}, $\tau(w)<\tau(u)$, and $u \in V_{\mathcal{C}}\setminus{r_1\cup N(r_1, G_1)}$. And, since $w'\rightarrow u \in M$, from \cref{claim:LBFS-ordering-for-directed-edges}, $\tau(w') < \tau(u)$. But, then, from \cref{def:LBFS}, $w-w' \in \mathcal{C}$, a contradiction, as $w-u-w'$ is an induced subgraph in $M_1$. This implies this case cannot occur.

In conclusion, we have shown that in all possible cases, if $u\rightarrow v$ is strongly protected in $M_1$, then it is also strongly protected in $M$. Hence, the claim is verified.

\item Suppose $u-v \in M_1$. This implies $u$ and $v$ belong to the same \ucc{} $\mathcal{C}$ of $M_1$.  
$u\rightarrow v$ has been added to $M$ either at step \ref{step-1-of-construction-of-M} of the construction of $M$, or at step \ref{step-2-of-construction-of-M} of the construction of $M$.

If $u\rightarrow v$ is added to $M$ at step \ref{step-1-of-construction-of-M} of the construction of $M$ then $u\rightarrow v \in U_M(M_1, M_2, O)$.  From \cref{obs:undirected-in-M_1-and-directed-in-M-implies-directed-in-O}, $u\rightarrow v \in O$. This implies $u,v \in V_O\cap V_{M_1} = V_{O_1}$. Since $O_1$ is an induced subgraph of $M_1$, this implies $u-v \in O_1$ and $u\rightarrow v \in O$. Then, from \cref{item-3-of-def:extension} of \cref{def:extension}, $u\rightarrow v$ is strongly protected in $O$. From \cref{obs:O-is-an-induced-subgraph-of-M}, $O$ is an induced subgraph of $M$. This implies $u\rightarrow v$ is strongly protected in $M$ as well. 

Suppose $u\rightarrow v$ is added to $M$ at step \ref{step-2-of-construction-of-M} of the construction of $M$. 
Then either it is part of an induced subgraph $w\rightarrow u\rightarrow v$ (from step \ref{step-2-b-a-of-construction-of-M}), as shown in \cref{fig:strongly-protected-edge}.a, or it is part of an induced subgraph $u\rightarrow w\rightarrow v \leftarrow u$ in $M$ (from step \ref{step-2-b-b-of-construction-of-M}), as shown in \cref{fig:strongly-protected-edge}.c. In both cases, $u\rightarrow v$ is strongly protected in $M$.
    \end{enumerate}

    We show that in all the possibilities, if $u\rightarrow v \in M$ then it is strongly protected in $M$. This completes the proof of \cref{obs:directed-edges-of-M-are-sp}.
\end{proof}

\begin{proof}[Proof of \cref{obs:M1-and-M2-are-projections-of-M}]
        We first prove that $\mathcal{P}(M, V_{G_1}) = M_1$. 
        Since $M_1$ is an MEC of $G_1$,  from \cref{def:projection}, to prove $\mathcal{P}(M, V_{G_1}) = M_1$, we only need to show $\mathcal{V}(M_1) = \mathcal{V}(M[V_{G_1}])$. 

        We first show that $\mathcal{V}(M_1) \subseteq \mathcal{V}(M[V_{G_1}])$. From step \ref{step-1-of-construction-of-M} of the construction of $M$, all directed edges of $M_1$ are also directed edges of $M$. This implies $\mathcal{V}(M_1) \subseteq \mathcal{V}(M[V_{G_1}])$, as from the construction of $M$, $\skel{M_1} = \skel{M[V_{G_1}]}$. 

        We now show that $\mathcal{V}(M[V_{G_1}]) \subseteq \mathcal{V}(M_1)$. Suppose $\mathcal{V}(M[V_{G_1}]) \nsubseteq \mathcal{V}(M_1)$. Then, there must exist $u,v,w \in V_{G_1}$ such that $u\rightarrow v \leftarrow w$ is a v-structure in $M[V_{G_1}]$, and $u\rightarrow v\leftarrow w
 \notin M_1$. From \cref{corr:directed-edge-in-M-is-either-directed-or-undirected-in-M1}, either $u\rightarrow v \in M_1$ or $u-v \in M_1$. Similarly, either $v\leftarrow w \in M_1$, or $v-w \in M_1$. Since  $u\rightarrow v\leftarrow w \notin M_1$ (from our assumption), either $u\rightarrow v-w \in M_1$, or $u-v\leftarrow w \in M_1$, or $u-v-w \in M_1$. Since $M_1$ is an MEC, from \cref{item-3-theorem-nec-suf-cond-for-MEC} of \cref{thm:nes-and-suf-cond-for-chordal-graph-to-be-an-MEC}, the first two possibilities are not possible. Thus, the only option that remains is $u-v-w \in M_1$. This implies $u,v,w$ belongs to the same \ucc{} (say $\mathcal{C}$) of $M_1$. 
 
 $u\rightarrow v$ and $w\rightarrow v$ have not been added to $M$ at step \ref{step-1-of-construction-of-M} of the construction of $M$. Otherwise, $u\rightarrow v, w\rightarrow v \in U_M(M_1,M_2, O)$. But, then, from \cref{obs:undirected-in-M_1-and-directed-in-M-implies-directed-in-O}, $u\rightarrow v, w\rightarrow v \in O$. This further implies $u,v,w \in V_{G_1}\cap V_O = V_{O_1}$. Since $O_1$ is an induced subgraph of $M_1$, this implies $u\rightarrow v\leftarrow w$ is a v-structure in $O[V_{O_1}]$ and not in $O_1$. This contradicts \cref{item-2-of-def:extension} of \cref{def:extension}, as $O$ is an extension of $O_1$ and $O_2$. This implies that out of $u\rightarrow v$ and $w\rightarrow v$, at least one has been added to $M$ at step \ref{step-2-of-construction-of-M} of the construction of $M$. 

 W.l.o.g., let us assume $u\rightarrow v$ is added to $M$ at step \ref{step-2-of-construction-of-M} of the construction of $M$. Then, from the precondition of \cref{step-2-b-of-construction-of-M}, $\tau(u)<\tau(v)$, and $v \in V_{\mathcal{C}}\setminus{r_1\cup N(r_1, G_1)}$. And, since $w\rightarrow v \in M$, from \cref{claim:LBFS-ordering-for-directed-edges}, $\tau(w) < \tau(v)$. But, then, from \cref{def:LBFS}, $u-w \in \mathcal{C}$. From the construction of $M$, $\skel{M[V_{\mathcal{C}}]} = \mathcal{C}$. This further implies there cannot be an induced subgraph $u\rightarrow v \leftarrow w$ in $M$, contradicting our assumption. This implies our assumption is wrong. This shows $\mathcal{V}(M[V_{G_1}]) \subseteq \mathcal{V}(M_1)$. This further implies $\mathcal{V}(M[V_{G_1}]) = \mathcal{V}(M_1)$. 
 
 Similarly, we can show that $\mathcal{P}(M, V_{G_2}) = M_2$. This completes the proof of \cref{obs:M1-and-M2-are-projections-of-M}.
    \end{proof}

    \begin{proof}[Proof of \cref{obs:M-is-unique}]
        Since $M \in \setofMECs{G, O}$ and $\mathcal{P}(M, V_{G_1}, V_{G_2}) = (M_1, M_2)$, $\mathcal{V}(M_1) \cup \mathcal{V}(M_2) \cup \mathcal{V}(O) \subseteq \mathcal{V}(M)$.
    Suppose $u\rightarrow v\leftarrow w$ is a v-structure in $M$.
    Since $I = r_1\cap r_2$ is a vertex separator of $G$,  either $u,v,w \in V_{G_1}$ or $u,v,w \in V_{G_2}$, or $u,v, w \in V_O$. Since $\mathcal{P}(M, V_{G_1}, V_{G_2}) = (M_1, M_2)$, if $u,v,w \in V_{G_1}$ then $u\rightarrow v\leftarrow w$ is a v-structure in $M_1$, if $u,v,w \in V_{G_2}$ then $u\rightarrow v\leftarrow w$ is a v-structure in $M_2$, and  if $u,v,w \in V_{O}$ then $u\rightarrow v\leftarrow w$ is a v-structure in $O$. This implies $\mathcal{V}(M) \subseteq \mathcal{V}(M_1) \cup \mathcal{V}(M_2) \cup \mathcal{V}(O)$. This further implies $\mathcal{V}(M) = \mathcal{V}(M_1) \cup \mathcal{V}(M_2) \cup \mathcal{V}(O)$. As discussed in the introduction, given a skeleton and a set of v-structures, there exists a unique MEC with that skeleton and v-structures. This implies that $M$ is unique.
    \end{proof}

    \begin{proof}[Proof of \cref{lem:formula-to-compute-MEC-G-O}]
We use \cref{def:extension} to define a bijective function $f_{O}$. We use the function for the computation of $|\setofMECs{G, O}|$. 
\begin{definition}[Function $f_O$]
\label{def:one-to-one-function-for-counting-MEC-H-O-P1-P2}
Let $G$ be a chordal graph, and $G_1$ and $G_2$ be two induced subgraphs of $G$ such that $G = G_1\cup G_2$, and $I = V_{G_1}\cap V_{G_2}$ is a vertex separator of $G$ that separates $V_{G_1}\setminus I$ and $V_{G_2}\setminus I$. Let $r_1$ and $r_2$ be cliques of, respectively, $G_1$ and $G_2$ such that $r_1\cap r_2 = I$. Let $X_1 = r_1 \cup N(r_1, G_1)$, $X_2 = r_2 \cup N(r_2, G_2)$, and $X = r_1 \cup r_2 \cup N(r_1\cup r_2, G)$.
 Consider $O \in \setofpartialMECs{G[X]}$. Define $T_1 = \setofMECs{G, O}$ and $T_2 = \{(M_1, M_2) :$ $\exists (O_1, O_2) \in O_1 \in \setofpartialMECs{G_1[X_1]} \times \setofpartialMECs{G_2[X_2]}$ such that $O \in \mathcal{E}(O_1,O_2)$, $M_1 \in \setofMECs{G_1, O_1}$, and $M_2 \in \setofMECs{G_2, O_2}\}$. The function $f_{O} : T_1 \rightarrow T_2$ is defined as $f_{O}(M) = (M_1, M_2)$ if $\mathcal{P}(M, V_{G_1}, V_{G_2}) = (M_1, M_2)$.
\end{definition}

 \begin{lemma}
 \label{lem:f-is-one-to-one}
 $f_{O}$ is bijective.
 \end{lemma}
 \begin{proof}
 We first prove $\rightarrow$ of \cref{lem:f-is-one-to-one}.
 \begin{claim}
 \label{claim:to-part-of-lem:f-is-one-to-one}
 For each $M \in \setofMECs{G, O}$, there exists a unique $(M_1, M_2)$ such that $f_{O}(M) = (M_1, M_2)$.
\end{claim}
 \begin{proof}
 Let $M \in \setofMECs{G, O)}$. This implies $M$ is an MEC of $G$, and $M[V_O] = O$. Let $\mathcal{P}(M, V_{G_1}, V_{G_2}) = (M_1, M_2)$.   From \cref{obs:there-exists-a-unique-projection}, there exists a unique $(M_1, M_2)$ such that $\mathcal{P}(M, V_{G_1}, V_{G_2}) = (M_1, M_2)$. Thus, we only have to show that $(M_1, M_2) \in T_2$. 
 Let $O_1 = M_1[X_1]$ and $O_2 = M_2[X_2]$. From \cref{corr:image-is-a-PMEC}, $O_1 \in \setofpartialMECs{G_1[X_1]}$, and $O_2 \in \setofpartialMECs{G_2[X_2]}$.
 This implies $M_1 \in \setofMECs{G_1, O_1}$ and $M_2 \in \setofMECs{G_2, O_2}$.
 From \cref{lem:nes-cond-for-chordal-graph}, $O \in \mathcal{E}(O_1, O_2)$. From the definition of $T_2$ (given at \cref{def:one-to-one-function-for-counting-MEC-H-O-P1-P2}), $(M_1, M_2) \in T_2$. Thus, from \cref{def:one-to-one-function-for-counting-MEC-H-O-P1-P2}, $f_{O}(M) = (M_1, M_2)$.
 
\end{proof}
 We now prove $\leftarrow$ of \cref{lem:f-is-one-to-one}.
 
\begin{claim}
\label{claim:from-part-of-lem:f-is-one-to-one}
Let $O_1$ be a partial MEC of $G_1[X_1]$ and $O_2$ be a partial MEC of $G_2[X_2]$ such that $O$ is an extension of $(O_1, O_2)$. For each $M_1 \in \setofMECs{G_1, O_1}$ and $M_2 \in \setofMECs{G_2, O_2}$, there exists a unique $M \in \setofMECs{G, O}$ such that $f_{O}(M) = (M_1, M_2)$.
\end{claim}

\begin{proof}
 \Cref{lem:suff-cond-for-chordal-graph,def:one-to-one-function-for-counting-MEC-H-O-P1-P2} imply \cref{claim:from-part-of-lem:f-is-one-to-one}.
\end{proof}

\Cref{claim:to-part-of-lem:f-is-one-to-one,claim:from-part-of-lem:f-is-one-to-one} prove \cref{lem:f-is-one-to-one}.
\end{proof}

From \cref{lem:f-is-one-to-one}, $|T_1| = |T_2|$. Since $T_1 = \setofMECs{G, O}$, therefore, $|T_1| = | \setofMECs{G, O}|$. Since $T_2$ contains tuples $(M_1, M_2)$ such that for some tuple $(O_1, O_2) \in \setofpartialMECs{G_1[X_1]} \times \setofpartialMECs{G_2[X_2]}$,  and $O \in \mathcal{E}(O_1,  O_2)$, $M_1 \in \setofMECs{G_1, O_1}$ and $M_2 \in \setofMECs{G_2, O_2}$. This implies
\begin{equation}
    T_2 = \bigcup_{\substack{O_1 \in \setofpartialMECs{G_1[X_1]} \\ O_2 \in \setofpartialMECs{G_2[X_2]}\\O \in \mathcal{E}(O_1, O_2) }}{\setofMECs{G_1, O_1} \times \setofMECs{G_2, O_2}}.
\end{equation}
We show that sets $\setofMECs{G_1, O_1} \times \setofMECs{G_2, O_2}$ are disjoint. Suppose  sets $\setofMECs{G_1, O_1} \times \setofMECs{G_2, O_2}$ are not disjoint. Then there exists $(M_1, M_2)$ such that  $(M_1, M_2) \in  \setofMECs{G_1, O_1} \times \setofMECs{G_2, O_2}$ as well as $(M_1, M_2) \in  \setofMECs{G_1, O_1'} \times \setofMECs{G_2, O_2'}$. 
This implies $M_1 \in \setofMECs{G_1, O_1}$ as well as $M_1 \in \setofMECs{G_1, O_1'}$. But, from \cref{corr:MEC-has-a-unique-MEC-on-X}, image of $M$ on $X_1$ is unique, i.e., $O_1 =  O_1'$. Similarly, we can say that  $O_2 =  O_2'$. This shows $(O_1, O_2) = (O_1', O_2')$. This implies that the sets $\setofMECs{G_1, O_1} \times \setofMECs{G_2, O_2}$ are disjoint. This further implies
\begin{equation}
    |T_2| = \bigcup_{\substack{O_1 \in \setofpartialMECs{G_1[X_1]} \\ O_2 \in \setofpartialMECs{G_2[X_2]}\\O \in \mathcal{E}(O_1, O_2) }}{|\setofMECs{G_1, O_1}| \times |\setofMECs{G_2, O_2}|}.
\end{equation}
This proves \cref{lem:formula-to-compute-MEC-G-O}.
\end{proof} 

\begin{proof}[Proof of \cref{lem:computation-of-MEC-G-O'}]
    We start with the following:
    \begin{equation}
        \label{eq:lem:formula-to-compute-MEC-G-O'-1}
        \setofMECs{G, O'} = \bigcup_{\substack{O \in \setofpartialMECs{G[X]}\\ O' = O[X']}}{\setofMECs{G, O}}.
    \end{equation}
    We first show that the L.H.S. of \cref{eq:lem:formula-to-compute-MEC-G-O'-1} is a subset of the R.H.S. of \cref{eq:lem:formula-to-compute-MEC-G-O'-1}.
   \begin{lemma}
    \label{lem:if-M-with-its-projection-then-M-with-O}
    Let $G$ be a chordal graph, $X'\subseteq X \subseteq V_G$, and $O' \in \setofpartialMECs{G[X]}$. 
    Let $M \in \setofMECs{G, O'}$. Then, there exists a partial MEC $O \in \setofpartialMECs{G[X]}$ such that $O' = O[X]$ and $M \in \setofMECs{G, O}$. 
\end{lemma}
\begin{proof}
    Since $O'$ is an image of $M$, from \cref{def:image-of-an-MEC}, $M[X'] = O'$. Let $O = M[X]$. From \cref{def:image-of-an-MEC}, $O$ is an image of $M$, and $M \in \setofMECs{G, O}$. From \cref{corr:image-is-a-PMEC}, $O \in \setofpartialMECs{G[X]}$. Also, since $X'\subseteq X$, and both $O$ and $O'$ are induced subgraph of $M$, $O[X'] = M[X'] = O'$. This proves \cref{lem:if-M-with-its-projection-then-M-with-O}.
\end{proof}
We now show that the R.H.S. of \cref{eq:lem:formula-to-compute-MEC-G-O'-1} is a subset of the L.H.S. of \cref{eq:lem:formula-to-compute-MEC-G-O'-1}.
   \begin{lemma}
    \label{lem:if-M-with-O-then-M-with-its-projection}
    Let $G$ be a chordal graph, $X'\subseteq X \subseteq V_G$, $O \in \setofpartialMECs{G[X]}$, and $O' \in \setofpartialMECs{G[X']}$ such that $O' = O[X']$. 
    Then,  $M \in \setofMECs{G, O'}$. 
\end{lemma}
\begin{proof}
    Since $M \in \setofMECs{G, O}$, and $O = M[V_O] = M[X]$. Also, since $O' = O[X']$, and $X' \subseteq X$, $O' = M[X']$. This implies $O'$ is an image of $M$, i.e., $M \in \setofMECs{G, O'}$. This proves \cref{lem:if-M-with-O-then-M-with-its-projection}.
\end{proof}
\cref{lem:if-M-with-its-projection-then-M-with-O,lem:if-M-with-O-then-M-with-its-projection} imply \cref{eq:lem:formula-to-compute-MEC-G-O'-1}.

From \cref{corr:MEC-has-a-unique-MEC-on-X}, each MEC has a unique image on $X$. This implies the R.H.S. of \cref{eq:lem:formula-to-compute-MEC-G-O'-1} is a disjoint union of sets $\setofMECs{G, O}$ for $O\in \setofpartialMECs{G[X]}$. This further implies \cref{eq:lem:formula-to-compute-MEC-G-O'}. This completes the proof of \cref{lem:computation-of-MEC-G-O'}.
\end{proof}

\begin{proof}[Proof of \cref{lem:computation-of-MEC-G-O'-2}]
    From \cref{lem:computation-of-MEC-G-O'}, since $X' \subseteq X \subseteq V_G$, we have
    \begin{equation}
        \label{eq:corr:computation-of-MEC-G-O'-1}
        |\setofMECs{G, O'}| = \sum_{\substack{O \in \setofpartialMECs{G[X]}\\ O' = O[X']}}{|\setofMECs{G, O}| }.
    \end{equation}
    Now, from \cref{lem:formula-to-compute-MEC-G-O},
    \begin{equation}
       \label{eq:corr:computation-of-MEC-G-O'-2}
        |\setofMECs{G, O}| = \sum_{\substack{O_1 \in \setofpartialMECs{G_1[X_1]}\\ O_2 \in \setofpartialMECs{G_2[X_2]}\\O \in \mathcal{E}(O_1, O_2)}}{|\setofMECs{G_1, O_1}| \times |\setofMECs{G_2, O_2}|}.
    \end{equation}
    Combining \cref{eq:corr:computation-of-MEC-G-O'-1,eq:corr:computation-of-MEC-G-O'-2}, we get \cref{eq:lem:computation-of-MEC-G-O'-2}. This completes the proof of \cref{lem:computation-of-MEC-G-O'-2}.
\end{proof}

\begin{proof}[Proof of \cref{lem:proof-of-correctness-of-count-MEC-algorithm}]
    At line \ref{alg:counting-MEC-chordal-tree-decommposition-creating-function-f}, we construct the required function $f$, and for each partial MEC $O \in \setofpartialMECs{G[r_1\cup N(r_1, G)]}$, we initialized it (at lines \ref{alg:counting-MEC-chordal-tree-decommposition-foreach-1-start}-\ref{alg:counting-MEC-chordal-tree-decommposition-for-each-1-end}) with $f(O) = 0$. Lines \ref{alg:counting-MEC-chordal-tree-decommposition-if-degree-is-1}-\ref{alg:counting-MEC-chordal-tree-decommposition-if-end} deal with the base case, when the degree of $r_1$ in $T$ is zero.
    
    If the degree of $r_1$ in $T$ is zero then $V_G  = r_1$, and $G$ is a clique. For a clique graph $G$, there exists only one MEC of $G$ that is $G$ itself. 
    
    Suppose $M$ is an MEC of a clique graph $G$, and there exists a directed edge in $M$. Pick an edge $u\rightarrow v \in M$ such that $u$ does not have any incoming edge, and there does not exist any incoming edge to $v$ from the neighbors of $u$. Since $M$ is a chain graph (from \cref{item-1-theorem-nec-suf-cond-for-MEC} of \cref{thm:nes-and-suf-cond-for-chordal-graph-to-be-an-MEC}), we must get such $u\rightarrow v$. But, then $u\rightarrow v$ is not strongly protected. Since $M$ is an MEC of a clique graph $G$, $u\rightarrow v$ cannot be part of any induced subgraph as shown  \cref{fig:strongly-protected-edge}.a, \cref{fig:strongly-protected-edge}.b, and \cref{fig:strongly-protected-edge}.d. And, $u\rightarrow v$ cannot be part of any induced subgraph as shown  \cref{fig:strongly-protected-edge}.c because of the choice of $u$ and $v$. This implies $M$ cannot have any directed edge. This further implies each edge of $M$ is undirected, i.e., $M$ is $G$ itself (since $G$ is chordal, $G$ is an MEC as it obeys \cref{item-1-theorem-nec-suf-cond-for-MEC,item-2-theorem-nec-suf-cond-for-MEC,item-3-theorem-nec-suf-cond-for-MEC,item-4-theorem-nec-suf-cond-for-MEC} \cref{thm:nes-and-suf-cond-for-tree-graph-to-be-an-MEC}), and $|\setofMECs{G}| = 1$. This also implies that the induced subgraph of the only MEC $M$ of $G$ on $V_{G}$ is $G$, i.e.,  $|\setofMECs{G, G}| = 1$. This is why we update $f$ at line \ref{alg:counting-MEC-chordal-tree-decommposition:inside-if-case} with $f(G) = 1$. This implies that when the degree of $r_1$ is zero,  \cref{alg:counting-MEC-chordal-tree-decommposition} outputs the required function.

    If the degree of $r_1$ is not zero then \cref{alg:counting-MEC-chordal-tree-decommposition} implements \cref{lem:computation-of-MEC-G-O'-2} at lines \ref{alg:counting-MEC-chordal-tree-decommposition-pick-edge-r1-r2} - \ref{alg:counting-MEC-chordal-tree-decommposition-final-return-statement}. \Cref{lem:computation-of-MEC-G-O'-2} validates the function $f$ that is returned at line \ref{alg:counting-MEC-chordal-tree-decommposition-final-return-statement}. This completes the proof of \cref{lem:proof-of-correctness-of-count-MEC-algorithm}. 
\end{proof}

\begin{proof}[Proof of \cref{lem:proof-of-correctness-of-count-MEC-algorithm-2}]
    At line \ref{alg:counting-MEC-chordal-construct-clique-tree}, the algorithm constructs a clique tree $T$ of $G$. Then, it picks an arbitrary node $r_1$ of $T$ (line \ref{alg:counting-MEC-chordal-pick-r_1}). 
    At line \ref{alg:counting-MEC-chordal-f-returns}, it calls \cref{alg:counting-MEC-chordal-construct-clique-tree} that returns a function $f: \setofpartialMECs{G[r_1\cup N(r_1, G_1)]}\rightarrow \mathbb{Z}$ such that $f(O) = |\setofMECs{G, O}|$. Lines \ref{alg:counting-MEC-chordal-for-each-start}-\ref{alg:counting-MEC-chordal-for-each-end} implement \cref{eq:summation-of-PMECs} to compute $|\setofMECs{G}|$. Line \ref{alg:counting-MEC-chordal-return-statement} returns the required result. This completes the proof of \cref{lem:proof-of-correctness-of-count-MEC-algorithm-2}.
\end{proof}

\begin{algorithm}
\caption{Is\textunderscore Extension}
\label{alg:is-extension}
\SetAlgoLined
\SetKwInOut{KwIn}{Input}
\SetKwInOut{KwOut}{Output}
\SetKwFunction{Is_extension}{Is_extension}
\KwIn{Three partial MECs $O$, $O_1$, and $O_2$ such that
$G$ is a chordal graph, and $G_1$ and $G_2$ are two induced subgraphs of $G$ such that $G = G_1\cup G_2$, and $I = V_{G_1}\cap V_{G_2}$ is a vertex separator of $G$ that separates $V_{G_1}\setminus I$ and $V_{G_2}\setminus I$. $r_1$ and $r_2$ are cliques of, respectively, $G_1$ and $G_2$ such that $r_1\cap r_2 = I$.  $O \in \setofpartialMECs{G[r_1 \cup r_2 \cup N(r_1\cup r_2, G)]}$, $O_1 \in \setofpartialMECs{G_1[r_1 \cup N(r_1, G_1)]}$, and $O_2 \in \setofpartialMECs{G_2[r_2 \cup N(r_2, G_2)]}$.
}
    \KwOut{ 1 : if $O \in \mathcal{E}(O_1, O_2)$,\\
     \hspace{4pt}0 : otherwise.
    }

    \ForEach{$a\in \{1,2\}$\label{alg:is-extension:first-for-each-start}}
    {
        \If{$\exists u\rightarrow v \in O_a$ such that $u\rightarrow v\notin O$}
        {
            \KwRet 0.
        }
    }\label{alg:is-extension:first-for-each-end}
    
    \ForEach{$a\in \{1,2\}$\label{alg:is-extension:second-for-each-start}}
    {
        \ForEach{$u,v,w \in V_{O_a}$\label{alg:is-extension:third-for-each-start}}
        {
\If{$u\rightarrow v\leftarrow w$ is a v-structure in $O$ but not a v-structure in $O_a$}
            {
                \KwRet 0.
            }
            
        }\label{alg:is-extension:third-for-each-end}
    }\label{alg:is-extension:second-for-each-end}

    \ForEach{$a\in \{1,2\}$\label{alg:is-extension:forth-for-each-start}}
    {
        \ForEach{$u-v \in O_a$\label{alg:is-extension:fifth-for-each-start}}
        {
            \If{$u\rightarrow v \in O$\label{alg:is-extension:first-if-start}}
            {
                \If{$u\rightarrow v$ is not strongly protected in $O$ \label{alg:is-extension:u-v-is-sp}}
                {
                    \KwRet 0.
                }
            }\label{alg:is-extension:first-if-end}

        }\label{alg:is-extension:fifth-for-each-end}
    }\label{alg:is-extension:forth-for-each-end}

    \KwRet 1.
\end{algorithm} 
\Cref{alg:is-extension} is constructed to verify whether for any partial MECs $O$, $O_1$, and $O_2$, $O$ is an extension (\cref{def:extension}) of $O_1$ and $O_2$ or not. \Cref{lem:validation-of-alg-is-extension} validates \cref{alg:is-extension}.

\begin{lemma}
    \label{lem:validation-of-alg-is-extension}
    Let $G$ be a chordal graph, and $G_1$ and $G_2$ be two induced subgraphs of $G$ such that $G = G_1\cup G_2$, and $I = V_{G_1}\cap V_{G_2}$ be a vertex separator of $G$ that separates $V_{G_1}\setminus I$ and $V_{G_2}\setminus I$. Let $r_1$ and $r_2$ be cliques of, respectively, $G_1$ and $G_2$ such that $r_1\cap r_2 = I$.  $O \in \setofpartialMECs{G[r_1 \cup r_2 \cup N(r_1\cup r_2, G)]}$, $O_1 \in \setofpartialMECs{G_1[r_1 \cup N(r_1, G_1)]}$, and $O_2 \in \setofpartialMECs{G_2[r_2 \cup N(r_2, G_2)]}$. \Cref{alg:is-extension} outputs 1 if $O \in \mathcal{E}(O_1, O_2)$, otherwise, 0.    
\end{lemma}
\begin{proof}
    Lines \ref{alg:is-extension:first-for-each-start}-\ref{alg:is-extension:first-for-each-end} verifies \cref{item-1-of-def:extensions} of \cref{def:extension}. 
    Since from the construction, for each $a\in \{1,2\}$, $\skel{O_a} = \skel{O[V_{O_a}]}$, \cref{item-1-of-def:extensions} of \cref{def:extension} implies that $\mathcal{V}(O[V_{O_a}]) \subseteq \mathcal{V}(O_a)$.
    Lines \ref{alg:is-extension:second-for-each-start}-\ref{alg:is-extension:second-for-each-end} verifies $\mathcal{V}(O_a) \subseteq \mathcal{V}(O[V_{O_a}])$. In this way,  \cref{item-2-of-def:extension} of \cref{def:extension} is verified.
    Lines \ref{alg:is-extension:forth-for-each-start}-\ref{alg:is-extension:forth-for-each-end} verifies \cref{item-3-of-def:extension} of \cref{def:extension}. If \cref{item-1-of-def:extensions,item-2-of-def:extension,item-3-of-def:extension} of \cref{def:extension} fails then the algorithm returns 0, otherwise, it returns 1. This validates \cref{alg:is-extension}.
\end{proof}

\subsection{Omitted proofs of \cref{sec:time-complexity}}
\begin{proof}[Proof of \cref{lem:time-complexity-of-tree}]
    At line \ref{alg:counting-MEC-of-tree:cut-edge-r1-r2}, \Cref{alg:counting-MEC-of-tree} cuts an edge of $G$. This creates two induced subgraphs $G_1$ and $G_2$ of $G$.  At lines \ref{alg:counting-MEC-of-tree:count-MEC-call-for-G1}-\ref{alg:counting-MEC-of-tree:count-MEC-call-for-G2}, the algorithm  recursively call itself for $G_1$ and $G_2$. Since the number of edges in $G$ is $n-1$, the number of times \cref{alg:counting-MEC-of-tree} recursively calls itself is $n-1$. 
    Computation of $N$, $b$, and $c_0$ (at lines \ref{alg:counting-MEC-of-tree:N-intro}-\ref{alg:counting-MEC-of-tree:c0-intro}) can be done in $O(d)$ time.
    The while loop at lines \ref{alg:counting-MEC-of-tree:while-start}-\ref{alg:counting-MEC-of-tree:while-end} runs for $O(d)$ time. And, the computation of $c_i$ at line \ref{alg:counting-MEC-of-tree:ci-intro} takes $O(d)$ time. This implies the running time of lines \ref{alg:counting-MEC-of-tree:while-start}-\ref{alg:counting-MEC-of-tree:while-end} is $O(d^2)$ time. Thus, the overall running time of \cref{alg:counting-MEC-of-tree} is $O(d^2n)$.
\end{proof}

\begin{proof}[Proof of \cref{lem:time-complexity-of-chordal-graph-tree-decomposition}]
    At line \ref{alg:counting-MEC-chordal-tree-decommposition-cut-edge} of \cref{alg:counting-MEC-chordal-tree-decommposition}, we cut an edge of $T$ to recursively call \cref{alg:counting-MEC-chordal-tree-decommposition} (at lines \ref{alg:counting-MEC-chordal-tree-decommposition-f1-introduction} and \ref{alg:counting-MEC-chordal-tree-decommposition-f2-introduction}). The number of edges in a clique tree of a chordal graph $G$ is $O(n)$. This implies the number of times \cref{alg:counting-MEC-chordal-tree-decommposition} is called recursively is $O(n)$. If the treewidth of $G$ is $k$ then the size of $r_1$ is $k+1$. If the degree of $G$ is $d$ then the number of nodes in $r_1\cup N(r_1, G)$ is at most $(d+1)(k+1)$, i.e., $O(dk)$. And, the number of edges in $G[r_1\cup N(r_1, G)]$ is $O(d^2k^2)$. Since in a partial MEC of $G[r_1\cup N(r_1, G)]$, each edge $u-v$ of $G[r_1\cup N(r_1, G)]$ is either $u-v$ or $u\rightarrow v$, or $v\rightarrow u$. This implies the number of partial MECs of $G[r_1\cup N(r_1, G)]$ is at most $3^{d^2k^2}$, i.e., $O(2^{O(d^2k^2)})$. Therefore, the running time of initialization of $f$ (lines \ref{alg:counting-MEC-chordal-tree-decommposition-creating-function-f}-\ref{alg:counting-MEC-chordal-tree-decommposition-for-each-1-end}) is $O(2^{O(d^2k^2)})$. 
    Running time of lines \ref{alg:counting-MEC-chordal-tree-decommposition-if-degree-is-1}-\ref{alg:counting-MEC-chordal-tree-decommposition-if-end} takes $O(1)$ time.
    The running time of lines \ref{alg:counting-MEC-chordal-tree-decommposition-pick-edge-r1-r2}--\ref{alg:counting-MEC-chordal-tree-decommposition-G2-definition} is $O(|V_G| + |E_G|)$.
    Similar to the partial MECs of $G[r_1\cup N(r_1, G)]$, the number of partial MECs of $G_1[r_1\cup N(r_1, G_1)]$, $G_2[r_2\cup N(r_2, G_2)]$, and $G[r_1\cup r_2\cup N(r_1\cup r_2, G)]$ is $O(2^{O(d^2k^2)})$. Therefore, each of the for loops at lines \ref{alg:counting-MEC-chordal-tree-decommposition-for-each-2-start}, \ref{alg:counting-MEC-chordal-tree-decommposition-for-each-3-start}, and \ref{alg:counting-MEC-chordal-tree-decommposition-for-each-4-start}  runs for $O(2^{O(d^2k^2)})$ time. Since the number of edges in $O, O_1$ and $O_2$ are $O(d^2k^2)$, from \cref{lem:time-complexity-of-is-extension}, we can verify whether $O \in \mathcal{E}(O_1, O_2)$ (at line \ref{alg:counting-MEC-chordal-tree-decommposition-if-2-start}) in $O(d^2k^2)$ time.
    Construction of $O'$ at line \ref{alg:counting-MEC-chordal-tree-decommposition-O'-introduction}, can be done in  $O(d^2k^2)$ time. And, updation of $f$ at line \ref{alg:counting-MEC-chordal-tree-decommposition-f-O'-update} runs for $O(1)$ time. This implies the overall run time of \cref{alg:counting-MEC-chordal-tree-decommposition} is $O(n(2^{O(d^2k^2)} + n^2))$. This completes the proof of \cref{lem:time-complexity-of-chordal-graph-tree-decomposition}.
\end{proof}

\begin{proof}[Proof of \cref{lem:time-complexity-of-chordal-graph}]
    We can compute a clique tree representation of a chordal graph $G$ (line \ref{alg:counting-MEC-chordal-construct-clique-tree} of \cref{alg:counting-MEC-chordal}) in $O(|V_G| + |E_G|)$ time \cite{blair1993introduction}. From \cref{lem:time-complexity-of-chordal-graph-tree-decomposition}, the running time of line \ref{alg:counting-MEC-chordal-f-returns} is $O(n(2^{O(d^2k^2)} + n^2))$. As discussed in the proof of \cref{lem:time-complexity-of-chordal-graph-tree-decomposition}, the number of partial MECs in $G[r_1\cup N(r_1, G)]$ is $O(2^{O(d^2k^2)})$. This implies the running time of lines \ref{alg:counting-MEC-chordal-for-each-start}-\ref{alg:counting-MEC-chordal-for-each-end} if $O(2^{O(d^2k^2)})$. This implies the overall running time of \cref{alg:counting-MEC-chordal} is $O(n(2^{O(d^2k^2)} + n^2))$. This completes the proof of \cref{lem:time-complexity-of-chordal-graph}.
\end{proof}

\begin{proof}[Proof of \cref{thm:counting-MECs-for-tree-graph}]
    We can pick a root node $r_1$ of $G$ in $O(1)$ time. Then, we call \cref{alg:counting-MEC-of-tree} for input $G$ and $r_1$. From \cref{lem:time-complexity-of-tree}, the running time of the algorithm is $O(d^2 n)$. The algorithm returns $n_1(G, r_1)$, $n_0^0(G, r_1)$, $n_0^1(G, r_1)$, \ldots, $n_0^{d}(G, r_1)$, where $d$ is the degree of $r_1$ in $G$. Then, using \cref{eq:counting-MEC-of-tree-1,eq:counting-MEC-of-tree-2}, we can compute $|\setofMECs{G}|$. Therefore, in $O(d^2n)$ time, we can count the MECs of $G$.
\end{proof}

\begin{proof}[Proof of \cref{thm:counting-MECs-for-chordal-graph}]
    \Cref{lem:proof-of-correctness-of-count-MEC-algorithm-2,lem:time-complexity-of-chordal-graph} validates the theorem.
\end{proof}

\begin{lemma}
    \label{lem:time-complexity-of-is-extension}
    The time complexity of \cref{alg:is-extension} is $O(d^3k^3)$.
\end{lemma}
\begin{proof}
As discussed in the proof of \cref{alg:counting-MEC-chordal-tree-decommposition}, the number of nodes in $O, O_1$, and $O_2$ is $O(dk)$, and the number of edges in them is $O(d^2k^2)$.
    This implies the running time of lines \ref{alg:is-extension:first-for-each-start}-\ref{alg:is-extension:first-for-each-end} is $O(d^2k^2)$.
    The running time of lines \ref{alg:is-extension:second-for-each-start}-\ref{alg:is-extension:second-for-each-end} is $O(d^3k^3)$. Verifying each directed edge to determine whether it is strongly protected or not cumulatively takes $O(|V_O| + |E_O|)$ time.
    This implies the running time of lines \ref{alg:is-extension:third-for-each-start}-\ref{alg:is-extension:third-for-each-end} is $O(d^2k^2)$. Therefore, the overall running time of \cref{alg:is-extension} is $O(d^3k^3)$. This completes the proof of \cref{lem:time-complexity-of-is-extension}.
\end{proof} \end{document}